\documentclass[
	final,
	paper=a4,
	pagesize=auto,
	fontsize=11pt,
	version=last,
]{scrartcl}

\usepackage[utf8]{inputenc}
\usepackage[
	english,
]{babel}

\KOMAoptions{%
   DIV=10,
}%
\KOMAoptions{%
   numbers=noenddot
}%
\KOMAoptions{
   twoside=semi,  
   twocolumn=false, 
   headinclude=false,%
   footinclude=false,%
   mpinclude=false,%
   headlines=2.1,%
   headsepline=true,%
   footsepline=false,%
   cleardoublepage=empty 
}%
\KOMAoptions{%
   parskip=relative, 
}%
\KOMAoptions{%
}%
\KOMAoptions{%
}%
\KOMAoptions{%
   bibliography=totoc 
}%
\KOMAoptions{%
}%
\KOMAoptions{%
}%
\KOMAoptions{%
}%

\makeatletter
\newcommand{\LoadPackagesNow}{}
\newcommand{\LoadPackageLater}[2][]{%
   \g@addto@macro{\LoadPackagesNow}{%
      \usepackage[#1]{#2}%
   }%
}
\makeatother

\usepackage{xspace}
\usepackage{xargs}
\usepackage{ifthen}
\usepackage{etoolbox}

\usepackage[
	margin=3cm,
	]{geometry}
\usepackage{framed}
\usepackage{mdframed}
\usepackage{pbox}
\usepackage{enumitem}
\usepackage{array}
\usepackage{multirow}
\usepackage{longtable}
\usepackage{hhline}
\usepackage{afterpage}
\usepackage{pdflscape}
\usepackage{rotating}
\usepackage{booktabs}
\usepackage{setspace}
\usepackage[
   bottom,      
   stable,      
   perpage,     
   multiple,    
   flushmargin,
   hang,
]{footmisc}

\usepackage{mathpazo}

\usepackage[%
   headsepline,
   automark,
   komastyle,
]{scrpage2}

\clearscrheadings
\clearscrplain
\lehead{\pagemark}
\rohead{\pagemark}
\rehead{\headmark}
\lohead{\shortauthor}
\automark[section]{section} 

\usepackage{relsize}
\usepackage[normalem]{ulem}      
\usepackage{soul}		           
\usepackage{url}
\usepackage{lipsum}

\usepackage[table]{xcolor}
\usepackage[%
]{graphicx}

\usepackage{epstopdf}
\usepackage{wrapfig}
\usepackage{caption}
\usepackage{subcaption}
\usepackage[
   tbtags,    
   sumlimits,  
   nointlimits, 
   namelimits, 
   reqno,     
]{amsmath} %

\usepackage{amsfonts}
\usepackage{mathrsfs} 
\usepackage{dsfont}
\usepackage{amssymb}
\usepackage{units}
\LoadPackageLater{amsthm}
\LoadPackageLater{thmtools}
\usepackage[fixamsmath,disallowspaces]{mathtools}
\mathtoolsset{showonlyrefs}
\mathtoolsset{centercolon=true}
\usepackage{bm} 
\makeatletter
\g@addto@macro\bfseries{\boldmath}
\makeatother
\allowdisplaybreaks[4]
\numberwithin{equation}{section}

\usepackage[boxruled,slide,longend,vlined,algosection]{algorithm2e}
\IncMargin{.2em}
\SetEndCharOfAlgoLine{}
\SetKwHangingKw{Input}{Input:}
\SetKwHangingKw{Output}{Output:}

\usepackage[toc]{appendix}

\usepackage{csquotes}
\usepackage[
	bibstyle=alphabetic,
	citestyle=alphabetic,
	sorting=ynt,
	sortcites=true,
	giveninits=true,
	maxbibnames=99,
	backend=biber,
	maxalphanames=5,
]{biblatex}

\renewbibmacro{in:}{}

\usepackage[textsize=tiny,english,colorinlistoftodos,
	disable,
	]{todonotes}

\definecolor{pdfurlcolor}{rgb}{0,0,0.6}
\definecolor{pdffilecolor}{rgb}{0.7,0,0}
\definecolor{pdflinkcolor}{rgb}{0,0,0.6}
\definecolor{pdfcitecolor}{rgb}{0,0,0.6}
\usepackage[
  colorlinks=true,         
  urlcolor=pdfurlcolor,    
  filecolor=pdffilecolor,  
  linkcolor=pdflinkcolor,  
  citecolor=pdfcitecolor,  %
  raiselinks=true,			 
  breaklinks,              
  verbose,
  hyperindex=true,         
  linktocpage=true,        
  hyperfootnotes=false,     
  bookmarks=true,          
  bookmarksopenlevel=1,    
  bookmarksopen=true,      
  bookmarksnumbered=true,  
  bookmarkstype=toc,       
  plainpages=false,        
  pageanchor=true,         
  pdfdisplaydoctitle=true, 
  pdfstartview=FitH,       
  pdfpagemode=UseOutlines, 
  pdfpagelabels=true,           
  pdfpagelayout=OneColumn, 
]{hyperref}

\LoadPackagesNow

\newcommand{\ifargdef}[3][{}]{\ifthenelse{\equal{#2}{}}{#1}{#3}}

\setkomafont{section}{\Large\rmfamily\bfseries}
\setkomafont{subsection}{\large\rmfamily\bfseries}
\setkomafont{subsubsection}{\rmfamily\bfseries}
\setkomafont{paragraph}{\rmfamily\bfseries}
\setkomafont{pageheadfoot}{\smaller\scshape\rmfamily}

\hypersetup{
	pdfauthor={Martin Genzel, Peter Jung},
	pdftitle={Recovering Structured Data From Superimposed Non-Linear Measurements},
	pdfsubject={Article},
	pdfcreator={PDF-LaTeX},
}

\newenvironment{highlight}{\vspace{.5\baselineskip}\begin{addmargin}[2em]{2em}\itshape}{\end{addmargin}\vspace{.5\baselineskip}}

\newenvironment{properties}[2][2em]
{\begin{enumerate}[label={\textsc{(#2\arabic*)}},leftmargin=#1]}
{\end{enumerate}} 

\newenvironment{rmklist}
{\begin{enumerate}[label={(\arabic*)},itemindent=2em,leftmargin=0em]}
{\end{enumerate}}

\newenvironment{thmlist}
{\begin{enumerate}[label={(\arabic*)}]}
{\end{enumerate}}

\newenvironment{deflist}
{\begin{enumerate}[label={(\arabic*)}]}
{\end{enumerate}}

\renewcommand{\qedsymbol}{$_\blacksquare$}

\providecommand{\qedhere}{\hfill\qedsymbol}

\newtheoremstyle{claim}
	{\topsep}{\topsep}%
	{\itshape}
	{}
	{}
	{}
	{.5em}
	{{\bfseries\boldmath\thmname{#1} \thmnumber{#2}} \thmnote{(#3)}}

\newtheoremstyle{definition}
	{\topsep}{\topsep}%
	{}
	{}
	{}
	{}
	{.5em}
	{\textbf{\thmname{#1} \thmnumber{#2}} \thmnote{(#3)}}
	
\newtheoremstyle{algorithm}
	{\topsep}{\topsep}%
	{}
	{}
	{\bfseries\boldmath}
	{}
	{.5em}
	{\thmname{#1} \thmnumber{#2} \thmnote{(#3)}}

\declaretheorem[style=claim,numberwithin=section]{theorem}
\declaretheorem[style=claim,sibling=theorem]{lemma}

\declaretheorem[style=claim,sibling=theorem]{proposition}
\declaretheorem[style=definition,sibling=theorem]{definition}

\declaretheorem[style=definition,sibling=theorem]{model}
\declaretheorem[style=definition,sibling=theorem,qed=$\Diamond$]{remark}
\declaretheorem[style=definition,sibling=theorem,qed=$\Diamond$]{example}

\newcommand{\opleft}[1]{\mathopen{}\left#1}
\newcommand{\opright}[1]{\right#1\mathclose{}}
\newcommandx{\braces}[4]{%
\ifstrequal{#3}{normal}{#1#4#2}{%
\ifstrequal{#3}{auto}{\left#1#4\right#2}{%
\ifstrequal{#3}{opauto}{\opleft#1#4\opright#2}{%
#3#1#4#3#2}}}%
}
\newcommandx{\opannot}[3][3=\downarrow]{\stackrel{\mathclap{\substack{#1 \\ #3 \vspace{2pt}}}}{#2}}
\newcommandx{\lineannot}[3][3=\rightarrow]{\mathllap{\boxed{\text{\textsmaller{#1}}} #3} #2}
\newcommandx{\multilineannot}[4][4=\rightarrow]{\mathllap{\boxed{\parbox{#1}{\RaggedRight\textsmaller{#2}}} #4} #3}
 
\newcommand{\N}{\mathbb{N}} 
\newcommand{\R}{\mathbb{R}} 

\newcommand{\suchthat}[1][normal]{\ifstrequal{#1}{normal}{\mid}{#1|}} 

\newcommand{\cardinality}[1]{\abs{#1}} 
\newcommand{\intersec}{\cap} 
\newcommand{\dist}[2]{\operatorname{dist}(#1, #2)} 

\newcommandx{\intvcl}[3][1=normal]{\braces{[}{]}{#1}{#2, #3}} 
\newcommandx{\intvop}[3][1=normal]{\braces{(}{)}{#1}{#2, #3}} 
\newcommandx{\intvclop}[3][1=normal]{\braces{[}{)}{#1}{#2, #3}} 
\newcommandx{\intvopcl}[3][1=normal]{\braces{(}{]}{#1}{#2, #3}} 
\DeclareMathOperator*{\argmin}{argmin} 
\DeclareMathOperator{\sign}{sign}
\newcommandx{\abs}[2][1=normal]{\braces{\lvert}{\rvert}{#1}{#2}} 
\newcommandx{\ceil}[2][1=normal]{\braces{\lceil}{\rceil}{#1}{#2}} 
\newcommandx{\floor}[2][1=normal]{\braces{\lfloor}{\rfloor}{#1}{#2}} 
\newcommandx{\round}[2][1=normal]{\braces{[}{]}{#1}{#2}} 
\newcommandx{\der}[1]{D^{#1}} 
\newcommandx{\gradient}{\nabla} 
\newcommandx{\partder}[4][1={},4={}]{\frac{\partial^{#4} #2}{\partial #3^{#4}}\ifargdef{#1}{\Big|_{#1}}} 
\newcommandx{\integ}[4][1={},2={}]{\int_{#1}^{#2} #3 \, #4} 
\newcommandx{\asympffaster}[2][1=normal]{o\braces{(}{)}{#1}{#2}} 
\newcommandx{\asympfaster}[2][1=normal]{O\braces{(}{)}{#1}{#2}} 
\newcommandx{\asympeq}[2][1=normal]{\Theta\braces{(}{)}{#1}{#2}} 
\newcommandx{\asympsslower}[2][1=normal]{\omega\braces{(}{)}{#1}{#2}} 
\newcommandx{\asympslower}[2][1=normal]{\Omega\braces{(}{)}{#1}{#2}} 

\DeclareMathOperator{\Id}{Id} 
\newcommandx{\norm}[2][1=normal]{\braces{\|}{\|}{#1}{#2}} 
\renewcommandx{\sp}[3][1=normal]{\braces{\langle}{\rangle}{#1}{#2, #3}} 
\newcommandx{\End}[2][2={}]{\mathcal{L}\opleft( #1 \ifargdef{#2}{, #2} \opright)} 
\newcommand{\orthcompl}[1]{{#1}^\perp} 
\DeclareMathOperator{\spann}{\operatorname{span}} 
\newcommand{\T}{\mathsf{T}} 
\renewcommand{\vec}[1]{\boldsymbol{#1}} 

\newcommandx{\measure}[2][1=normal]{\operatorname{vol}\braces{(}{)}{#1}{#2}} 
\DeclareMathOperator{\supp}{supp} 
\newcommandx{\Leb}[3][1={},3=normal]{L^{#2}\ifargdef{#1}{\braces{(}{)}{#3}{#1}}{}} 
\newcommandx{\Lebnorm}[4][1=normal,3={2},4={}]{\norm[#1]{#2}_{#3}} 
\renewcommandx{\l}[3][1={},3=normal]{\ell^{#2}\ifargdef{#1}{\braces{(}{)}{#3}{#1}}} 
\newcommandx{\lnorm}[4][1=normal,3={2},4={}]{\norm[#1]{#2}_{#3}} 
\newcommandx{\Smooth}[4][1={},3={},4=normal]{C_{#3}^{#2}\ifargdef{#1}{\braces{(}{)}{#4}{#1}}} 
\newcommandx{\Schwartz}[2][1={},2=normal]{\mathscr{S}\ifargdef{#1}{\braces{(}{)}{#2}{#1}}} 
\newcommandx{\Schwartzpoly}[2][1=normal]{\braces{\langle}{\rangle}{#1}{\abs[#1]{#2}} } 
\newcommandx{\Tempdistr}[2][1={},2=normal]{\mathscr{S}'\ifargdef{#1}{\braces{(}{)}{#2}{#1}}} 
\newcommandx{\distrinp}[3][1=normal]{\braces{\langle}{\rangle}{#1}{#2, #3}} 
\newcommandx{\ft}[3][1=default,2=auto]{
\ifstrequal{#1}{default}{\widehat{#3}}{
\ifstrequal{#1}{long}{{\braces{(}{)}{#2}{#3}}^{\wedge}}{}}} 
\newcommandx{\ift}[3][1=default,2=auto]{
\ifstrequal{#1}{default}{\check{#3}}{
\ifstrequal{#1}{long}{{\braces{(}{)}{#2}{#3}}^{\vee}}{}}} 

\renewcommand{\vec}[1]{\bm{#1}}

\newcommand{\hybW}{\vec{W}}
\newcommand{\hybw}{\vec{w}}
\newcommand{\U}{\vec{U}}
\newcommand{\uv}{\vec{u}}
\newcommand{\vv}{\vec{v}}
\renewcommand{\a}{\vec{a}}
\newcommand{\ahyb}{\tilde{\vec{a}}}
\newcommand{\Ahyb}{\tilde{\vec{A}}}
\newcommand{\adir}{\bar{\vec{a}}}
\newcommand{\alift}{\vec{A}}
\newcommand{\agen}{\vec{a}}

\newcommand{\x}{\vec{x}}
\newcommand{\xgen}{\vec{x}}
\newcommand{\X}{\vec{X}}
\newcommand{\grtr}{\vec{x}_0}
\newcommand{\grtrgen}{\x^\natural}

\newcommand{\hybgrtr}{\tilde{\vec{\x}}_0}
\newcommand{\suppx}{\mathcal{S}}
\newcommandx{\solu}[1][1={}]{\ifargdef[\vec{\hat{x}}]{#1}{\hat{#1}}}
\newcommand{\sset}{K}
\newcommand{\ssetalt}{L}
\newcommand{\h}{\vec{h}}
\newcommand{\w}{\vec{w}}
\newcommand{\W}{\vec{W}}

\newcommand{\hgen}{\vec{h}}
\newcommand{\dict}{\vec{D}}
\newcommand{\dictatom}{\vec{d}}
\newcommand{\proj}{\vec{P}}

\newcommand{\loss}{\mathcal{L}}

\newcommand{\score}{\mathcal{S}}
\newcommand{\lossemp}[1][{}]{\bar{\mathcal{L}}_{#1}}

\newcommand{\multiplterm}[2]{\mathcal{M}_{#1}(#2)}
\newcommand{\quadrterm}[2]{\mathcal{Q}_{#1}(#2)}
\newcommand{\fobs}{f}

\newcommand{\scalfac}{\mu}
\newcommand{\hybscalfac}{\tilde\mu}
\newcommand{\scalvec}{\vec{\mu}}
\newcommand{\hybscalvec}{\tilde{\vec{\mu}}}

\newcommand{\modeldev}{\sigma}

\newcommand{\modelcovar}{\rho}
\newcommand{\orthmodelbias}{\vec{\rho}}

\newcommand{\Err}{\operatorname{Err}}
\newcommand{\vnull}{\vec{0}}
\newcommand{\I}[1]{\vec{I}_{#1}}

\newcommandx{\prob}[2][1={},2=normal]{\mathbb{P}\ifargdef{#1}{\braces{[}{]}{#2}{#1}}}

\newcommandx{\mean}[2][1={},2=normal]{\mathbb{E}\ifargdef{#1}{\braces{[}{]}{#2}{#1}}}
\newcommandx{\var}[2][1={},2=normal]{\mathbb{V}\ifargdef{#1}{\braces{[}{]}{#2}{#1}}}
\newcommandx{\normsubg}[3][1=normal,3={2}]{\norm[#1]{#2}_{\psi_2}} 

\newcommand{\subgparam}{\kappa}

\newcommand{\distributed}{\sim}

\newcommand{\Normdistr}[2]{\mathcal{N}(#1, #2)}
\newcommand{\gaussian}{\vec{g}}
\newcommand{\gaussianuniv}{g}

\newcommand{\Covmatr}{\vec{\Sigma}}

\newcommand{\pospart}[1]{\left[#1\right]_+}

\newcommandx{\ball}[2][1={},2={}]{B_{#1}^{#2}}
\DeclareMathOperator{\convhull}{conv}
\renewcommand{\S}{S}

\newcommand{\meanwidth}[2][{}]{w_{#1}(#2)}

\newcommand{\cone}[2]{\mathcal{C}(#1,#2)}
\newcommandx{\opnorm}[3][1=normal,3={2}]{\norm[#1]{#2}_{\operatorname{op}}}

\graphicspath{{images/}}

\newcommand{\revision}[1]{#1}

\begin{document}

\renewcommand*{\thefootnote}{\fnsymbol{footnote}}
\pagestyle{scrheadings}

\begin{center}
	\bfseries\larger[3]{Recovering Structured Data From \\ Superimposed Non-Linear Measurements}
\end{center}

\vspace{1\baselineskip}

\begin{addmargin}[2em]{2em}
\begin{center}
\noindent{\normalsize\bfseries{Martin Genzel\footnote{Technische Universit\"at Berlin, Department of Mathematics, 10623 Berlin, Germany, E-Mail: \href{mailto:genzel@math.tu-berlin.de}{\texttt{genzel@math.tu-berlin.de}}} \qquad Peter Jung\footnote{Technische Universit\"at Berlin, Communications and Information Theory Group, 10587 Berlin, Germany, E-Mail: \href{mailto:peter.jung@tu-berlin.de}{\texttt{peter.jung@tu-berlin.de}}}}}
\end{center}

%

\vspace{\baselineskip}

\vspace{1\baselineskip}
{\smaller
\noindent\textbf{Abstract.} 
This work deals with the problem of distributed data acquisition under non-linear communication constraints.
More specifically, we consider a model setup where $M$ \emph{distributed nodes} take individual measurements of an unknown \emph{structured source vector} $\grtr \in \R^n$, communicating their readings simultaneously to a \emph{central receiver}.
Since this procedure involves collisions and is usually imperfect, the receiver measures a \emph{superposition of non-linearly distorted signals}. 
In a first step, we will show that an $s$-sparse vector $\grtr$ can be successfully recovered from $\asympfaster{s \cdot\log(2n/s)}$ of such superimposed measurements, using a traditional Lasso estimator that does not rely on any knowledge about the non-linear corruptions.
This \emph{direct} method however fails to work for several ``uncalibrated'' system configurations. These \emph{blind reconstruction tasks} can be easily handled with the \mbox{$\l{1,2}$-Group-Lasso}, but coming along with an increased sampling rate of $\asympfaster{s\cdot \max\{M, \log(2n/s) \}}$ observations --- in fact, the purpose of this \emph{lifting} strategy is to extend a certain class of \emph{bilinear inverse problems} to \emph{non-linear} acquisition.
Our two algorithmic approaches are a special instance of a more abstract framework which includes sub-Gaussian measurement designs as well as general (convex) structural constraints. These results are of independent interest for various recovery and learning tasks, as they apply to arbitrary non-linear observation models.
Finally, to illustrate the practical scope of our theoretical findings, an application to \emph{wireless sensor networks} is discussed, which actually serves as the prototypical example of our methodology.

\vspace{.5\baselineskip}
\noindent\textbf{Key words.}
Distributed and non-linearly distorted measurements, compressed sensing, structured and blind recovery, Group-Lasso, wireless sensor networks

\vspace{.5\baselineskip}
\noindent Parts of this work have been presented in \cite{GJ17:sampta} at the 12th International Conference on ``Sampling Theory and Applications'' (SampTA 2017) and at the conference on ``Signal Processing with Adaptive Sparse Structured Representations'' (SPARS 2017).

The current version of the manuscript has been accepted to \emph{IEEE Transactions in Information Theory}, Digital Object Identifier 10.1109/TIT.2019.2932426

\textcopyright\ 2019 IEEE. Personal use of this material is permitted. Permission from IEEE must be obtained for all other uses, in any current or future media, including reprinting/republishing this material for advertising or promotional purposes, creating new collective works, for resale or redistribution to servers or lists, or reuse of any copyrighted
component of this work in other works.

}
\end{addmargin}
\newcommand{\shortauthor}{Genzel \& Jung: Recovery From Superimposed Non-Linear Measurements}

\renewcommand*{\thefootnote}{\arabic{footnote}}
\setcounter{footnote}{0}

\thispagestyle{plain}


\section{Introduction}
\label{sec:intro}

Initiated by the seminal works in \emph{compressed sensing} \cite{Candes2005,Donoho2006a,candes:stablesignalrecovery}, the considerable research on inverse data mining during the last decade has fundamentally changed our viewpoint on how to exploit structure in reconstruction problems. 
A broad variety of analytical results has shown that the sampling and storage complexity of many recovery methods can be dramatically reduced if the structure of the unknown data is explicitly taken into account.
The key objective of these approaches is to retrieve an unknown \emph{source vector} $\grtr \in \R^n$ from a collection of \emph{linear and non-adaptive measurements}
\begin{equation} \label{eq:intro:linmeas}
	\{\sp{\a_i}{\grtr}\}_{i=1}^m,
\end{equation}
where the \emph{measurement vectors} $\{\a_i\}_{i=1}^m \subset \R^n$ are supposed to be known.
While this problem is generally ill-posed as long as $m < n$, reliable estimates of $\grtr$ are often possible with $m \ll n$ if $\grtr$ carries some additional structure.
Perhaps, the most popular structural constraint is \emph{sparsity}, i.e., only a very few entries of $\grtr$ are non-zero.
For this situation, there exist numerous convex and greedy recovery algorithms, which do enjoy both provable performance guarantees and efficient implementations.
Moreover, it has turned out that many available principles can be naturally extended to concepts beyond sparse vectors, e.g., group- and tree-sparsity,
low-rankness, compressibility, or atomic representations. This universality in fact makes compressed sensing highly attractive to many practical challenges and particularly explains its great success.

While a large portion of research still focuses on the traditional linear setup of \eqref{eq:intro:linmeas},
many real-world applications come along with further restrictions on the measurement process.
This has led to two important lines of research in compressed sensing which have received considerable attention in the last years: 
\begin{deflist}
\item
	\emph{Non-linear acquisition schemes.}
	The assumption of perfect linear measurements is quite restrictive for many sensing devices used in practice, even if additive noise is permitted.
	In fact, it might be more realistic to consider a \emph{single-index model} of the form
	\begin{equation}\label{eq:intro:singleindex}
		\{\fobs(\sp{\a_i}{\grtr})\}_{i=1}^m.
	\end{equation}
	For instance, an analog-to-digital conversion may lead to quantized measurements. In its most extreme case, this corresponds to \emph{$1$-bit compressed sensing} \cite{boufounos2008onebit}, where $\fobs(v) = \sign(v)$ returns binary outputs.
	Another very important example is (sparse) \emph{phase retrieval}, where just the magnitudes of the measurements are given, i.e., $\fobs(v) = \abs{v}$.
	Apart from that, the function $\fobs$ is oftentimes only partially known, which could be due to uncertainties in the hardware configuration.
\item
	\emph{Distributed observations.}
	In many realistic applications, one cannot expect that each individual measurement of \eqref{eq:intro:singleindex} can be accessed at any time and any precision. The additional costs of storing and communicating observations often forces engineers to make use of \emph{distributed data acquisition architectures}.
	Mathematically, we may assume that $M$ distributed \emph{sensing nodes} are taking individual (non-linearly distorted) measurements of $\grtr$ \emph{in parallel}, meanwhile contributing to the overall measurement process in a certain way:
	\begin{equation}\label{eq:intro:distrmeas}
		\Big\{F \Big( \fobs_1(\sp{\a_i^1}{\grtr}), \dots, \fobs_M(\sp{\a_i^M}{\grtr}) \Big) \Big\}_{i=1}^m.
	\end{equation}
	Here, the ``fusion function'' $F \colon \R^M \to \R$ specifies how each single node contributes to the $i$-th measurement step. Adapting the notation from above, $\fobs_j \colon \R \to \R$ and $\a_i^j \in \R^n$ represent the (possibly unknown) non-linearity and $i$-th measurement vector of the $j$-th node, respectively.
	
	In contrast, the communication of all observations to a central entity in a \emph{sequential} manner would require $m \cdot M$ measurement steps in total, particularly coming along with an additional scheduling and control overhead. Hence, the restrictions imposed by distributed architectures may drastically change our way of thinking about many sensing tasks.
\end{deflist}

One of the prototypical applications of such distributed observation schemes are \emph{wireless sensor networks}, which form a state-of-the-art approach to many types of environmental monitoring problems. In this setup, each of the $M$ nodes corresponds to an autonomous \emph{sensor unit}, acquiring
$m$ individual measurements of a (structured) source $\grtr\in \R^n$. For example, one could think of measurements of a temperature field at different locations whereby the
global spatial fluctuation is specified by the vector $\grtr$. 
All devices transmit simultaneously to a \emph{central receiver},
leading to \emph{additive} collisions. 
Since this process is imperfect, mainly caused
by low-quality sensors and the wireless channel, the receiver eventually
measures a \emph{superposition} of corrupted signals; see Figure~\ref{fig:sensornetwork} for an illustration (more details of this specific application are presented in Section~\ref{sec:practice}).

\begin{figure}
  \centering
  \includegraphics[width=.6\linewidth]{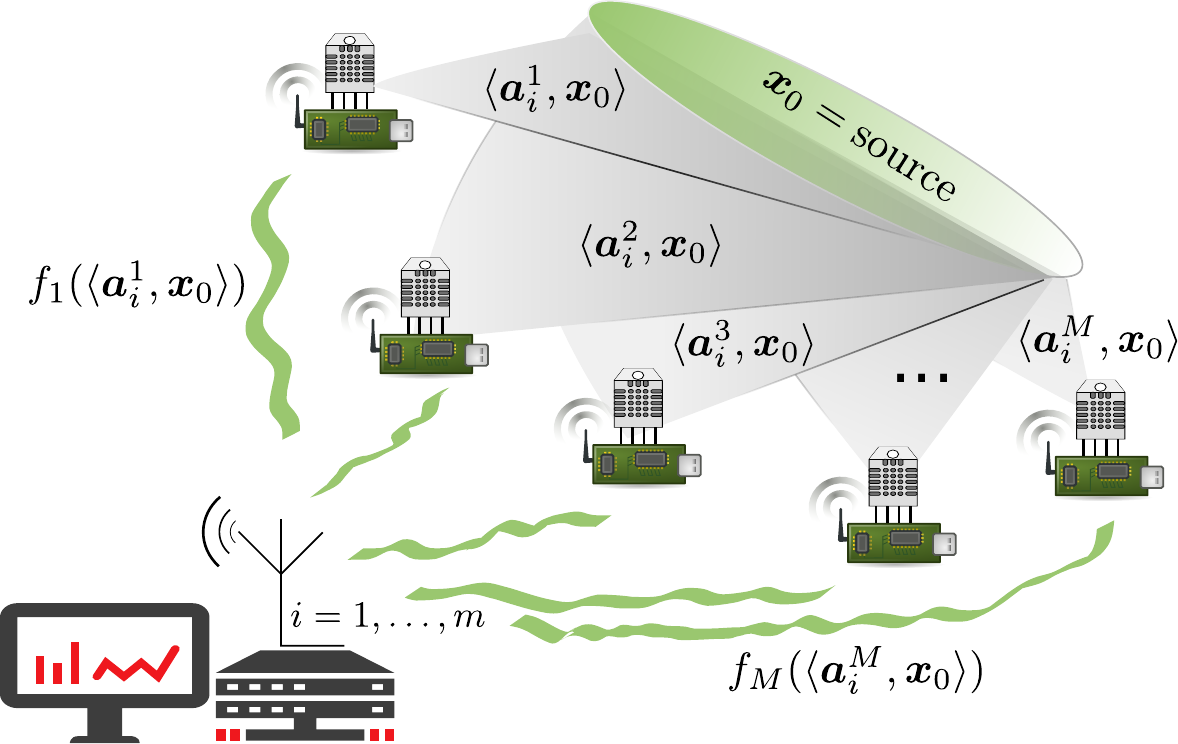}
  \caption{A schematic sensor network. Each wireless sensor node $j = 1,\dots, M$ acquires $i = 1,\dots,m$ individual linear measurements of a source vector $\grtr \in \R^n$ using different ``viewpoints'' $\a_i^j \in \R^n$. These measurements are simultaneously transmitted to a central receiver for recovery. Thereby, the sensor readings $\{\sp{\a_i^j}{\grtr}\}_{i=1}^m$ are affected by unknown non-linear distortions $\fobs_j\colon \R \to \R$, typically caused by hardware imperfections and the wireless channel. The fusion center finally receives the superposition (sum) of these autonomous measurements.
  }
  \label{fig:sensornetwork}
\end{figure}

The superposition principle of wireless sensor networks indicates that computing the sum of signals is a quite natural choice of the fusion function in \eqref{eq:intro:distrmeas}, i.e.,
\begin{equation}
	F(v_1, \dots, v_M) = v_1 + \dots + v_M.
\end{equation}
In fact, this important example motivates the distributed observation model that we focus on in this work, namely \emph{superimposed, non-linearly distorted measurements}:
\begin{equation} \label{eq:intro:measurements} 
	y_i = \sum_{j = 1}^M \fobs_j (\sp{\a_i^j}{\grtr}) + e_i, \qquad i = 1, \dots, m,
\end{equation}
where $e_i \in \R$ is additive \emph{noise}, which may also disturb the acquisition process.

Our ultimate goal is now to efficiently recover the (structured) source vector $\grtr \in \R^n$ from a given
measurement ensemble $\{(\{ \a_i^j \}_{j = 1}^M, y_i)\}_{i = 1}^m$. We
would like to emphasize that the exact behavior of the non-linear
distortions $\fobs_j$ is typically unknown, so that our reconstruction
methods will not require explicit knowledge of them.  These
challenges give rise to the following problem issues that we would like to address
in this paper:
\begin{properties}[3em]{Q}
\item\label{question:intro:samplecomplexity} How many (superimposed) measurements $m$
  are required for highly accurate and stable recovery?  In how far
  can we exploit the underlying structure of the source $\grtr$, such
  as sparsity?
\item\label{question:intro:nonlinear} What is the impact of the
  unknown non-linearities? Under which conditions on
  $\fobs_1, \dots, \fobs_M$ can we expect a similar
  performance as in the linear case?
\item\label{question:intro:nodecount} What role is played by the total node
  count $M$? Is it always beneficial to increase the node count or is there a trade-off which we have to take into account?
\end{properties}
We point out that these questions are highly relevant for designing wireless sensor networks, as the number of used nodes is often driven by a compromise between costs and quality.

\subsection{Algorithmic Approaches}
\label{subsec:intro:algo}

The basic idea behind our methodology is based on the classical approach of ``linearizing'' the underlying measurement process, thereby treating the model mismatch as noise.
\revision{In this context, the traditional Lasso \cite{tibshirani1996lasso} has recently proven very useful for high-dimensional estimation problems with non-linear observations \cite{plan2015lasso,thrampoulidis2015lasso}, though the underlying ideas date back to a much older work by Brillinger \cite{Brillinger1982}.}
As we will see next, such a strategy works out for the superimposed measurement scheme of \eqref{eq:intro:measurements}.
For the sake of clarity, the following two methods are based on $\l{1}$-constraints, tailored to the popular case of \emph{sparse} source models.
The generalization to more sophisticated structural assumptions will be elaborated in Section~\ref{sec:extensions}.

\subsubsection*{The Direct Method}

In the linear setting, i.e., $\fobs_j = \Id$, the model of \eqref{eq:intro:measurements} would degenerate to
$y_i = \sp{\adir_i}{\grtr} + e_i$ where $\adir_i \coloneqq \sum_{j = 1}^M \a_i^j$, $i = 1, \dots, m$, are \emph{superimposed measurement vectors}.
While Lasso-type estimators have proven to perform very robustly in such a simple situation, we may hope that this approach even succeeds in the non-linear case.
This naive idea of ``mimicking the linear counterpart'' is precisely what our first method is doing:

\vspace{.5\baselineskip}
\begin{algorithm}[H]
  \caption{Direct Method}
  \label{algo:direct}
  \Input{Measurement ensemble $\{(\{ \a_i^j \}_{j = 1}^M, y_i)\}_{i = 1}^m$, sparsity parameter $R > 0$}
  \Output{Estimated source vector $\solu \in \R^n$}
  \BlankLine
  \SetKwBlock{Compute}{Compute}{}
  \Compute{\BlankLine
    \nl Compute superimposed measurement vectors:
    \begin{equation}
    	\adir_i = \sum_{j = 1}^M \a_i^j, \quad i = 1, \dots, m.
    \end{equation}
    \BlankLine
    \nl Solve the Lasso and pick any minimizer:
    \begin{equation}
      \label{eq:direct:lasso}\tag{$P_R^\text{Dir}$}
      \solu = \argmin_{\x \in \R^n} \tfrac{1}{2m} \sum_{i = 1}^m (y_i - \sp{\adir_i}{\x})^2 \quad \text{s.t. $\lnorm{\x}[1] \leq R$.}
    \end{equation}
  }
\end{algorithm}
\vspace{.5\baselineskip}

The only tuning parameter that needs to be (adaptively) chosen is $R > 0$, controlling the level of sparsity of the minimizer.  
Remarkably, Algorithm~\ref{algo:direct} does neither explicitly
depend on the non-linearities $\fobs_j$ nor on the number of sensors
$M$.\footnote{Even though the non-linearities $\fobs_j$ could be known, incorporating them directly into \eqref{eq:direct:lasso} would typically lead to a challenging \emph{non-convex} problem.} 
The individual measurement vectors $\a_i^j$ do not have to be known to the optimization program \eqref{eq:direct:lasso}, implying that the overall computational costs of \eqref{eq:direct:lasso} will not increase as $M$ grows. 
From a practical perspective, the specific identities of the nodes are not relevant to this method, which has important consequences for network planning and maintenance issues (cf. Section~\ref{sec:practice}).

\subsubsection*{The Lifting Method}

Due to limited knowledge of the non-linear distortions in \eqref{eq:intro:measurements}, there are some important situations where \eqref{eq:direct:lasso} fails to work.
For example, if the wireless channel of a sensor network changes rapidly, the non-linearities $\fobs_j$ may involve unknown sign-changes $\sp{\a_i^j}{\grtr} \mapsto - \sp{\a_i^j}{\grtr}$. Such a scenario is problematic for Algorithm~\ref{algo:direct}, since a simple superposition $\adir_i = \sum_{j = 1}^M \a_i^j$ does not account for possible sign-flips.
As a way out, one may rather try fit each node of \eqref{eq:intro:measurements} \emph{individually}. This strategy leads to the following algorithm:

\vspace{.5\baselineskip}
\begin{algorithm}[H]
  \caption{Lifting Method}
  \label{algo:lifting}
  \Input{Measurement ensemble $\{(\{ \a_i^j \}_{j = 1}^M, y_i)\}_{i = 1}^m$, sparsity parameter $R > 0$}
  \Output{Estimated source vectors $\solu^1,\dots,\solu^M \in \R^n$}
  \BlankLine
  \SetKwBlock{Compute}{Compute}{}
  \Compute{\BlankLine Solve the Group-Lasso and pick any minimizer:	
    \begin{equation}\label{eq:lifting:l1l2lasso}\tag{$P_R^\text{Lift}$}
      [\solu^1\cdots\solu^M] = \argmin_{\substack{[\x^1\cdots \x^M] \\ \in \R^{n\times M}}} \ \tfrac{1}{2m} \sum_{i = 1}^m
      \Big(y_i - \sum_{j = 1}^M\sp{\a_i^j}{\x^j}\Big)^2 \quad
      \text{s.t. $\lnorm[\big]{[\x^1 \cdots \x^M]}[1,2] \leq R$.}\quad
    \end{equation}
    \BlankLine\revision{\emph{Optional.} Compute the dominating left singular vector $\solu \in \S^{n-1}$ of the matrix $\solu[\X] \coloneqq [\solu^1\cdots\solu^M] \in \R^{n \times M}$ to obtain an estimate of the normalized source vector $\grtr / \lnorm{\grtr}$, up to a possible sign-change.}
  }
\end{algorithm}
\vspace{.5\baselineskip}

The constraint of \eqref{eq:lifting:l1l2lasso} involves the
$\l{1,2}$-norm of the matrix $[\x^1 \cdots \x^M] \in \R^{n \times M}$ which is
defined by
\begin{equation}
  \lnorm[\big]{[\x^1 \cdots \x^M]}[1,2] \coloneqq \sum_{l = 1}^n \Big( \sum_{j = 1}^M \abs{x_l^j}^2 \Big)^{1/2},
\end{equation}
where $x_l^j$ denotes the $l$-th entry of $\x^j \in \R^n$.  The purpose
of this $\l{2}$-group constraint is to enforce a certain ``coupling''
between all vectors, since every node of \eqref{eq:intro:measurements}
is actually supposed to communicate the same source $\grtr$.  Compared to the direct
method, the optimization program of Algorithm~\ref{algo:lifting}
operates in the higher dimensional (``lifted'') matrix space of
$\R^{n\times M}$, which also explains the term ``lifting method.''
This feature comes along with additional computational burdens, but we
will see in Section~\ref{sec:results} that the flexibility of
\eqref{eq:lifting:l1l2lasso} in turn allows us to handle more difficult
situations than \eqref{eq:direct:lasso}.
\revision{Finally, we note that the optional step in Algorithm~\ref{algo:lifting} provides a common strategy to compute an explicit estimator of (the direction of) $\grtr$ from the solution of \eqref{eq:lifting:l1l2lasso}.}

\subsection{Main Contributions}
\label{subsec:intro:contrib}

One of the major concerns of this work is to analyze the proposed methods (Algorithm~\ref{algo:direct} and Algorithm~\ref{algo:lifting}) with respect to the problem questions stated in \ref{question:intro:samplecomplexity}--\ref{question:intro:nodecount}. 
In Section~\ref{sec:results}, we first consider the relatively simple situation of i.i.d.\ Gaussian measurement vectors and sparse source vectors, showing that recovery from superimposed, non-linearly distorted measurements \eqref{eq:intro:measurements} becomes indeed feasible. Very roughly speaking, it will turn out that an $s$-sparse vector $\grtr \in \R^n$ can be estimated with accuracy $\delta \in \intvopcl{0}{1}$ if
\begin{equation}
	m \gtrsim \delta^{-2} \cdot s \cdot \log(\tfrac{2n}{s}) \quad \text{using the direct method (cf. Theorem~\ref{thm:results:direct}),}
\end{equation}
and
\begin{equation}
	m \gtrsim \delta^{-2} \cdot s \cdot \max\{M, \log(\tfrac{2n}{s})\} \quad \text{using the lifting method (cf. Theorem~\ref{thm:results:lifting}).}
\end{equation}
These informal statements already give us a first answer to \ref{question:intro:samplecomplexity} and \ref{question:intro:nodecount}: The sample complexity of the direct method \eqref{eq:direct:lasso} does not depend on the node count at all, whereas the respective rate of the lifting method \eqref{eq:lifting:l1l2lasso} scales linearly in $M$. Hence, it appears that the direct approach is superior, but in fact, several important configurations of $\fobs_1, \dots, \fobs_M$ are (implicitly) excluded by Theorem~\ref{thm:results:direct} (cf. \ref{question:intro:nonlinear}).  On the other hand, the lifting method is able to handle most of these situations, with the price of taking more measurements as $M$ grows.

The proofs of the guarantees from Section~\ref{sec:results} are actually an application of a substantially more general framework that is developed in Section~\ref{sec:proofs}.
More specifically, we permit arbitrary non-linear observations, sub-Gaussian measurement vectors as well as general convex constraints (cf. Model~\ref{model:proofs:generalobs} and Theorem~\ref{thm:proofs:abstractrecovery}).
Since these results go far beyond the specific model of superimposed measurements, they may be of independent interest in various recovery and learning tasks.
The flexibility of our methodology particularly allows us to unify the two methods from Subsection~\ref{subsec:intro:algo} into a \emph{hybrid method} in Section~\ref{sec:extensions} (cf. Algorithm~\ref{algo:hybrid}), which is useful for incorporating \emph{prior knowledge} about the model configuration.

\subsection{Related Literature}

We have already pointed out in the introductory part that, in order to model many real-world applications (e.g., wireless sensor networks), it is substantial to take account of both \emph{distributed observation schemes} and \emph{non-linear distortions} at the same time.  However, most recent approaches from the literature do only focus on either one of these two problems.

The setting of single-index models \eqref{eq:intro:singleindex}\footnote{This is equivalent to \eqref{eq:intro:measurements} with $M = 1$.} --- as natural extension of classical linear compressed sensing --- has gained increasing attention within the past years. One branch of the recent literature has focused on the situation where the non-linearity is known in advance, so that its specific structure can be directly exploited, e.g., see \cite{yang2015sparse,mei2016landscape} and the references therein. Our problem setup unfortunately forbids such an assumption, implying that we have to treat any non-linear perturbation as an additional source of uncertainty. Therefore, this work is more closely related to a different line of research in which the estimator does not rely on knowledge of $\fobs$; see \cite{plan2014highdim,plan2015lasso,thrampoulidis2015lasso,genzel2016estimation,goldstein2016structured,oymak2016fast}.
The proofs of our main results are in fact based on different tools from empirical process theory, which particularly allows us to overcome several shortcomings of these findings, such as the limitation to Gaussian measurement vectors; see Remark~\ref{rmk:app:abstractrecovery:techniques} for a more detailed discussion.

Regarding distributed acquisition, the superimposed model of \eqref{eq:intro:measurements} generalizes the linear case in which the non-linearities $\fobs_j$ just correspond to rescaling:
\begin{equation} \label{eq:intro:bilinearmodel} 
	y_i = \sum_{j = 1}^M h_j \sp{\a_i^j}{\grtr} + e_i, \qquad i = 1, \dots, m,
\end{equation}
where $h_1 ,\dots, h_M \in \R$ are \emph{unknown} scalar factors. The task of recovering both $\grtr \in \R^n$ and $\h \coloneqq (h_1 ,\dots, h_M) \in \R^M$ is actually a \emph{bilinear inverse problem}. Such problems typically occur in blind reconstruction tasks where the linear system model is not precisely known, for instance, in non-coherent sporadic communication \cite{Jung2014}. Taking the lifting perspective, this can be regarded as recovery of the outer product $\grtr\h^\T \in \R^{n \times M}$ from linear observations. Since the factors of such a rank-one matrix may enjoy additional structure, e.g., if $\grtr$ is $s$-sparse, estimation via convex programming becomes a difficult challenge.
\revision{For example, Oymak et al.\ showed in \cite{Oymak2015} that minimizing fixed convex combinations of multiple regularizers (promoting both low-rankness and sparsity) leads to a sampling rate of $m = \asympfaster{\min\{n + M, s \cdot M\}}$, which is in fact suboptimal.}\footnote{\revision{We note that this rate is still slightly better than the one achieved by Theorem~\ref{thm:results:lifting}, which is due to the missing nuclear norm penalty in the lifting method \eqref{eq:lifting:l1l2lasso}. However, such an additional constraint could be easily incorporated by our extensions presented in Section~\ref{sec:extensions}, leading to the same sampling rate as in \cite{Oymak2015}.}} Under certain further restriction, it has turned out that iterative methods can break this bottleneck \cite{lee2016powerfac}, particularly in the situation of blind deconvolution \cite{ahmed2014blinddeconv,Lee2017}, which often serves as a prototypical example. 
However, the actual issue of this paper is even more tough, since the model of \eqref{eq:intro:measurements} includes non-linear distortions that are not covered by the traditional setup of \eqref{eq:intro:bilinearmodel}.

For these reasons, a key objective of this work is a verification that recovery from non-linear \emph{and} distributed observations is still achievable by simple Lasso estimators, whereby many known theoretical guarantees do naturally translate into this combined setting.

\subsection{Outline and Notation}

In Section~\ref{sec:results}, we will present our main recovery results for the direct (Theorem~\ref{thm:results:direct}) and lifting method (Theorem~\ref{thm:results:lifting}).
This also involves a precise definition of several model parameters that allow us to make qualitative and quantitative statements on the issues of \ref{question:intro:samplecomplexity}--\ref{question:intro:nodecount}.
All related proofs are postponed to Section~\ref{sec:proofs}.
Section~\ref{sec:extensions} then elaborates on various important extensions, such as sub-Gaussian measurement designs and arbitrary convex constraint sets, which however require some additional technical preliminaries.
Section~\ref{sec:practice} returns to the initial prototypical example of wireless sensor networks. In this course, we will study this specific application in greater detail and discuss the practical scope of our results.
Final remarks can be found in Section~\ref{sec:conclusion}, including potential improvements that could be investigated in future works.

Throughout this paper, we will use several (standard) notations and conventions, compiled by the following list:
\begin{itemize}
\item
	\emph{Generic constants.} The letter $C$ is always reserved for a constant.
	We refer to $C$ as a \emph{numerical} constant  if its value is independent from all present parameters.
	Note that the value of $C$ might change from time to time, while we still use the same letter.
	If an \mbox{(in-)equality} holds true up to a numerical constant $C$, we sometimes simply write $A \lesssim B$ instead of $A \leq C \cdot B$, and if $C_1 \cdot A \leq B \leq C_2 \cdot A$ for numerical constants $C_1, C_2 > 0$, the abbreviation $A \asymp B$ is used.
\item
	For an integer $d \in \N$, we set $[d] \coloneqq \{1, \dots, d\}$.
\item
	\emph{Vectors} and \emph{matrices} are denoted by lower- and uppercase boldface letters, respectively. Their entries are indicated by subscript indices and lowercase letters, e.g., $\x = (x_1, \dots, x_d) \in \R^d$ for a vector and $\vec{B} = [b_{kl}] \in \R^{d'\times d}$ for a matrix.
\item
	Let $\x = (x_1, \dots, x_d) \in \R^d$. The \emph{support} of $\x$ is defined by the set of its non-zero entries $\supp(\x) \coloneqq \{ k \in [d] \suchthat x_k \neq 0 \}$ and we set $\lnorm{\x}[0] \coloneqq \cardinality{\supp(\x)}$. For $1 \leq p \leq \infty$, the \emph{$\l{p}$-norm} is given by
	\begin{equation}
	\lnorm{\x}[p] \coloneqq \begin{cases} (\sum_{k = 1}^d \abs{x_k}^p)^{1/p}, & p < \infty, \\ \max_{k \in [d]} \abs{x_k}, & p = \infty.
	\end{cases}
	\end{equation}
	The associated \emph{unit ball} is denoted by $\ball[p][d] \coloneqq \{ \x \in \R^d \suchthat \lnorm{\x}[p] \leq 1 \}$ and the \emph{(Euclidean) unit sphere} is $S^{d-1} \coloneqq \{ \x \in \R^d \suchthat \lnorm{\x} = 1 \}$. 
	If $\vec{B}, \tilde{\vec{B}} \in \R^{d' \times d}$ are matrices, then $\lnorm{\vec{B}}$ and $\sp{\vec{B}}{\tilde{\vec{B}}}$ always refer to the Euclidean norm (\emph{Frobenius norm}) and scalar product (\emph{Hilbert-Schmidt inner product}), respectively. \revision{Moreover, we denote the \emph{spectral norm} of $\vec{B}$ by $\opnorm{\vec{B}}$.}
\item
	Let $\ssetalt \subset \R^d$ and $\x \in \ssetalt$. The \emph{cone of $\ssetalt$ at $\x$} is given by
	\begin{equation}\label{eq:intro:notation:cone}
		\cone{\ssetalt}{\x} \coloneqq \{ \lambda \h \suchthat \h \in \ssetalt - \x, \lambda \geq 0 \}.
	\end{equation}
	If $\ssetalt$ is convex, then $\cone{\ssetalt}{\x}$ is convex as well.
	For a linear subspace $U \subset \R^d$, we denote the \emph{orthogonal projection} onto $U$ by $\proj_{U}$.
	And for $\x \in \R^d \setminus \{ \vnull\}$, we just write $\proj_{\x} \coloneqq \proj_{\spann\{\x\}} = \sp{\cdot}{\tfrac{\x}{\lnorm{\x}}}\tfrac{\x}{\lnorm{\x}}$.
\item
	\emph{Sub-Gaussian random variables.} Let $a$ be a real-valued random variable. Then $a$ is \emph{sub-Gaussian} if 
	\begin{equation}\label{eq:intro:notation:normsubg}
		\normsubg{a} \coloneqq \sup_{p \geq 1} p^{-1/2} (\mean[\abs{a}^p])^{1/p} < \infty,
	\end{equation}
	and $\normsubg{\cdot}$ is called the \emph{sub-Gaussian} norm.
	Throughout this work, we will frequently apply an \emph{inequality of Hoeffding-type for sub-Gaussian variables} (cf. \cite[Lem.~5.9]{vershynin2012random}): If $a_1, \dots, a_N$ are independent mean-zero sub-Gaussian random variables, then $\sum_{k = 1}^N a_k$ is also sub-Gaussian with zero mean, and we have
	\begin{equation}\label{eq:intro:notation:hoeffding}
		\normsubg[\Big]{\sum_{k = 1}^N a_k}^2 \lesssim \sum_{k = 1}^N \normsubg{a_k}^2.
	\end{equation}
	
	Now, let $\a$ be a random vector in $\R^d$. Then, $\a$ is called \emph{isotropic} if $\mean[\a \a^\T] = \I{d}$ or equivalently
	\begin{equation}\label{eq:intro:notation:isotropicsp}
		\mean[\sp{\a}{\x} \sp{\a}{\x'}] = \sp{\x}{\x'} \quad \text{for all $\x, \x' \in \R^d$.}
	\end{equation}
	This particularly implies that
	\begin{equation}\label{eq:intro:notation:intropic}
		\mean[\sp{\a}{\x}^2] = \lnorm{\x}^2 \quad \text{for all $\x \in \R^d$.}
	\end{equation}
	Adapted from the scalar case, the \emph{sub-Gaussian norm} of $\a$ is given by
	\begin{equation}
		\normsubg{\a} \coloneqq \sup_{\x \in \S^{d-1}} \normsubg{\sp{\a}{\x}}.
	\end{equation}
	Finally, if $\a$ is a (mean-zero) \emph{Gaussian random vector} with covariance matrix $\Covmatr \in \R^{d \times d}$, we write $\a \distributed \Normdistr{\vnull}{\Covmatr}$.
\end{itemize}

\section{Main Results}
\label{sec:results}

In this section, we analyze the performance of our approaches in
Algorithm~\ref{algo:direct} and Algorithm~\ref{algo:lifting} for the
situation of Gaussian measurements and sparse source vectors.
While there are certainly other important examples of measurement designs,
the Gaussian case typically serves as a proof-of-concept that allows
for a rigorous statistical analysis and highlights the key methodology of
this paper. However, several relevant extensions are presented in Section~\ref{sec:extensions}.
Before stating the actual results, let
us set up a formal (random) measurement model, which is assumed to
hold true for the remainder of this section.
\begin{model}[Measurement Scheme -- Gaussian Case]
  \label{model:results:measurements}
  Let $\a^1, \dots, \a^M \distributed \Normdistr{\vnull}{\I{n}}$ be
  independent standard Gaussian random vectors and let
  $\grtr \in \R^n$ with $\lnorm{\grtr} = 1$.  We define the \emph{superimposed, non-linearly distorted measurement process} by
  \begin{equation}\label{eq:results:measurements:meas}
    y \coloneqq \sum_{j = 1}^M \fobs_j(\sp{\a^j}{\grtr}) + e,
  \end{equation}
  where $e$ is independent, mean-zero, sub-Gaussian noise with
  $\normsubg{e} \leq \nu$ and $\fobs_j \colon \R \to \R$,
  $j = 1, \dots, M$ are (unknown) scalar functions. Moreover, we
  assume $\mean[\fobs_j(\sp{\a^j}{\grtr})] = 0$ for all
  $j = 1, \dots, M$.  Each of the $m$ samples
  $\{(\{ \a_i^j \}_{j = 1}^M, y_i)\}_{i = 1}^m$ is then drawn as an
  independent copy of the random ensemble
  $(\{ \a^j \}_{j = 1}^M, y)$.
\end{model}

\begin{remark}\label{rmk:results:normalization}
We would like to emphasize that a normalization of $\grtr$ as in Model~\ref{model:results:measurements} is quite natural in the setup of non-linear measurements.
\revision{
Indeed, if $\grtr \in \R^n \setminus \{\vnull\}$, we observe that 
\begin{equation}
	\fobs_j(\sp{\a^j}{\grtr}) = \fobs_j(\lnorm{\grtr} \cdot \sp{\a^j}{\tfrac{\grtr}{\lnorm{\grtr}}}).
\end{equation}
Hence, changing the norm of $\grtr$ in Model~\ref{model:results:measurements} can be ``compensated'' by rescaling the output functions $\fobs_1, \dots, \fobs_M$, i.e., replacing $v \mapsto \fobs_j(v)$ with $v \mapsto \tilde{\fobs}_j(v) \coloneqq \fobs_j(\lnorm{\grtr} \cdot v)$.
On the other hand, our estimators \eqref{eq:direct:lasso} and \eqref{eq:lifting:l1l2lasso} are not aware of possibly non-linear output functions, so that it is reasonable to fix a certain ``operating point.'' Moreover, recovery of the magnitude of $\grtr$ might be impossible for specific instances of Model~\ref{model:results:measurements}, e.g., if $\fobs_j = \sign(\cdot)$ for all $j = 1, \dots, M$, as these non-linearities are invariant under rescaling.}

The only notable structural constraint on the output functions $\fobs_j$ in Model~\ref{model:results:measurements} is the mean-zero assumption $\mean[\fobs_j(\sp{\a^j}{\grtr})] = 0$. However, this is usually not a big issue from an application viewpoint, since the non-linear distortions can be often expressed by odd functions, such as in the case of wireless sensor networks (cf. \eqref{eq:practice:clip} and \eqref{eq:practice:channelmultiplication}).
Apart from this, if the transmitted signal $\fobs_j(\sp{\a^j}{\grtr})$ is affected by a known bias, it is a common step in practice to use an online-correction for this bias.
\end{remark}

\subsection{Recovery via the Direct Method}

While Algorithm~\ref{algo:direct} does not explicitly depend on the output
functions $\fobs_j$, it is not
surprising that they should have a certain impact on its actual recovery
performance.  In order to quantify the degree of distortion that is
generated by these non-linearities, let us first introduce the following quantities:
\begin{definition}\label{def:results:scalparam}
  For each node $j = 1, \dots, M$, we define the \emph{scaling
    parameter}
  \begin{equation}\label{eq:results:scalparam:scalfac}
    \scalfac_j \coloneqq \mean_\gaussianuniv{}[\fobs_j(\gaussianuniv) \cdot \gaussianuniv], \quad \gaussianuniv \distributed \Normdistr{0}{1},
  \end{equation}
  and the \emph{scaling vector}
  $\scalvec \coloneqq (\scalfac_1, \dots, \scalfac_M) \in \R^M$.  The
  \emph{mean scaling parameter} is given by
  \begin{equation}
    \bar\scalfac \coloneqq \tfrac{1}{M} \sum_{j = 1}^M \scalfac_j.
  \end{equation}
\end{definition}
\begin{figure}
	\centering
	\begin{subfigure}[t]{0.25\textwidth}
		\centering
		\includegraphics[width=\textwidth]{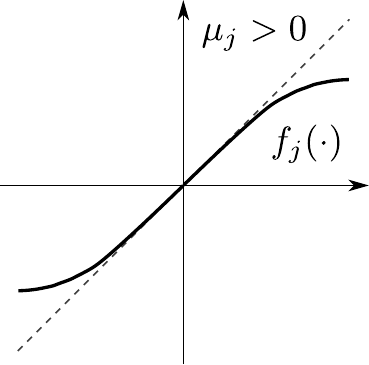}
		\caption{}
		\label{fig:results:nonlinearity:pos}
	\end{subfigure}%
	\qquad\qquad
	\begin{subfigure}[t]{0.25\textwidth}
		\centering
		\includegraphics[width=\textwidth]{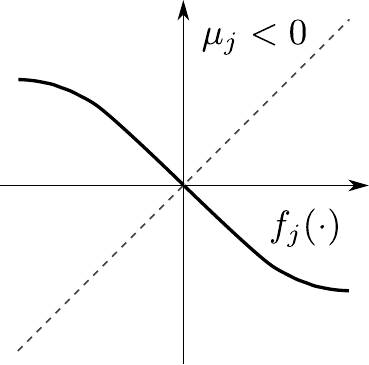}
		\caption{}
		\label{fig:results:nonlinearity:neg}
	\end{subfigure}%
	\caption{Illustration of the scaling parameters $\scalfac_j$. \subref{fig:results:nonlinearity:pos} In this case, $\fobs_j$ ``aligns well'' with $\Id$ so that $\scalfac_j$ is close to $1$. \subref{fig:results:nonlinearity:neg} Here, $\fobs_j$ involves a sign-flip and therefore rather matches with ${-\!\Id}$. This is particularly reflected by a sign-change of $\scalfac_j$.}
	\label{fig:results:nonlinearity}
\end{figure}
From a statistical perspective, \eqref{eq:results:scalparam:scalfac}
just computes the covariance between the distortion function
$\gaussianuniv \mapsto \fobs_j(\gaussianuniv)$ of the $j$-th node and its
linear counterpart.  Intuitively, $\scalfac_j$ measures the
expected\footnote{Note that we assume $\lnorm{\grtr} = 1$ in
  Model~\ref{model:results:measurements} so that indeed
  $\gaussianuniv = \sp{\a^j}{\grtr} \distributed \Normdistr{0}{1}$.}
rescaling and sign-change caused by $\fobs_j$ (with respect to the
identity), whereas $\bar\scalfac$ captures the ``average''
rescaling of the entire measurement process
\eqref{eq:results:measurements:meas}; see
Figure~\ref{fig:results:nonlinearity}.  The simple ``matching
principle'' of \eqref{eq:results:scalparam:scalfac} resembles the definitions \revision{of \cite[Eq.~(3.1)]{Brillinger1982} and} \cite[Eq.~(I.5)]{plan2015lasso}, where the case $M = 1$ is studied. The following recovery guarantee for the direct method shows that the empirical mean computed by $\bar\scalfac$ is the appropriate scaling parameter to handle superimposed
measurements:
\begin{theorem}[Direct Method -- Algorithm~\ref{algo:direct}]\label{thm:results:direct}
  Let Model~\ref{model:results:measurements} hold true and assume that
  $\grtr$ is $s$-sparse.\footnote{That means $\lnorm{\grtr}[0] \leq s$.}  Then, there exists a numerical constant $C > 0$
  such that the following holds true for every (fixed)
  $\delta \in (0,1]$ with probability at least
  $1 - 5 \exp(- C \cdot \delta^2 \cdot m )$: If
  \begin{equation}\label{eq:results:direct:meas}
    m \gtrsim \delta^{-2} \cdot s \cdot \log(\tfrac{2n}{s}),
  \end{equation}
  then any minimizer $\solu \in \R^{n}$ of \eqref{eq:direct:lasso} with $R = \lnorm{\bar\scalfac \grtr}[1]$ satisfies
  \begin{equation}\label{eq:results:direct:bound}
    \lnorm{\solu - \bar\scalfac\grtr} \lesssim (\modeldev_\text{Dir}^2 + \tfrac{\nu^2}{M})^{\frac{1}{2}} \cdot \delta,
  \end{equation} 
  where 
  \begin{equation}
    \modeldev_\text{Dir}^2 \coloneqq \tfrac{1}{M} \sum_{j = 1}^M \normsubg{\fobs_j(\gaussianuniv) - \bar\scalfac \gaussianuniv}^2\quad \text{with $\gaussianuniv \distributed \Normdistr{0}{1}$.}
  \end{equation}
\end{theorem} 
Roughly speaking, Theorem~\ref{thm:results:direct} states that, with
high probability, estimating $\grtr$ via the direct method
(Algorithm~\ref{algo:direct}) is very accurate if $m$ (greatly)
exceeds $s \log(2n/s)$ --- which is a sampling rate that resembles
the typical flavor of results from compressed sensing theory.  \revision{We
would like to emphasize that the non-linearities $\fobs_j$ as well as
the sensor count $M$ affect the error bound
\eqref{eq:results:direct:bound} only in terms of the (constant) factors $\bar\scalfac$ and $\modeldev_{\text{Dir}}$:
The mean scaling parameter $\bar\scalfac$ leads to a rescaling of the source vector $\grtr$, implying that $\solu$ yields a consistent estimator of $\bar\scalfac\grtr$ but not of $\grtr$ itself. Noteworthy, due to the assumption $\lnorm{\grtr} = 1$ in Model~\ref{model:results:measurements}, we therefore do not only obtain an approximation of $\grtr$ (by computing $\solu / \lnorm{\solu}$) but also of the unknown factor $\bar\scalfac$ (by computing $\lnorm{\solu}$).
In comparison, the \emph{model deviation} $\modeldev_{\text{Dir}}$ quantifies the absolute \emph{mismatch} caused by
non-linear perturbations, which in turn controls the variance of the estimator \eqref{eq:direct:lasso}.}

The impact of the additive noise term $e$ becomes even
smaller in \eqref{eq:results:direct:bound} as $M$
grows.\footnote{Recall that Model~\ref{model:results:measurements}
  assumes $\normsubg{e} \leq \nu$.} Such a behavior is well-known
from the linear case (i.e., when all $f_j$ are linear functions) and
Theorem~\ref{thm:results:direct} shows that this desirable ``rule-of-thumb'' even remains
true in the general non-linear situation.  This particularly
reflects the fact that the signal-to-noise ratio of
\eqref{eq:results:measurements:meas} increases as $M$ gets larger,
indicating that we can expect more accurate and stable reconstruction
results.
Regarding our initial problem issues \ref{question:intro:samplecomplexity}--\ref{question:intro:nodecount}, we can therefore draw the following somewhat informal conclusion:
\begin{highlight}
	The Lasso \eqref{eq:direct:lasso} with superimposed, non-linear measurements \eqref{eq:results:measurements:meas} essentially performs as well as if it would be applied to noisy linear observations (communicated by a single node). Since this estimator does not make explicit use of the non-linearities, the price to pay is the presence of an unknown factor $\bar\scalfac$ that rescales the source vector $\grtr$.
\end{highlight}

\begin{remark}
\begin{rmklist}
\item
	The underlying proof techniques of Theorem~\ref{thm:results:direct} do only marginally rely on the fact that $\grtr$ is assumed to be $s$-sparse (cf. Theorem~\ref{thm:proofs:abstractrecovery}). However, an extension to more general structural constraints requires some technical preliminaries, which will be presented in the course of Section~\ref{sec:extensions}.
\item\label{rmk:results:direct:tuning}
	The choice of the tuning parameter $R$ in Theorem~\ref{thm:results:direct} is quite restrictive, since both $\lnorm{\grtr}[1]$ and $\bar\scalfac$ are often unknown in practice.
	In fact, this condition can be relaxed in the sense that the above statement essentially remains valid if $R$ was not perfectly chosen (see Subsection~\ref{subsec:extensions:hybridglobal}, in particular Theorem~\ref{thm:extensions:hybridglobal}).\qedhere
\end{rmklist}
\label{rmk:results:direct}
\end{remark}

\subsection{Recovery via the Lifting Method}

Operationally, as long as $\bar\scalfac\neq 0$,
the recovered vector $\solu$ only needs to be rescaled by a factor of $1/\bar\scalfac$ in order to become a reliable estimator of the normalized vector $\grtr$.
The significance of Theorem~\ref{thm:results:direct} is however
lost if $\bar\scalfac \approx 0$, since dividing \eqref{eq:results:direct:bound} by $\abs{\bar\scalfac}$ would blow up the error bound.
Our next main result shows that this problematic situation can be resolved by the lifting method:

\begin{theorem}[Lifting Method -- Algorithm~\ref{algo:lifting}]\label{thm:results:lifting} 
	Let Model~\ref{model:results:measurements} hold true and assume that $\grtr$ is $s$-sparse.
	Then, there exists a numerical constant $C > 0$ such that the following holds true for every (fixed)
	$\delta \in (0,1]$ with probability at least $1 - 5 \exp(- C \cdot \delta^2 \cdot m )$: If
	\begin{equation}\label{eq:results:lifting:meas}
		m \gtrsim \delta^{-2} \cdot s \cdot \max\{M, \log(\tfrac{2n}{s})\},
	\end{equation}
	\revision{then any minimizer $\solu[\X] \coloneqq [\solu^1 \cdots \solu^M] \in \R^{n\times M}$ of \eqref{eq:lifting:l1l2lasso} with $R = \lnorm[\big]{[\scalfac_1\grtr \cdots \scalfac_M\grtr]}[1,2] = \lnorm{\scalvec} \lnorm{\grtr}[1]$ satisfies
	\begin{equation}\label{eq:results:lifting:bound}
		\tfrac{1}{\sqrt{M}}\lnorm{\solu[\X] - \grtr\scalvec^\T} = \Big(\tfrac{1}{M}\sum_{j = 1}^M\lnorm{\solu^j - \scalfac_j\grtr}^2 \Big)^{1/2} \lesssim (\modeldev_\text{Lift}^2 + \tfrac{\nu^2}{M})^{\frac{1}{2}} \cdot \delta,
	\end{equation} 
	where 
	\begin{equation}
		\modeldev_\text{Lift}^2 \coloneqq \tfrac{1}{M} \sum_{j = 1}^M \normsubg{\fobs_j(\gaussianuniv) - \scalfac_j \gaussianuniv}^2\quad \text{with $\gaussianuniv \distributed \Normdistr{0}{1}$.}
	\end{equation}
	Moreover, if $\Err \coloneqq (\modeldev_\text{Lift}^2 + \nu^2 / M)^{\frac{1}{2}} \cdot \delta \lesssim \lnorm{\scalvec} / \sqrt{M}$, we also have the following: let $\solu \in \S^{n-1}$ and $\solu[\scalvec] \in \S^{M-1}$ denote the dominating left and right singular vectors of $\solu[\X]$, respectively, and let $\tau$ be the corresponding singular value. Then
	\begin{equation}\label{eq:results:lifting:bound:svd}
		\lnorm{\solu - \theta_1 \grtr} \lesssim \Err \cdot \tfrac{\sqrt{M}}{\lnorm{\scalvec}} \quad \text{and} \quad \lnorm{\tau\solu[\scalvec] - \theta_2 \scalvec} \lesssim \Err \cdot \sqrt{M},
	\end{equation}
	where $\theta_1 \coloneqq \sign(\sp{\grtr}{\solu})$ and $\theta_2 \coloneqq \sign(\sp{\scalvec}{\solu[\scalvec]})$.}
\end{theorem}

First, we would like to point out that the mean scaling parameter $\bar\scalfac$ does not appear in Theorem~\ref{thm:results:lifting} anymore. 
Indeed, the \emph{model deviation parameter} $\modeldev_\text{Lift}$ is also refined, since each mismatch term $\fobs_j(\gaussianuniv) - \scalfac_j \gaussianuniv$ now involves the correct scaling factor $\scalfac_j$. 
The actual error bound \eqref{eq:results:lifting:bound} reads slightly more complicated;
roughly speaking, each column $\solu^j \in \R^n$ of the minimizer estimates a scaled version of $\grtr$ and \eqref{eq:results:lifting:bound} states that the $\l{2}$-mean of all approximation errors is small.
\revision{Interestingly, Theorem~\ref{thm:results:lifting} tells us even more: the solution matrix $\solu[\X]$ of the Group-Lasso \eqref{eq:lifting:l1l2lasso} approximates the rank-one
matrix $\grtr\scalvec^\T \in \R^{n \times M}$.  Hence, according to \eqref{eq:results:lifting:bound:svd}, we do not only obtain an estimate of the source $\grtr$ but also
of the unknown scaling parameters
$\scalvec = (\scalfac_1, \dots, \scalfac_M) \in \R^M$, by simply computing the dominating left
and right singular vectors of $\solu[\X]$; see also the optional step in Algorithm~\ref{algo:lifting}. Note that the typical scaling ambiguity does only appear in form of sign-changes here, due to the unit-norm assumption in Model~\ref{model:results:measurements}; see also Remark~\ref{rmk:results:normalization}.} 
This important observation
underpins the relation to the bilinear problem of
\eqref{eq:intro:bilinearmodel} and indicates that the lifting method
(Algorithm~\ref{algo:lifting}) is indeed capable of handling more
complicated scenarios than the direct method
(Algorithm~\ref{algo:direct}).

In terms of required measurements \eqref{eq:results:lifting:meas},
this improvement comes at no extra costs as long as
$M = \asympfaster{\log(2n/s)}$.  But for larger distributed systems, one has to
take significantly more measurements because
$m = \asympfaster{s\cdot M}$ then grows linearly with the node count
$M$.  However, the total number of unknown non-zero parameters is
actually in the order of $\asympfaster{s + M}$.  This gap between
multiplicative and additive scaling of the (sample) complexity is due to
the fact that \eqref{eq:lifting:l1l2lasso} does not account for
the low-rankness of $\grtr\scalvec^\T$.  While our focus is
rather on \emph{non-linear} distortions than on solving a bilinear
factorization problem, there is certainly space for improvements in
this direction, which we intend to study in future works (see Section~\ref{sec:conclusion}). Thus, in light of \ref{question:intro:samplecomplexity}--\ref{question:intro:nodecount}, our discussion can be informally summarized as follows:
\begin{highlight}
	\revision{The Group-Lasso estimator \eqref{eq:lifting:l1l2lasso}, when combined with the optional step of Algorithm~\ref{algo:lifting}, is capable of recovering both the source vector $\grtr$ and the scaling vector $\scalvec$ (up to sign-changes).}
	The price to pay is that the number of required measurement (at some point) scales linearly in $M$.
\end{highlight}

\begin{remark}
\begin{rmklist}
\item\label{rmk:results:lifting:matrixversion}
	In some situations (especially in the proofs of Section~\ref{sec:proofs}), it is useful to restate \eqref{eq:lifting:l1l2lasso} in its matrix version.
	For this purpose, we simply concatenate the measurement vectors $\alift_i \coloneqq [\a_i^1 \cdots \a_i^M] \in \R^{n \times M}$ and source vectors $\X \coloneqq [\x^1 \cdots \x^M] \in \R^{n \times M}$.
	Then, \eqref{eq:lifting:l1l2lasso} takes the form
	\begin{equation}\tag{${P_R^\text{Lift}}'$}
		\min_{\X \in \R^{n\times M}} \sum_{i = 1}^m (y_i - \sp{\alift_i}{\X})^2 \quad \text{s.t. $\lnorm{\X}[1,2] \leq R$.}
	\end{equation}
	This perspective reveals that the lifting method actually tries to fit a bilinear structure to \mbox{(bi-)non}-linear observations, eventually estimating a rank-one matrix (cf. \eqref{eq:results:lifting:bound}).  Note that a similar lifting approach combined with the $\l{1,2}$-norm was recently studied in \cite{ling2015selfcalib,flinth2016sparsedeconv}, considering problems from \emph{self-calibration} and \emph{sparse blind deconvolution}.
\item
	Although the lifting method works in a quite general setting, there are still some scenarios which are (implicitly) excluded.
	More precisely, if all scaling factors $\scalfac_j$ do vanish, Theorem~\ref{thm:results:lifting} states that \eqref{eq:lifting:l1l2lasso} just recovers the $\vnull$-matrix and there is particularly no hope to retrieve $\grtr$.
	\revision{An important example of such an undesirable situation is \emph{phase retrieval} \cite{shechtman2015phase}, $\fobs_j(v) = \abs{v}$, where $\scalfac_j = \mean[\abs{\gaussianuniv} \cdot \gaussianuniv] = 0$. Interestingly, it recently turned out that signal recovery is still possible in situations where $\scalfac_j = 0$ (which includes even, non-linear functions), e.g., see \cite{thrampoulidis2018lifting} for a convex approach based on lifting and \cite{yang2017misspecified} for a non-convex method.}
\item 
	\revision{As already pointed out in the above discussion of Theorem~\ref{thm:results:lifting}, the sampling rate $m = \asympfaster{s \cdot \max\{M, \log(\tfrac{2n}{s})\}}$ proposed by \eqref{eq:results:lifting:meas} does not match with the best possible rate $m = \asympfaster{s + M}$. This gap is due to a fundamental limitation of the lifting method \eqref{eq:lifting:l1l2lasso} and not of our statistical analysis, which is based on bounds for the conic mean width (see Proposition~\ref{prop:proofs:meanwidth}); in other words, the sampling rate $m = \asympfaster{s \cdot \max\{M, \log(\tfrac{2n}{s})\}}$ is expected to be (asymptotically) optimal, even if one is only interested in recovery of $\grtr$ via \eqref{eq:lifting:l1l2lasso}, or if the measurement process is linear.
	On the other hand, there exists evidence that it is possible to break the multiplicative complexity barrier of Theorem~\ref{thm:results:lifting} by using more advanced (non-convex) methods instead of \eqref{eq:lifting:l1l2lasso}; see Section~\ref{sec:conclusion}.}
	\qedhere
\end{rmklist}
\label{rmk:results:lifting}
\end{remark}

\section{Extensions}
\label{sec:extensions}

While the framework developed in Section~\ref{sec:results} is quite elegant and simple, there are several undesirable restrictions regarding
practical purposes.  For example, we would like to allow for a larger
class of sub-Gaussian distributions.  The following model includes
this extension. Note that it coincides with
Model~\ref{model:results:measurements} in the Gaussian case. Moreover,
we have also relaxed the unit-norm assumption on $\grtr$ (cf. Remark~\ref{rmk:results:normalization}).

\begin{model}[Measurement Scheme -- Sub-Gaussian Case]\label{model:extensions:measurements}
	Let $\grtr \in \R^n \setminus \{\vnull\}$ and assume that $\a^1, \dots, \a^M$ are independent, isotropic, mean-zero sub-Gaussian random vectors with $\normsubg{\a^j} \leq \subgparam$ for all $j = 1, \dots, M$.
	As before, we define the measurement process by
	\begin{equation}\label{eq:extensions:measurements:meas}
		y \coloneqq \sum_{j = 1}^M \fobs_j(\sp{\a^j}{\grtr}) + e,
	\end{equation}
	where $e$ is independent, mean-zero, sub-Gaussian noise with $\normsubg{e} \leq \nu$ and $\fobs_j \colon \R \to \R$, $j = 1, \dots, M$ are scalar functions.
	Again, we assume that $\mean[\fobs_j(\sp{\a^j}{\grtr})] = 0$ for all $j = 1, \dots, M$.	
	The individual samples $\{(\{ \a_i^j \}_{j = 1}^M, y_i)\}_{i = 1}^m$ are then drawn as independent copies of the random ensemble $(\{ \a^j \}_{j = 1}^M, y)$.
\end{model}

\subsection{The Hybrid Method}
\label{subsec:extensions:hybrid}

A downside of Algorithm~\ref{algo:direct} and Algorithm~\ref{algo:lifting} is that they do not offer an incorporation of prior knowledge about the model configuration.
Let us consider a simple example:
Suppose that the non-linearities take the form $\fobs_j(\cdot) = h_j \cdot \fobs(\cdot)$, where $h_j \in \R$ and $\fobs\colon \R \to \R$.
Using the direct method with a ``positive'' superposition $\adir = \sum_{j = 1}^M \a^j$, we actually insinuate that all nodes align in the same direction, meaning that $h_1, \dots, h_M$ are all of the same sign.
In this case, we would have
\begin{equation}
	\abs{\bar\scalfac} = \tfrac{1}{M}\sum_{j = 1}^M \abs{h_j} \cdot \abs{\underbrace{\mean[\fobs(\gaussianuniv)\gaussianuniv]}_{\eqqcolon \scalfac}} \gg 0.
\end{equation}
However, if $h_1, \dots, h_M$ have different signs, we might end up with
\begin{equation}
	\abs{\bar\scalfac} = \tfrac{1}{M}\abs[\Big]{\sum_{j = 1}^M h_j} \cdot \abs{\scalfac} \approx 0,
\end{equation}
and the statement of Theorem~\ref{thm:results:direct} becomes essentially meaningless.
The underlying problem is that one part of the summands of $\adir = \sum_{j = 1}^M \a^j$ ``matches'' with the superimposed measurement model \eqref{eq:extensions:measurements:meas}, whereas the other part does not.
On the other hand, if the signs of $h_1, \dots, h_M$ are available as prior information, we may easily circumvent this ``cancellation'' phenomenon by computing adapted superpositions $\ahyb = \sum_{j = 1}^M \sign(h_j) \a^j$.
We will see below that this strategy is indeed very helpful (see Example~\ref{ex:extensions:hybrid}\ref{ex:extensions:hybrid:priorinfo}).

Let us now introduce an algorithmic framework that enables us to design such (linear) combinations of measurement vectors in a very general way:

\enlargethispage{1.5\baselineskip}
\begin{algorithm}[H]
\caption{Hybrid Method}
\label{algo:hybrid}
	\Input{Measurement ensemble $\{(\{ \a_i^j \}_{j = 1}^M, y_i)\}_{i = 1}^m$, convex set $\sset \subset \R^{n\times N}$,\newline weight matrix $\hybW = [w_{j,k}] \in \R^{M\times N}$}
	\Output{Estimated source vectors $\solu^1,\dots,\solu^{N} \in \R^n$}
\BlankLine
\SetKwBlock{Compute}{Compute}{}
\Compute{\BlankLine
	\nl Compute hybrid measurement vectors:
	\begin{equation}\label{eq:hybrid:hybmeas}
		\ahyb_i^k \coloneqq \sum_{j = 1}^M w_{j,k}\a_i^j, \quad i = 1, \dots, m, \quad k = 1, \dots, N.
	\end{equation}
	\BlankLine
	\nl Solve the Group-Lasso and pick any minimizer:
	\begin{equation}
		\label{eq:hybrid:lasso}\tag{$P_{\sset}^\text{Hyb}$}
		[\solu^1\cdots\solu^N] = \argmin_{\substack{[\x^1\cdots\x^N] \\ \in \R^{n\times N}}} \ \tfrac{1}{2m} \sum_{i = 1}^m
		\Big(y_i - \sum_{k = 1}^N\sp{\ahyb_i^k}{\x^k}\Big)^2 \quad
		\text{s.t. $[\x^1 \cdots \x^N] \in \sset$.}
	\end{equation}
	}
\end{algorithm}

Picking $\hybW = (1, \dots, 1) \in \R^{M\times 1} $ for the direct method and $\hybW = \I{M} \in \R^{M \times M}$ for the lifting method, we immediately obtain the measurement designs of Algorithm~\ref{algo:direct} and Algorithm~\ref{algo:lifting}, respectively.
This particular shows that our initial approaches are special (extreme) cases of the hybrid method; see also Example~\ref{ex:extensions:hybrid}\ref{ex:extensions:hybrid:direct} and \ref{ex:extensions:hybrid:lifting} below.
\revision{The above algorithm however leaves us much more freedom to tackle the recovery problem: It permits arbitrary linear combinations of the measurement ensembles $\{ \a_i^j \}_{j = 1}^M$;
each resulting hybrid measurement vector $\ahyb_i^k$ in \eqref{eq:hybrid:hybmeas} reflects a certain \emph{hypothesis} on how the raw measurements vectors $\{ \a_i^j \}_{j = 1}^M$ should be appropriately linked with each other.
Thus, the total count $N$ of hybrid measurement vectors corresponds to the number of hypotheses made; in particular, there is only $N = 1$ ``strong'' hypothesis for the direct method, while there are $N = M$ ``independent'' hypotheses for the lifting method, avoiding any additional presumptions.
Finally, let us point out that Algorithm~\ref{algo:hybrid} is not restricted to sparsity-promoting constraints anymore but allows for general convex constraint sets.}

\subsection{Recovery Based on the Conic Mean Width}

Before stating our main recovery guarantee, we first need to adapt the scaling parameters from Definition~\ref{def:results:scalparam}:
\begin{definition}\label{def:extensions:scalparam}
	We define the \emph{scaling parameters}
	\begin{equation}
		\scalfac_j \coloneqq \tfrac{1}{\lnorm{\grtr}^2} \cdot \mean{}_{\a^j}[\fobs_j(\sp{\a^j}{\grtr}) \sp{\a^j}{\grtr}], \quad j = 1, \dots, M,
	\end{equation}
	and $\scalvec \coloneqq (\scalfac_1, \dots, \scalfac_M) \in \R^M$.
	Given $\hybW \in \R^{M\times N}$, we introduce the \emph{hybrid scaling vector} by
	\begin{equation}\label{eq:extensions:scalparam:hybscalvec}
		\hybscalvec = (\hybscalfac_1, \dots, \hybscalfac_N) \coloneqq \tfrac{N}{M} \hybW^\T \scalvec \in \R^N.
	\end{equation}
	Furthermore, we call
	\begin{equation}
		\orthmodelbias^j \coloneqq \mean{}_{\a^j}[\fobs_{j}(\sp{\a^{j}}{\grtr})\proj_{\orthcompl{\{\grtr\}}}(\a^{j})] \in \R^{n} \quad \text{for $j = 1, \dots, M$,}
	\end{equation}
	the \emph{isotropy mismatch vectors} and set $\orthmodelbias \coloneqq [\orthmodelbias^1 \cdots \orthmodelbias^M] \in \R^{n \times M}$.
\end{definition}

Under the hypothesis of Model~\ref{model:results:measurements}, the scaling parameters $\scalfac_j$ exactly coincide with those of \eqref{eq:results:scalparam:scalfac}, and the dependence on $\lnorm{\grtr}$ is due to the missing unit-norm assumption in Model~\ref{model:extensions:measurements}.
In contrast, the hybrid scaling vector $\hybscalvec$ arises from a linear transformation by the weight matrix $\hybW$. \revision{Our recovery results and their proofs (in Subsection~\ref{subsec:proofs:hybridresults}) will reveal that the above definitions of $\hybscalvec$ and $\orthmodelbias$ are both quite natural, since they lead in a certain sense to the ``best approximation'' of the non-linear observation rule by a linear model.}

The isotropy mismatch vectors vanish for Gaussian random vectors ($\orthmodelbias^j = \vnull$) because $\sp{\a^{j}}{\grtr}$ and $\proj_{\orthcompl{\{\grtr\}}}(\a^{j})$ are independent in this case (and not just uncorrelated).
However, this does not happen in the general sub-Gaussian scenario, unless $\fobs_j$ is linear. 
In fact, $\orthmodelbias^j$ computes the covariance between the distorted projection $\fobs_j(\sp{\a^{j}}{\grtr})$ of $\a^{j}$ onto $\spann\{\grtr\}$ and its orthogonal complement $\proj_{\orthcompl{\{\grtr\}}}(\a^{j})$. Thus, we may regard the mismatch vector $\orthmodelbias^j$ as a compatibility measure of the non-linearity $\fobs_j$ and the isotropic measurement vector $\a^j$.
Our next result shows that these terms indeed play an important role in the estimation performance of the hybrid method:

\begin{theorem}[Hybrid Method -- Algorithm~\ref{algo:hybrid}]\label{thm:extensions:hybrid} 
	Assume that Model~\ref{model:extensions:measurements} holds true. Let $\hybW = [\hybw^1 \cdots \hybw^N] \in \R^{M \times N}$ be a weight matrix satisfying $\hybW^\T \hybW = \tfrac{M}{N} \I{N}$. Moreover, suppose that $\grtr\hybscalvec^\T \in \sset$, where $\sset \subset \R^{n \times N}$ is a convex set.
	Then, there exists a numerical constant $C > 0$ such that the following holds true for every (fixed)
	$\delta \in (0,1]$ with probability at least $1 - 5 \exp(- C \cdot \subgparam^{-4} \cdot \delta^2 \cdot m )$: If \footnote{Here, $\meanwidth[1]{\cone{\sset}{\grtr\hybscalvec^\T}}$ denotes the \emph{conic mean width} of $\sset$ at $\grtr\hybscalvec^\T$, which is formally introduced below in Definition~\ref{def:proofs:meanwidth}.}
	\begin{equation}\label{eq:extensions:hybrid:meas}
		m \gtrsim \subgparam^{4} \cdot \delta^{-2} \cdot \meanwidth[1]{\cone{\sset}{\grtr\hybscalvec^\T}}^2,
	\end{equation}
	then any minimizer $[\solu^1 \cdots \solu^N] \in \R^{n\times N}$ of \eqref{eq:hybrid:lasso} satisfies
	\begin{equation}\label{eq:extensions:hybrid:bound}
		\Big(\tfrac{1}{N}\sum_{k = 1}^N\lnorm{\solu^k - \hybscalfac_k\grtr}^2 \Big)^{1/2} \lesssim \subgparam^{-1} \cdot (\modeldev_\text{Hyb}^2 + \tfrac{\nu^2}{M})^{\frac{1}{2}} \cdot \delta + \tfrac{N}{M} \cdot \modelcovar_\text{Hyb},
	\end{equation} 
	where 
	\begin{equation}\label{eq:extensions:hybrid:modeldev}
		\modeldev_\text{Hyb}^2 \coloneqq \tfrac{1}{M} \sum_{j = 1}^M \normsubg{\sp{\a^j}{\hybgrtr^j} - \fobs_j(\sp{\a^j}{\grtr})}^2\quad \text{with $[\hybgrtr^1 \cdots \hybgrtr^M] \coloneqq \tfrac{N}{M} \grtr \scalvec^\T \hybW \hybW^\T$}
	\end{equation}
	and
	\begin{equation}\label{eq:extensions:hybrid:modelcovar}
		\modelcovar_\text{Hyb} \coloneqq \Big(\tfrac{1}{N} \sum_{k = 1}^N \lnorm{\orthmodelbias \hybw^k}^2\Big)^{1/2} = \tfrac{1}{\sqrt{N}} \lnorm{\orthmodelbias \hybW}.
	\end{equation}
\end{theorem}

The statement of Theorem~\ref{thm:extensions:hybrid} strongly resembles Theorem~\ref{thm:results:lifting}.
In particular, we observe again that
\begin{equation}
	\Big(\tfrac{1}{N}\sum_{k = 1}^N\lnorm{\solu^k - \hybscalfac_k\grtr}^2 \Big)^{1/2} = \tfrac{1}{\sqrt{N}}\lnorm{[\solu^1 \cdots \solu^N] - \grtr\hybscalvec^\T},
\end{equation}
implying that the Group-Lasso \eqref{eq:hybrid:lasso} allows for estimation of both the desired source vector $\grtr$ and the hybrid scaling vector $\hybscalvec$ (up to a constant scaling factor).

While the sub-Gaussian parameter $\subgparam$ just appears as a constant factor, the impact of $\modelcovar_\text{Hyb}$ is much more significant.
This additional model parameter involves the isotropy mismatch vectors $\orthmodelbias^1, \dots, \orthmodelbias^M$ and precisely quantifies the \revision{\emph{asymptotic bias} (approximation error) of the estimator \eqref{eq:hybrid:lasso}} that is due to the sub-Gaussianity of the measurement vectors $\a^1, \dots, \a^M$.
Note that an additive error term of this type did not occur in Section~\ref{sec:results}, since we always have $\orthmodelbias = \vnull$ in the Gaussian case, and therefore $\modelcovar_\text{Hyb} = 0$.
\revision{Theorem~\ref{thm:extensions:hybrid} indicates that recovery is still feasible for sub-Gaussian distributions, but the resulting accuracy strongly depends on the size of $\modelcovar_\text{Hyb}$; in particular, the estimator \eqref{eq:hybrid:lasso} is not necessarily consistent if $\modelcovar_\text{Hyb} > 0$.
Without any further assumptions, it is in fact very difficult to provide sharp bounds on $\modelcovar_\text{Hyb}$.
However, it will turn out later in the context of Proposition~\ref{prop:proofs:modelcovar:hybrid} that $\modelcovar_\text{Hyb}$ is a specific version of the mismatch covariance, which is a key ingredient of our proofs in Section~\ref{sec:proofs} (see Definition~\ref{def:proofs:modelmismatch}). The mismatch covariance is in turn very closely related to an expression studied in \cite[Eq.~(2.1)]{ai2014onebitsubgauss}, dealing with signal recovery from binary, sub-Gaussian measurements.
We therefore expect that the tools of \cite{ai2014onebitsubgauss} apply to our setting as well (at least for special types of output functions $\fobs_j$), leading to worst-case upper bounds on $\modelcovar_\text{Hyb}$. An alternative way to control the size of $\modelcovar_\text{Hyb}$ is based on \emph{dithering}, which would involve a slight modification of the measurement process in Model~\ref{model:extensions:measurements} (see Subsection~\ref{subsec:extensions:further}).} 

Compared to \eqref{eq:results:direct:meas} and \eqref{eq:results:lifting:meas}, the sampling rate of \eqref{eq:extensions:hybrid:meas} is now determined by the conic mean width $\meanwidth[1]{\cone{\sset}{\grtr\hybscalvec^\T}}$.
This condition is not as explicit as in the case of sparse vectors but clearly leaves space for many different structured representations, such as sparsity in dictionaries. We will return to this point in the course of Theorem~\ref{thm:extensions:hybridglobal}, which can be regarded as a global version of Theorem~\ref{thm:extensions:hybrid}.
Let us conclude our discussion of Theorem~\ref{thm:extensions:hybrid} with the following informal summary:
\revision{
\begin{highlight}
	The hybrid method (Algorithm~\ref{algo:hybrid}) is a natural extension of the direct and lifting method, allowing for a unified treatment of both.
	The theoretical recovery guarantees from Section~\ref{sec:results} can be extended to general convex constraints on the source vector as well as sub-Gaussian measurment ensembles.
	However, going beyond Gaussian measurements comes along with an asymptotic bias of the estimator, which needs to be sufficiently small to ensure accurate reconstructions.
\end{highlight}}

\begin{remark}\label{rmk:extensions:hybrid:weightcondition}
	The semi-orthogonality condition $\hybW^\T \hybW = \tfrac{M}{N} \I{N}$ in Theorem~\ref{thm:extensions:hybrid} is essential to ensure that the measurement designs $[\ahyb_i^1 \cdots \ahyb_i^N]$ are isotropic.
	Fortunately, this assumption can be easily relaxed by a simple trick that was already applied in \cite[Cor.~1.6]{plan2015lasso} and \cite[Thm.~2.8]{genzel2016estimation}:
	Let us assume that the weight matrix $\hybW \in \R^{M \times N}$ in Algorithm~\ref{algo:hybrid} is just a matrix of full rank with $M \geq N$. Then, $\U \coloneqq \tfrac{N}{M} \hybW^\T \hybW \in \R^{N \times N}$ is a positive definite matrix and $\bar{\hybW} \coloneqq  \hybW \U^{-1/2} \in \R^{M \times N}$ is well-defined. This modified weight matrix satisfies the desired semi-orthogonality:
	\begin{align}
		\bar{\hybW}^\T \bar{\hybW} &= \underbrace{(\U^{-1/2})^\T}_{= \U^{-1/2}} \underbrace{(\hybW^\T  \hybW) \U^{-1/2}}_{= \U^{-1/2} ( \hybW^\T  \hybW)} = (\U^{-1/2})^2 ( \hybW^\T  \hybW) \\
		&= \tfrac{M}{N} \cdot ( \hybW^\T  \hybW)^{-1} ( \hybW^\T  \hybW) = \tfrac{M}{N} \I{N}, \label{eq:extension:hybrid:semiortho}
	\end{align}
	where we have used that $\U^{-1/2}$ and $\hybW^\T \hybW$ commute.
	
	Adapting the notation from Remark~\ref{rmk:results:lifting}\ref{rmk:results:lifting:matrixversion}, we set $\alift_i \coloneqq [\a_i^1 \cdots \a_i^M] \in \R^{n \times M}$ and $\Ahyb_i \coloneqq \alift \hybW = [\ahyb_i^1 \cdots \ahyb_i^N] \in \R^{n \times N}$.
	Since $\hybW = \bar\hybW \U^{1/2}$, the Group-Lasso \eqref{eq:hybrid:lasso} can be rewritten as follows:\footnote{Note that if the corresponding minimizers are not unique, the following equalities are understood as identities of sets.}
	\begin{align}
		\solu[\X] &= \argmin_{\substack{[\x^1 \cdots \x^N] \\ \in \sset}} \ \tfrac{1}{2m} \sum_{i = 1}^m
		\Big(y_i - \sum_{k = 1}^N\sp{\ahyb_i^k}{\x^k}\Big)^2 \\
		&= \argmin_{\X \in \sset} \ \tfrac{1}{2m} \sum_{i = 1}^m (y_i - \sp{\Ahyb_i}{\X})^2 \\
		&= \argmin_{\X \in \sset} \ \tfrac{1}{2m} \sum_{i = 1}^m (y_i - \sp{\alift_i \hybW}{\X})^2 \\
		&= \argmin_{\X \in \sset} \ \tfrac{1}{2m} \sum_{i = 1}^m (y_i - \sp{\alift_i \bar\hybW \U^{1/2}}{\X})^2 \\
		&= \argmin_{\X \in \sset} \ \tfrac{1}{2m} \sum_{i = 1}^m (y_i - \sp{\alift_i \bar\hybW }{\underbrace{\X \U^{1/2}}_{\eqqcolon \bar{\X}}})^2 \\
		&= \Big[ \argmin_{\bar{\X} \in \sset \U^{1/2}} \ \tfrac{1}{2m} \sum_{i = 1}^m (y_i - \sp{\alift_i \bar\hybW }{\bar{\X}})^2 \Big] \cdot \U^{-1/2} \eqqcolon \solu[{\bar{\X}}] \cdot \U^{-1/2}.
	\end{align}
	Hence, applying the hybrid method of Algorithm~\ref{algo:hybrid} with weight matrix $\hybW$ and constraint set $\sset$ is actually equivalent to applying it with $\bar{\hybW} =  \hybW \U^{-1/2}$ and $\bar\sset \coloneqq \bar{\hybW} \coloneqq  \sset \U^{1/2}$.
	Due to \eqref{eq:extension:hybrid:semiortho}, we can apply Theorem~\ref{thm:extensions:hybrid} in the latter formulation, bounding the approximation error of $\solu[{\bar{\X}}]$.
	And since $\solu[{\bar{\X}}] = \solu[\X] \cdot \U^{1/2}$, this particularly implies a recovery result for the hybrid method with input $\hybW$ and $\sset$.
	We leave the formal details of this statement to the interested reader.
\end{remark}

Finally, let us illustrate the versatility of the hybrid method by some simple examples:
\begin{example}
\begin{rmklist}
\item\label{ex:extensions:hybrid:direct}
	\emph{The direct method.} Selecting $\hybW \coloneqq (1, \dots, 1) \in \R^{M \times 1}$ ($N= 1$) and $\sset \coloneqq R \ball[1][n]$, the direct method of Algorithm~\ref{algo:direct} coincides with Algorithm~\ref{algo:hybrid}.
	We observe that $\hybW^\T \hybW = M \I{1}$ and
	\begin{equation}
		\hybscalvec = \tfrac{1}{M} \hybW^\T \scalvec = \tfrac{1}{M} \sp{\hybW}{\scalvec} = \tfrac{1}{M} \sum_{j = 1}^M \scalfac_j = \bar\scalfac,
	\end{equation}
	which precisely corresponds to the mean scaling parameter introduced in Definition~\ref{def:results:scalparam}.
	And since
	\begin{equation}
		\meanwidth[1]{\cone{\sset}{\bar\scalfac\grtr}}^2 \lesssim s \cdot \log(\tfrac{2n}{s})
	\end{equation}
	by Proposition~\ref{prop:proofs:meanwidth}\ref{prop:proofs:meanwidth:direct} below, we can conclude that the statements of Theorem~\ref{thm:results:direct} ($R = \lnorm{\bar\scalfac\grtr}[1]$) and Theorem~\ref{thm:extensions:hybrid} are precisely the same under the hypothesis of Model~\ref{model:results:measurements}.
\item\label{ex:extensions:hybrid:lifting}
	\emph{The lifting method.} Choosing $\hybW \coloneqq \I{M} \in \R^{M \times M}$ ($N= M$) and $\sset \coloneqq \{ \X \in \R^{n\times M} \suchthat \lnorm{\X}[1,2] \leq \lnorm[\big]{\grtr\scalvec^\T}[1,2] = R \}$, Algorithm~\ref{algo:hybrid} equals the lifting approach of Algorithm~\ref{algo:lifting}.
	In this case, we have $\hybW^\T \hybW = \I{M}$ and $\hybscalvec = \tfrac{M}{M} \hybW^\T \scalvec = \scalvec$.
	Moreover, by Proposition~\ref{prop:proofs:meanwidth}\ref{prop:proofs:meanwidth:lifting},
	\begin{equation}
		\meanwidth[1]{\cone{\sset}{\grtr\scalvec^\T}}^2 \lesssim s \cdot \max\{M, \log(\tfrac{2n}{s})\},
	\end{equation}
	which implies that Theorem~\ref{thm:results:lifting} and Theorem~\ref{thm:extensions:hybrid} do also coincide under Model~\ref{model:results:measurements}.
\item\label{ex:extensions:hybrid:priorinfo}
	\emph{Incorporating prior knowledge.}
	Motivated by the example discussed at the beginning of Subsection~\ref{subsec:extensions:hybrid}, let us assume that $\sign(\scalfac_j) \neq 0$, $j = 1, \dots, M$ are available as prior information about the sensing model. Then, it is quite natural to apply the hybrid method with 
	\begin{equation}
		\hybW \coloneqq (\sign(\scalfac_1), \dots, \sign(\scalfac_M)) \in \R^{M \times 1},
	\end{equation}
	leading to the hybrid measurement vector $\ahyb \coloneqq \sum_{j = 1}^M \sign(\scalfac_j) \a^j$.
	Similarly to Example~\ref{ex:extensions:hybrid}\ref{ex:extensions:hybrid:direct}, we now have $\hybW^\T \hybW = M \I{1}$ and
	\begin{equation}
		\hybscalfac \coloneqq \hybscalvec = \tfrac{1}{M} \hybW^\T \scalvec = \tfrac{1}{M} \sum_{j = 1}^M \sign(\scalfac_j) \scalfac_j = \tfrac{1}{M} \sum_{j = 1}^M \abs{\scalfac_j}.
	\end{equation}
	Hence, Theorem~\ref{thm:extensions:hybrid} states that \eqref{eq:hybrid:lasso} recovers the vector $\hybscalfac\grtr$ under the hypothesis of Model~\ref{model:extensions:measurements}. In contrast, by Theorem~\ref{thm:results:direct}, the direct method with $\adir \coloneqq \sum_{j = 1}^M \a^j$ would just approximate $\bar\scalfac\grtr$, which is not meaningful if $\bar\scalfac = \tfrac{1}{M} \sum_{j = 1}^M \scalfac_j \approx 0$.
	\qedhere
\end{rmklist}
\label{ex:extensions:hybrid}
\end{example}

\subsection{Recovery Based on the Global Mean Width}
\label{subsec:extensions:hybridglobal}

The condition of \eqref{eq:extensions:hybrid:meas} shows that the conic mean width $\meanwidth[1]{\cone{\sset}{\grtr\hybscalvec^\T}}$ has a dramatic impact on the sampling rate of Theorem~\ref{thm:extensions:hybrid}.
Indeed, if $\grtr\hybscalvec^\T$ does not lie on the boundary of $\sset$, the cone $\cone{\sset}{\grtr\hybscalvec^\T}$ might simply equal $\R^{n \times N}$ so that (cf. \cite[Ex.~3.1\revision{(}a)]{genzel2016estimation})
\begin{equation}
	\meanwidth[1]{\cone{\sset}{\grtr\hybscalvec^\T}}^2 = \meanwidth[1]{\R^{n \times N}}^2 \asymp n \cdot N.
\end{equation}
This phenomenon is actually the reason why Theorem~\ref{thm:results:direct} and Theorem~\ref{thm:results:lifting} both rely on a ``perfect'' tuning of the sparsity parameter $R$.
Another drawback is that the mapping $\grtr \mapsto \meanwidth[1]{\cone{\lnorm{\grtr}[1]\ball[1][n]}{\grtr}}$ turns out to be discontinuous in the neighborhood of sparse vectors, which could be very problematic when dealing with \emph{compressible} source vectors.
As a way out, we now state a different version of Theorem~\ref{thm:extensions:hybrid} that is just based on the \emph{global mean width} of $\sset$:\footnote{For a formal definition of the global mean width, see Definition~\ref{def:proofs:meanwidth} below.}
\begin{theorem}[Hybrid Method -- Global Version]\label{thm:extensions:hybridglobal} 
	Assume that Model~\ref{model:extensions:measurements} holds true. Let $\hybW = [\hybw^1 \cdots \hybw^N] \in \R^{M \times N}$ be a weight matrix with $\hybW^\T \hybW = \tfrac{M}{N} \I{N}$. Moreover, suppose that $\grtr\hybscalvec^\T \in \sset$, where $\sset \subset \R^{n \times N}$ is a bounded\footnote{If the set $\sset$ is unbounded, we may have $\meanwidth{\sset} = \infty$ so that the statement of Theorem~\ref{thm:extensions:hybridglobal} becomes meaningless.} convex set.
	Then, there exists a numerical constant $C > 0$ such that the following holds true for every (fixed)
	$\delta \in (0,1]$ with probability at least $1 - 5 \exp(- C \cdot \subgparam^{-4} \cdot \delta^2 \cdot m )$: If
	\begin{equation}\label{eq:extensions:hybridglobal:meas}
		m \gtrsim \subgparam^{4} \cdot \delta^{-4} \cdot \meanwidth{\sset}^2,
	\end{equation}
	then any minimizer $[\solu^1 \cdots \solu^N] \in \R^{n\times N}$ of \eqref{eq:hybrid:lasso} satisfies
	\begin{equation}\label{eq:extensions:hybridglobal:bound}
		\Big(\tfrac{1}{N}\sum_{k = 1}^N\lnorm{\solu^k - \hybscalfac_k\grtr}^2 \Big)^{1/2} \lesssim \max\Big\{\tfrac{1}{\sqrt{N}}, \subgparam \cdot (\modeldev_\text{Hyb}^2 + \tfrac{\nu^2}{M})^{1/2} \Big\} \cdot \delta + \tfrac{N}{M} \cdot  \modelcovar_\text{Hyb},
	\end{equation} 
	where $\modeldev_\text{Hyb}^2$ and $\modelcovar_\text{Hyb}$ are given by \eqref{eq:extensions:hybrid:modeldev} and \eqref{eq:extensions:hybrid:modelcovar}, respectively.
\end{theorem}

The most striking difference to Theorem~\ref{thm:extensions:hybrid} is that sample-complexity condition of \eqref{eq:extensions:hybridglobal:meas} only involves the \emph{global} mean width.
Indeed, this particularly resolves the tuning issue (cf. Remark~\ref{rmk:results:direct}\ref{rmk:results:direct:tuning}), since the value of $\meanwidth{\sset}$ does not depend on $\grtr\hybscalvec^\T$ anymore and we just have to ensure that $\grtr\hybscalvec^\T$ is contained (``somewhere'') in $\sset$.
The price to pay for this simplification is the unusual oversampling factor of $\delta^{-4}$ in \eqref{eq:extensions:hybridglobal:meas}. 
In other words, the error decay in $m$ now just scales as $\asympfaster{m^{-1/4}}$, which is slower than the error rate $\asympfaster{m^{-1/2}}$ achieved by Theorem~\ref{thm:extensions:hybrid}.\footnote{For linear observations ($\fobs_j = \Id$), the error rate of Theorem~\ref{thm:extensions:hybridglobal} can be easily improved to $\asympfaster{m^{-1/2}}$, but as already mentioned, our focus is rather on the non-linear case which is more complicated.}
\revision{Moreover, the appearance of the additional term $1/\sqrt{N}$ on the right-hand side of \eqref{eq:extensions:hybridglobal:bound} implies that Theorem~\ref{thm:extensions:hybridglobal} is also suboptimal in the low-noise regime, i.e., $\modeldev_\text{Hyb}^2 + \nu^2 / M$ is small compared to $1 / N$. In particular, Theorem~\ref{thm:extensions:hybridglobal} does not guarantee exact recovery of $\grtr\hybscalvec^\T$ if $\modeldev_\text{Hyb} = \modelcovar_\text{Hyb} = \nu = 0$, while Theorem~\ref{thm:extensions:hybrid} would do.}

Let us finally give an example that emphasizes the main benefit of Theorem~\ref{thm:extensions:hybridglobal}, namely that the global mean width is oftentimes easier to control than its conic counterpart:
\begin{example}[Sparse representations in a dictionary]\label{ex:extensions:hybridglobal:dictionaries}
	We assume that $\sset$ takes the form $\sset = R \dict \ball[1][n']$, where $\dict = [\dictatom_{1} \cdots \dictatom_{n'}] \in \R^{n \times n'}$ is a \emph{dictionary}. Then, by \cite[Ex.~3.1\revision{(}b)]{genzel2016estimation}, we have the following bound on the mean width:
	\begin{align}
		\meanwidth{\sset}^2 = R^2 \cdot \meanwidth{\dict \ball[1][n']} = R^2 \cdot \meanwidth{\convhull\{\pm\dictatom_1, \dots, \pm\dictatom_{n'}\}}^2 \lesssim R^2 \cdot \max_{1 \leq l \leq n'} \lnorm{\dictatom_l}^2 \cdot \log(2 n'),
	\end{align}
	where $\convhull\{\pm\dictatom_1, \dots, \pm\dictatom_{n'}\}$ denotes the convex hull of the (signed) dictionary atoms of $\dict$.
	In particular, if $\grtr \in \sset$ possesses an $s$-sparse representation $\vec{c} \in \S^{n' - 1}$ (i.e., $\grtr = \dict \vec{c}$), we can conclude that $\meanwidth{\sset}^2$ can be essentially controlled by the sparsity $s$:
	\begin{equation}
		\lnorm{\vec{c}}[1] \leq \sqrt{\lnorm{\vec{c}}[0]} \cdot \lnorm{\vec{c}} \leq \sqrt{s} \cdot 1 = \sqrt{s}  \ (\approx R).
	\end{equation}
\end{example}

\subsection{Further Extensions}
\label{subsec:extensions:further}

The following list sketches several extensions that we did not include in our results for the sake of brevity.
Although this may require some technical preliminaries and additional assumptions, we regard these steps to be relatively straightforward.

\begin{itemize}
\item
	\emph{Adversarial noise.} The noise terms $e_1, \dots e_m \in \R$ in Model~\ref{model:extensions:measurements} do not have to be independent and could be even deterministic.
	This would lead to an extra additive term in the error bounds which takes the form
	\begin{equation}
		\tfrac{1}{\sqrt{M}} \cdot \Big(\tfrac{1}{m} \sum_{i = 1}^m e_i^2 \Big)^{1/2}.
	\end{equation}
	For further details, see also \cite{genzel2016estimation}.
\item
	\emph{Convex loss functions.} So far, we have only considered the squared loss for our recovery programs.
	But it could be also beneficial to use a different loss function in \eqref{eq:hybrid:lasso}.
	More precisely, we might replace \eqref{eq:hybrid:lasso} by
	\begin{equation}
		\min_{\substack{[\x^1\cdots\x^N] \\ \in \R^{n\times N}}} \ \tfrac{1}{m} \sum_{i = 1}^m
		\loss\Big(\sum_{k = 1}^N\sp{\ahyb_i^k}{\x^k}, y_i\Big) \quad
		\text{s.t. $[\x^1 \cdots \x^N] \in \sset$,}
	\end{equation}
	where $\loss \colon \R \times \R \to \R$ is a convex function.
	The squared loss then just corresponds to $\loss(v_1, v_2) = \tfrac{1}{2} (v_1 - v_2)^2$.
	Under relatively mild conditions on $\loss$, such as \emph{restricted strong convexity}, similar recovery guarantees as above can be proven; see again \cite{genzel2016estimation}.
\item
	\revision{\emph{Random non-linearities.} The non-linearities $\fobs_1, \dots, \fobs_M$ could be random functions which are independent of the measurement ensemble $\{\a^j\}_{j = 1}^M$;
	in other words, the output functions are allowed to change between different measurements in an i.i.d.\ random manner.
	As an example, they could take the form $\fobs_j = \xi \cdot \fobs$, where $\fobs \colon \R \to \R$ is a scalar function and $\xi$ is a $\pm 1$-valued random variable, modeling (independent) random sign-flips in every node and measurement step $i = 1, \dots, m$.
	Our results do literally hold true in this advanced situation, but note that the expected values then have to computed with respect to $\fobs_j$ as well.}
\item
	\revision{\emph{Anisotropic measurement vectors.} Instead of isotropy in Model~\ref{model:extensions:measurements}, one could just assume that $\mean[\a^j (\a^j)^\T] = \Covmatr_j$ for $j = 1, \dots, M$ and positive definite covariance matrices $\Covmatr_j \in \R^{n \times n}$, which can be even \emph{unknown}.
	Technically, this extension relies on the same argument as in Remark~\ref{rmk:extensions:hybrid:weightcondition}, with the difference that the unknown covariance matrices $\Covmatr_j$ are only implicitly used by the estimator \eqref{eq:hybrid:lasso} here.}
\item
	\revision{\emph{Dithering.} As pointed out in the discussion subsequent to Theorem~\ref{thm:extensions:hybrid}, a typical problem with sub-Gaussian measurement ensembles is the presence of a bias term $\modelcovar_\text{Hyb} > 0$, which may prevent \eqref{eq:hybrid:lasso} from being a consistent estimator.
	In fact, this shortcoming is not an artifact of our statistical analysis, but rather a limitation of the model setup, i.e., signal recovery from non-linear and non-Gaussian measurement is sometimes impossible; see also \cite{ai2014onebitsubgauss}.
	It therefore came as a surprise that in the situation of \emph{quantized compressed sensing}, a slight modification of the measurement process, based on \emph{dithering}, enables consistent estimation in the general sub-Gaussian case; see \cite{xu2018dithering,dirksen2018robust,thrampoulidis2018dithering} for recent advances in this direction.
	Being in line with the findings of \cite{thrampoulidis2018dithering}, one can also apply the abstract framework of Section~\ref{sec:proofs} to derive a recovery guarantee for the Lasso estimator with dithered $1$-bit observations; more specifically, an appropriate dithering step allows us to control the bias term (i.e., the mismatch covariance) by the sample size $m$, in such a way that it tends to $0$ as $m \to \infty$; see \cite[Subsec.~4.2.2]{genzel2019thesis} for details.
	However, it is worth pointing out that all these results are restricted to the case of $M = 1$ (i.e., single-index models) and do only consider specific types on non-linearities (uniform and $1$-bit quantizers). But we still believe that dithering could be incorporated into the more general setup of this work at least up to a certain extent, although this is expected to be more challenging than the previously mentioned extensions.}
\end{itemize}

\section{Practical Applications and Numerics}
\label{sec:practice}

In this part, we return to the prototypical example of \emph{wireless sensor networks}, which was already sketched in the introduction and has served as motivation for our superposition model.
In particular, we discuss the theoretical findings of the previous sections with respect to this specific application and provide some numerical simulations.

\subsection{Wireless Sensor Networks}
\label{subsec:practical:networks}

Distributed sparse parameter estimation using wireless sensor
networks is a promising approach to many environmental monitoring
problems \cite{Akyildiz:2002} and forms a natural application of compressed sensing
\cite{Bajwa2006,Luo2009,Cao2013}.  In fact, such network architectures have advantages over
conventional sensing technologies in terms of costs, coverage,
redundancy, and reliability.  Typical applications are structural
health monitoring, medical sensor solutions, traffic monitoring as
well as warning systems for heat, fire, seismic activities, or
meteorologic disturbances.  While several communication standards,
embedded platforms, and operating systems are available for this
problem settings (e.g., TinyOS and IEEE 802.15.4), some of the
inherent limitations of these transceiver designs are low transmission
and computing power due to battery saving. 
It is therefore important to devise approaches to recovery under such
non-ideal conditions (e.g., see ``Dirty RF'' \cite{Fettweis2007}).

As prototypical setup, we may consider a model
situation where multiple \emph{sensor nodes} perform individual measurements
on the same source.  For example, each sensor reading could correspond to a
spatial sample of a temperature field in a building or measurements of
the water flow and quality taken at different locations.  The
fluctuation of these quantities are typically specified by only a
small number of active parameters which can be often modeled as a
\emph{sparse vector} $\grtr \in \R^n$ in a known transform domain (e.g.,
Fourier or wavelets).  The task of the wireless sensor network is now
to communicate $\grtr$ to a central \emph{fusion center} in an \emph{ad hoc}
and \emph{autonomous} manner, bypassing an additional resource
and time overhead.  During the $i$-th communication step, all
sensor nodes $j = 1, \dots, M$ transmit their measured data $\sp{\a_i^j}{\grtr}$.  But due to the
overall processing, the low-quality hardware components as well as the
wireless channel, the $j$-th transceiver node effectively only contributes
a distorted signal $\fobs_j(\sp{\a_i^j}{\grtr})$. Since the transmission procedure takes place within a shared wireless medium, this finally leads to a superposition of non-linear signals at the central receiver; see Figure~\ref{fig:sensornetwork} for an illustration. 
Mathematically, we precisely end up with our initial model from \eqref{eq:intro:measurements}:
\begin{equation}
	y_i = \sum_{j = 1}^M \fobs_j (\sp{\a_i^j}{\grtr}) + e_i, \qquad i = 1, \dots, m.
\end{equation}

\subsubsection*{The Non-Linear Functions $\fobs_1, \dots, \fobs_M$}

We have not specified yet how the non-linearities $\fobs_1, \dots, \fobs_M$ usually look like in practice.
It was already pointed out that the main purpose of these functions is to capture the common effects caused by the wireless channel and hardware imperfections.
For example, the latter issue is highly relevant to low-cost transceiver nodes whose radio-frequency (RF) components only provide very low signal quality.
Here, severe degradations are caused by phase noise and non-linear distortions, such as ADC impairments or IQ imbalances. 
As an illustration, let us briefly discuss two typical phenomena that often arise in applications:
\begin{itemize}
\item
	\emph{Power amplifiers.} An important type of disturbance is caused by the non-linear characteristics of low-cost amplifiers used at the nodes; see \cite{Rapp1991} for a widely used model. In the extreme case, this leads to a clipping at a certain \emph{amplitude level} (threshold) $A > 0$:
	\begin{equation}\label{eq:practice:clip}
		\fobs^{(A)}(v) \coloneqq\sign(v)\cdot\min\{\abs{v}, A\}.
	\end{equation}
	While the sign (phase) of the signal is still preserved in this generic model, the amplitude undergoes a (data-dependent) deformation.
\item
	\emph{Wireless channel.} In a realistic setup, each node modulates its sensor readings on particular waveforms, propagating through the wireless channel after amplification. A filtering and sampling step is then performed at the central receiver.  As a simple model, we may assume that the effective channel is approximately constant over the entire communication period, including all transceiver operations. Formally, this corresponds to a scalar multiplication
	\begin{equation}\label{eq:practice:channelmultiplication}
		\fobs^{(h)}(v) \coloneqq h \cdot v,
	\end{equation}
	where $h \in \R$ is the \emph{channel coefficient}, which might be unknown a priori.  If the individual channel configuration of a sensor device is approximately known, one may determine the sign (phase) of $h$ and consider $\fobs^{(\abs{h})}(v)=\abs{h}\cdot v$ instead.  A common approach to achieve such a sign-compensation is the concept of \emph{channel reciprocity} for narrow-band time-division multiplexing transmission.  Here, pilot signals are periodically broadcasted by the receiver to all sensors (simultaneously).  Each node is now capable of estimating its individual coefficient in the reverse direction (downlink), and in that way, also approximating the parameter $h$. However, due to limited transmission power, this step usually only allows for specifying the sign of $h$ and not its magnitude.
\end{itemize}
Therefore, put together, each contribution to the superposition of
\eqref{eq:intro:measurements} could be modeled by a function of the
form $\fobs_j = \fobs^{(h_j)}\circ \fobs^{(A)}$.
These types of non-linearities will be also used for our numerical experiments below in Subsection~\ref{subsec:practice:experiments}.
However, this is still just a simplified example of distortion, since the channel may be outdated in many applications and further disturbances could be present, like an oscillator mismatch or phase noise.

\subsubsection*{Coherent vs. Non-Coherent Communication}

If all sensors use a common, synchronized clock (cf. \cite{Brown05}) and possess sufficient knowledge about the wireless channel, a \emph{coherent cooperative transmission} can be achieved, e.g., by channel reciprocity.
This essentially means that the phases (signs) of the individually communicated signals $\sp{\a_i^j}{\grtr}$ are preserved by the functions $\fobs_j$.
In other words, (most of) the scaling parameters $\scalfac_1, \dots, \scalfac_M$ (cf. Definition~\ref{def:results:scalparam}) are of the same sign and we can expect that
\begin{equation}
	\abs{\bar\scalfac} = \tfrac{1}{M} \abs[\Big]{\sum_{j = 1}^M \scalfac_j} \approx \tfrac{1}{M} \sum_{j = 1}^M \abs{\scalfac_j} \gg 0.
\end{equation}
With such a pre-compensation, the signals add up \emph{coherently} at the receiver, so that the strategy of the direct method (Algorithm~\ref{algo:direct}), i.e., computing ``aligned'' superimposed vectors $\adir_i = \sum_{j = 1}^M \a_i^j$, turns out to be very natural.
Moreover, due to the independent channel conditions, the probability of outage can be significantly reduced in that way, which is known as \emph{cooperative} or \emph{multiuser diversity} in communication engineering.
The statement of Theorem~\ref{thm:results:direct} shows that Algorithm~\ref{algo:direct} is indeed very appropriate in the this scenario, since already $\asympfaster{s\log(2n/s)}$ superimposed measurements are sufficient for recovery.
Thus, the Lasso estimator \eqref{eq:direct:lasso} basically \emph{achieves the same sampling rate as in the linear case}. 
Our main results even suggest that the direct method should be preferred to the lifting method (Algorithm~\ref{algo:lifting}) in the situation of coherent transmission, since the latter one is computationally more challenging and requires more measurements. 

On the other hand, Theorem~\ref{thm:results:direct} indicates that successful source estimation heavily relies on the assumption of $\abs{\bar\scalfac} \gg 0$.
Unfortunately, such a coherent setup is often difficult to implement in practice, especially for ad-hoc wireless sensor networks.  
For this reason, there also exists a research branch investigating \emph{non-coherent cooperative transmission}, although most works rather focus on achieving higher power gains at the receiver (e.g., see \cite{Scaglione03}).
The generic task of non-coherent network architectures is to estimate the structured data $\grtr$ in the presence of unknown channel coefficients (cf. \eqref{eq:practice:channelmultiplication}), implying that the (signs of the) sensor parameters $\scalfac_j$ are unknown.
Note that, for ideal transmitters, i.e., $A\to\infty$ in \eqref{eq:practice:clip}, this challenge is actually an instance of a \emph{bilinear inverse problem} for blind sensor calibration (see \eqref{eq:intro:bilinearmodel}).
While the direct method may completely fail in these general situations, the lifting approach of Algorithm~\ref{algo:lifting} now becomes advantageous. \revision{Theorem~\ref{thm:results:lifting} shows that, with $m = \asympfaster{s \cdot \max\{M, \log(2n / s)\}}$ samples, \eqref{eq:lifting:l1l2lasso} does not only allow for recovery of the source vector $\grtr$ but even allows for an estimate of the unknown scaling factors $\scalfac_1, \dots, \scalfac_M$ (cf. \eqref{eq:results:lifting:bound:svd}).} Hence, to a certain extent, \eqref{eq:lifting:l1l2lasso} enables us to ``learn'' the underlying system configuration. Once such a ``calibration step'' has been performed, one may continue using the hybrid approach of Algorithm~\ref{algo:hybrid}, tuned by our additional information on the sensor environment (see Example~\ref{ex:extensions:hybrid}\ref{ex:extensions:hybrid:priorinfo}).

%

\subsection{Numerical Experiments}
\label{subsec:practice:experiments}

In this part, we validate our recovery approaches by several numerical simulations. 
\revision{For this purpose, we have generated normalized $s$-sparse random
vectors $\grtr\in\R^n$ for two different parameter regimes,  $(s,n)=(4,64)$ and $(s,n)=(10,256)$,}
where each of the non-zero entries was drawn from an independent Gaussian random variable.
Following Model~\ref{model:results:measurements}, every node $j =
1,\dots, M$ performs $i = 1,\dots, m$ measurements of $\grtr$ with
i.i.d.\ standard Gaussian vectors $\a_i^j \distributed
\Normdistr{\vnull}{\I{n}}$. 
Figure~\ref{fig:practice:txsignal} visualizes the signal vectors $[\fobs_j(\sp{\a_i^j}{\grtr})]_{1\leq i \leq m}$ of two exemplary nodes for $m = 128$. Here, each of the taken measurements (red) undergoes a clipping \eqref{eq:practice:clip} with threshold $A=1$ (green) and is then transmitted into the channel,
which corresponds to a scalar multiplication according to \eqref{eq:practice:channelmultiplication} with coefficients 
$h_j\distributed \Normdistr{0}{1}$  (blue).
\begin{figure}
  \centering
  \includegraphics[width=.8\linewidth]{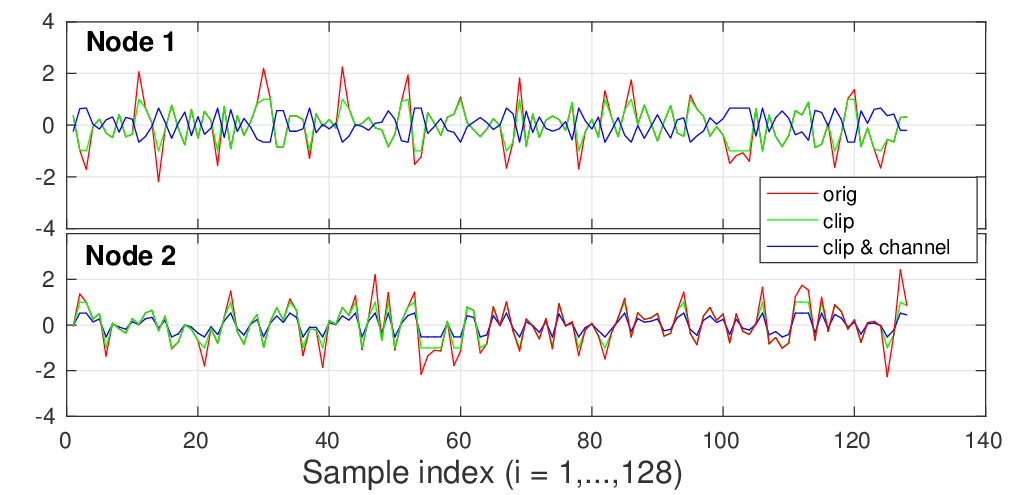}
  \caption{Transmitted signals of two different nodes: The
    original signal (red), clipped channel input (green), and clipped
    channel output (blue). Note that the first node was affected by a
    sign-flip, i.e., $\sign(h_j) = -1$.  }
  \label{fig:practice:txsignal}
\end{figure}

\emph{The impact of non-linear distortions.}
The plot of Figure~\ref{fig:practice:clio} demonstrates the performance of the direct method (Algorithm~\ref{algo:direct}) with respect to a ``softcut'' (clip) non-linearity \eqref{eq:practice:clip}, i.e., $\fobs_j=\fobs^{(A)}$ for $j = 1, \dots, M$.
Not very surprisingly, the reconstruction becomes more accurate as $A$ grows, since the $\fobs_j$ are then ``closer to be linear.''
The horizontal distances between the single curves are of particular interest here, since they determine the number of extra (sensor) nodes required to achieve the same recovery performance.  This can be of considerable practical relevance if the price of low-power devices (small $A$) is significantly lower than the one of high-quality devices (large $A$).
\begin{figure}
  \centering
  \includegraphics[width=1\linewidth]{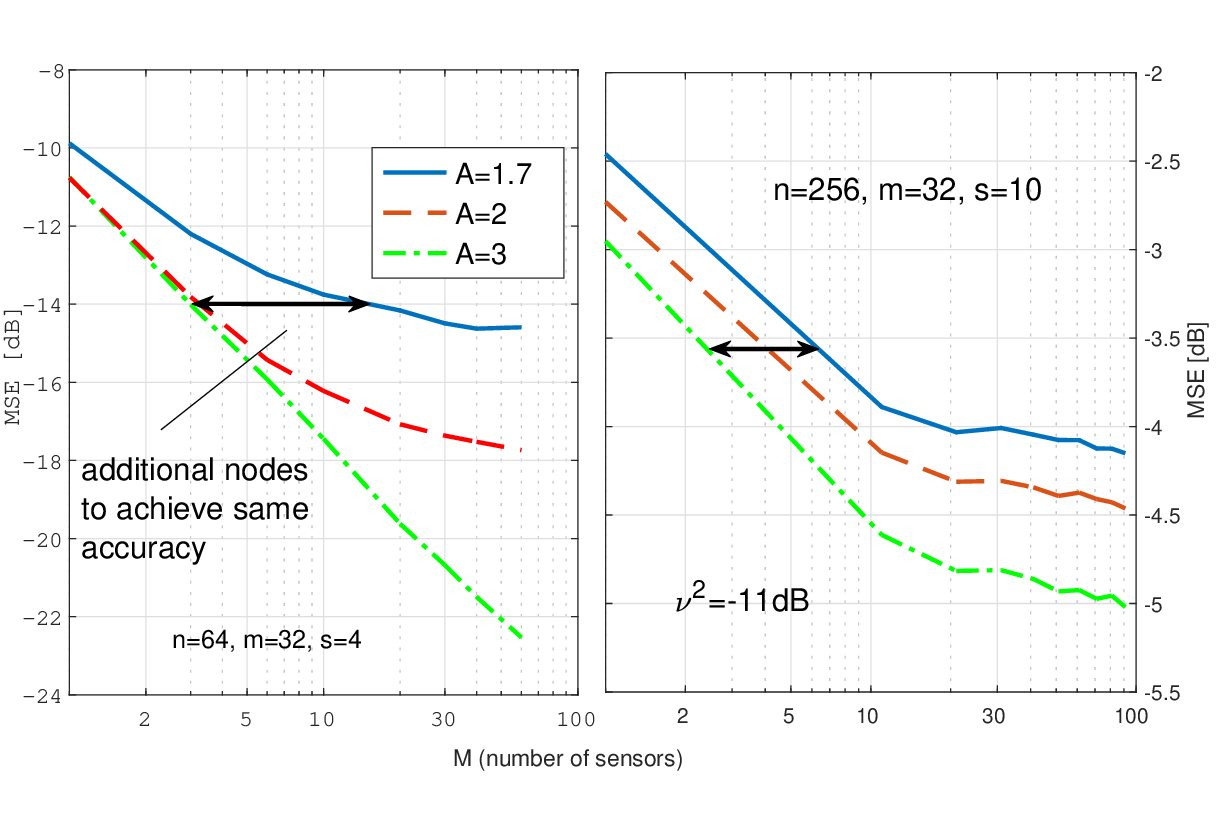}
  \caption{\revision{Recovery of $s$-sparse vectors in $\R^n$
      from $m = 32$ noisy measurements with a ``softcut''
      non-linearity $\fobs_j=\fobs^{(A)}$
      (cf. \eqref{eq:practice:clip}) for $A\in\{1.7, 2,3\}$ and
      $\a_i^j\distributed \Normdistr{\vnull}{\I{n}}$ being standard
      Gaussian vectors. The additive Gaussian noise has variance $\nu^2=-11$dB.
      The plots show the mean squared error (MSE) of
      the reconstruction via \eqref{eq:direct:lasso} for $(s,n)=(4,64)$ (left) and $(s,n)=(10,256)$ (right).  }}
  \label{fig:practice:clio}
\end{figure}
\revision{Figure~\ref{fig:practice:clio2} demonstrates the difference in recovery performance for $(s,n) = (4,64)$ when the sensor readings are obtained from Gaussian and Bernoulli measurements, respectively. As expected, the Bernoulli setup is more robust against the distortions caused by a ``softcut'' non-linearity, since these measurements have a smaller tail in their amplitudes compared to the Gaussian distribution.}
\begin{figure}
  \centering
  \includegraphics[width=.8\linewidth]{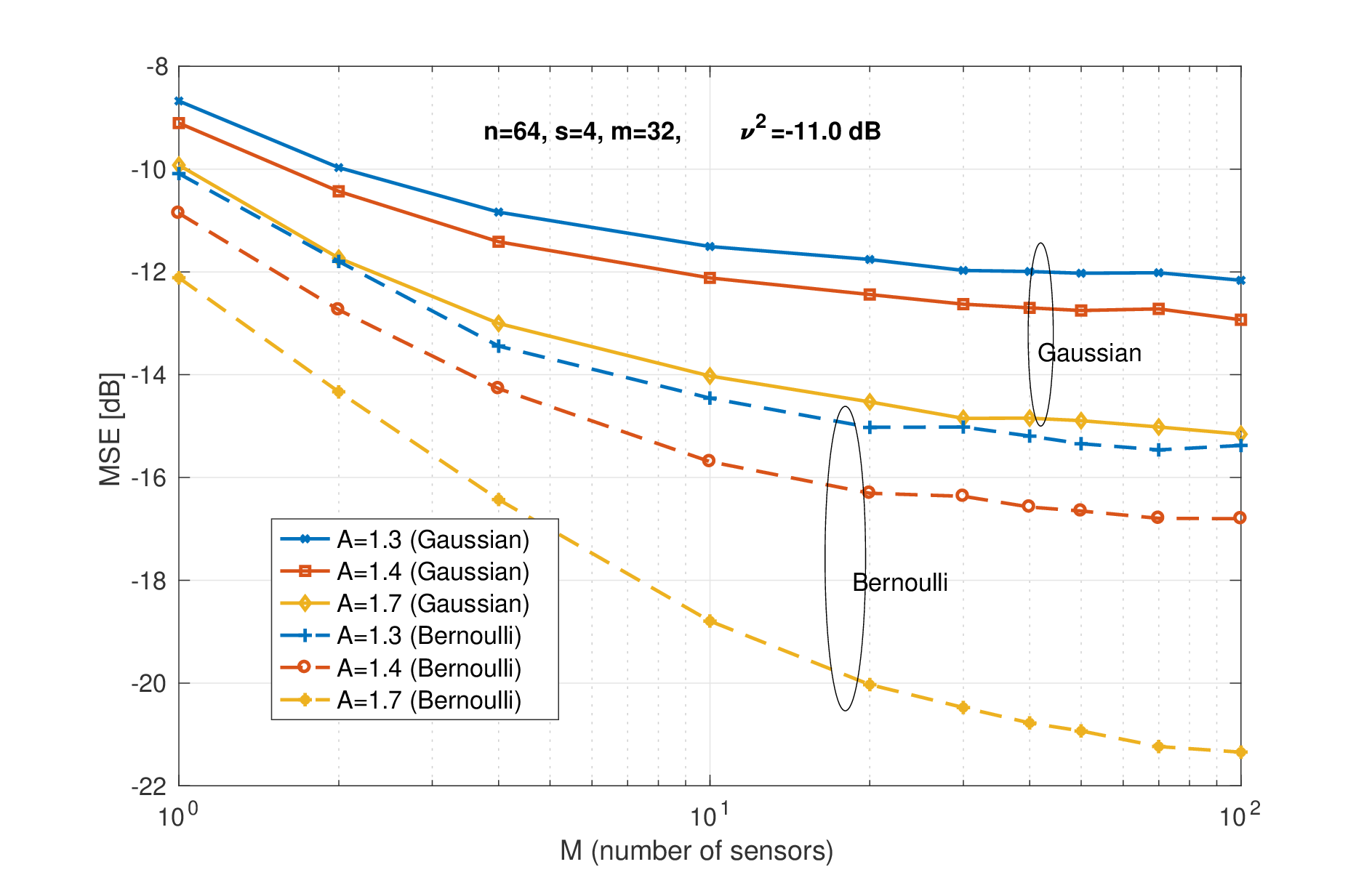}
  \caption{\revision{Recovery of $4$-sparse vectors in $\R^{64}$ from $m = 32$ noisy measurements. 
  		Here, a ``softcut'' non-linearity $\fobs_j=\fobs^{(A)}$ (cf. \eqref{eq:practice:clip}) for $A\in\{1.3, 1.4, 1.7\}$ has been applied to the sensor readings, where the measurement vectors $\a_i^j$ are either standard Gaussian vectors or i.i.d.\ Bernoulli vectors (symmetric $\pm1$-valued).  
  		The additive Gaussian noise has variance $\nu^2=-11$dB.
  		The plots show the mean squared error (MSE) of the reconstruction via \eqref{eq:direct:lasso}.}}
  \label{fig:practice:clio2}
\end{figure}

\emph{Coherent vs. non-coherent communication.}  For testing coherent transmission, we assume that
the phases of the individual channel coefficients $h_1, \dots, h_M$
have been already resolved, as discussed in the course of
\eqref{eq:practice:channelmultiplication}. Hence, we consider
$\fobs_j(v)=\abs{h_j} \cdot \fobs^{(A)}(v)$ as non-linearities with
$A = 1$ here. The recovery is then performed by solving
\eqref{eq:direct:lasso} with $R=\bar\scalfac \cdot \sqrt{s}$ and
rescaling the minimizer by
$1/\bar\scalfac$. 
In the non-coherent setting, we just use
$\fobs_j(v)=h_j\cdot \fobs^{(A)}(v)$ and apply
\eqref{eq:lifting:l1l2lasso} with $R= \sqrt{M} \cdot \sqrt{s}$ for
retrieval.  The numerical results are shown in
Figure~\ref{fig:practice:directvslifting} for different values of $m$
and $M$. \revision{For the direct method (red), the MSE decreases as $M$
grows.  It becomes in fact almost constant for large $M$, which
is in line with the observation that at some point, the model variance
$\modeldev_\text{Dir}^2$ dominates the noise term $\nu^2 / M$ in the
error bound of \eqref{eq:results:direct:bound} in
Theorem~\ref{thm:results:direct}.  In other words, enlarging a \emph{coherently communicating} sensor network can improve the signal-to-noise ratio of the
underlying measurement process, but as soon as $\nu^2 / M \approx \modeldev_\text{Dir}^2$ adding further sensor devices does not bring any more gain.} On the contrary, there is obviously a
``turning point'' when using the lifting method (blue). For
sufficiently small node counts, the recovery error indeed drops with
$M$ up to a certain level. Above this threshold, more measurements are
required to achieve the same accuracy.  This behavior is precisely
reflected by the statement of Theorem~\ref{thm:results:lifting}, which
indicates that $M = \asympfaster{\log(2n/s)}$ is the ``ideal'' size of
a network.

\revision{Finally, it is worth pointing out that increasing $m$ improves the performance of both the direct and lifting method in Figure~\ref{fig:practice:directvslifting}, which is also consistent with the respective guarantees of Theorem~\ref{thm:results:direct} and Theorem~\ref{thm:results:lifting}; indeed, by adjusting the parameter $\delta$, the error bounds of \eqref{eq:results:direct:bound} and \eqref{eq:results:lifting:bound} certify a decay rate of $\asympfaster{m^{-1/2}}$. Thus, we can conclude that taking more measurements at each node decreases the variance of the estimators (regardless of non-linear distortions), while increasing the number of nodes is only beneficial up to a certain point (depending on the specific model configuration). In practice, this implies a trade-off between increasing $m$ and $M$ which is controlled by the additional costs of each increase.}
\begin{figure}
  \centering
  \includegraphics[width=.8\linewidth]{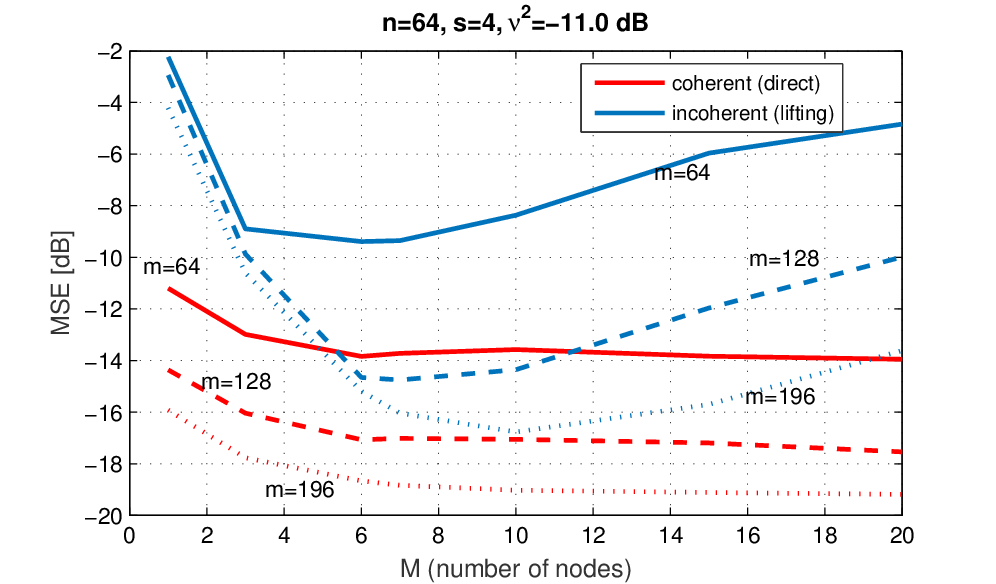}
  \caption{\revision{Recovery of $4$-sparse vectors in $\R^{64}$ from $m \in \{64, 128, 196\}$ measurements. The plots show the mean squared error (MSE) of the reconstruction with coherent (red) and non-coherent transmission (blue) via \eqref{eq:direct:lasso} and \eqref{eq:lifting:l1l2lasso}, respectively.}}
  \label{fig:practice:directvslifting}
\end{figure}

\section{Conclusion and Outlook}
\label{sec:conclusion}

Regarding our initial problem statements from Section~\ref{sec:intro}, we have shown that recovery from superimposed, non-linearly distorted observations \eqref{eq:intro:measurements} is already feasible by applying simple convex estimators of Lasso-type, which do not require any additional domain knowledge.
The statistical analysis of our main results, Theorem~\ref{thm:results:direct} and Theorem~\ref{thm:results:lifting}, provides detailed answers to the issues of \ref{question:intro:samplecomplexity}--\ref{question:intro:nodecount}, indicating that each of our two approaches comes along with its specific up- and downsides:
While the direct method (Algorithm~\ref{algo:direct}) has advantages with respect to sample complexity and efficiency, the lifting method (Algorithm~\ref{algo:lifting}) can handle more challenging model situations, such as blind sensor calibration.
From a practical perspective, these findings may be relevant to the application of wireless sensor networks, for which these points can be translated into economical aspects like quality, efficiency, and time budget.
\revision{Apart from that, it has turned out in Section~\ref{sec:extensions} that the scope of our results is not just limited to sparse source estimation from Gaussian observations, but can be extended to sub-Gaussian distributions (with the price of an asymptotic bias of the estimators), general convex constraints as well as more involved (hybrid) measurement designs.}

There are however several important open problems which could be investigated in the future works.
The following list sketches some potential improvements of our results that we consider to be of particular interest:
\begin{itemize}
\item 
	\emph{Exploiting prior knowledge.}
	The hybrid method of Algorithm~\ref{algo:hybrid} provides a lot of flexibility in setting up measurement ensembles and structural constraints.
	We gave a brief illustration in Example~\ref{ex:extensions:hybrid} and Example~\ref{ex:extensions:hybridglobal:dictionaries}, but our approach clearly offers much more possibilities.
	For example, it might be interesting to explore the following questions: Given limited information on the non-linearities $\fobs_1, \dots, \fobs_M$, what is the optimal choice of the weight matrix $\hybW$ for Algorithm~\ref{algo:hybrid}?
	Or how to choose $\hybW$ and $\sset$ if the number of active nodes ($\fobs_j \neq 0$) is much smaller than $M$?
\item 
	\emph{Structured measurements.} This point is very important from a practical perspective because the measurement vectors do not necessarily obey a sub-Gaussian distribution (with small $\subgparam$). Unfortunately, there are only a very few things known for the case of non-linear observations.
	One of the main difficulties is that our proofs and auxiliary results do heavily rely on statistical properties of sub-Gaussian distributions that fail to hold true for structured measurements like random Fourier samples.
\item 
	\emph{Breaking the multiplicative complexity barrier.}
	We have already mentioned in the course of Theorem~\ref{thm:results:lifting} that the multiplicative sampling rate $\asympfaster{s \cdot M}$ of the Group-Lasso \eqref{eq:lifting:l1l2lasso} is sub-optimal.
	Hence, we wish to come up with an algorithmic approach that is capable of recovering both $\grtr$ and $\scalvec$ but only requires $\asympfaster{s + M}$ distributed measurements.\footnote{This rate would precisely correspond to the degrees of freedom of $\grtr$ and $\scalvec$.}
	In fact, there has been recent progress in low-rank matrix factorization successfully tackling this bottleneck, e.g., see \cite{lee2016powerfac,richard2014tight}.
	But these approaches are still limited to the linear case of \eqref{eq:intro:bilinearmodel} and suffer either from restrictive model assumptions or involve computationally challenging (NP-hard) steps.
	Thus, in our non-linear setup, achieving the optimal additive sampling rate remains a big open problem.
\item 
	\emph{General distributed observations.}
	In this work, we have only studied the case where the fusion function $F$ in \eqref{eq:intro:distrmeas} corresponds to computing a sum.
	But there are clearly more interesting examples that may arise in practical applications, for instance, see \cite{Goldenbaum13}.
	The abstract statement of Theorem~\ref{thm:proofs:abstractrecovery} might be also useful in these general situations, since it does not make any restrictions on the actual observation model.
\end{itemize}

\section{Proofs of the Main Results}
\label{sec:proofs}

\subsection{Proofs of Theorem~\ref{thm:results:direct} and Theorem~\ref{thm:results:lifting}}
\label{subsec:proofs:mainresults}

The key idea of \eqref{eq:direct:lasso} and
\eqref{eq:lifting:l1l2lasso} (see Algorithm~\ref{algo:direct} and Algorithm~\ref{algo:lifting}, respectively) is to fit \emph{non-linear}
measurements $y_1, \dots, y_m$ by an appropriate \emph{linear}
counterpart.  This strategy is in fact a specific instance of a more
general approach, which is based on the so-called \emph{$\sset$-Lasso} with arbitrary observation rules:
\begin{equation}\label{eq:proofs:klasso} \tag{$P_\sset$}
  \min_{\xgen \in \R^d} \tfrac{1}{2m} \sum_{i = 1}^m (y_i - \sp{\agen_i}{\xgen})^2 \quad \text{s.t. $\xgen \in \sset$.}
\end{equation}
Here, $\agen_1, \dots, \agen_m \in \R^d$ are again certain
measurement vectors, while the convex \emph{constraint set}
$\sset \subset \R^d$ imposes structural assumptions on the solution.
For the direct method, we have chosen
$\agen_i \coloneqq \adir_i = \sum_{j = 1}^M \a_i^j \in \R^n$ and $\sset$
equals a rescaled $\l{1}$-unit ball ($d = n$), whereas for the lifting
method,
$\agen_i \coloneqq \alift_i = [\a_i^1 \cdots \a_i^M] \in \R^{n\times M}$ and
$\sset$ is a rescaled $\l{1,2}$-unit ball ($d = n \cdot M$, cf. Remark~\ref{rmk:results:lifting}\ref{rmk:results:lifting:matrixversion}).

A major challenge in the abstract setup of \eqref{eq:proofs:klasso} is to establish a relationship between a minimizer of \eqref{eq:proofs:klasso} and the underlying observations $y_1, \dots, y_m$.
For this purpose, let us first fix a general model:
\begin{model}[General Observations] \label{model:proofs:generalobs}
	Let $\{(\agen_i, y_i)\}_{i = 1}^m$ be independent samples of a joint random pair $(\agen, y) \in \R^d \times \R$, where $\agen$ is an isotropic, mean-zero sub-Gaussian random vector in $\R^d$ with $\normsubg{\agen} \leq \subgparam$ for some $\subgparam > 0$.
\end{model}
Our main goal is now to specify a linear mapping $\agen \mapsto \sp{\agen}{\grtrgen}$ with a certain $\grtrgen \in \R^d$ that ``mimics'' the observation variable $y$ as well as possible. In order to make this approach more precise, we introduce the following two quantities:
\begin{definition}\label{def:proofs:modelmismatch}
	Let $\grtrgen \in \R^d$ be a vector. Under the hypothesis of Model~\ref{model:proofs:generalobs} we define the \emph{mismatch covariance} as
	\begin{equation}\label{eq:modelmismatch:covariance}
		\modelcovar(\grtrgen) \coloneqq \modelcovar(\grtrgen; \agen, y) \coloneqq \lnorm{\mean{}_{(\agen,y)}[(\sp{\agen}{\grtrgen} - y) \agen]}.
	\end{equation}
	and \emph{mismatch deviation}\footnote{We implicitly assume that $y$ is sub-Gaussian here, implying that $\modeldev(\grtrgen) < \infty$.}
	\begin{equation}
		\modeldev(\grtrgen) \coloneqq \modeldev(\grtrgen; \agen, y) \coloneqq \normsubg{\sp{\agen}{\grtrgen} - y}.
	\end{equation}
\end{definition}
The purpose of these two parameters is to quantify the \emph{mismatch} that results from approximating non-linear observations by a linear model.
Intuitively, $\modelcovar(\grtrgen)$ measures the covariance between the mismatch term $\sp{\agen}{\grtrgen} - y$ and the measurement vector $\agen$, whereas $\modeldev(\grtrgen)$ essentially captures its deviation from zero.

The impact of the constraint set $\sset$, which forms the second important ingredient of \eqref{eq:proofs:klasso}, can be handled with the well-known concept of Gaussian mean width:
\begin{definition}\label{def:proofs:meanwidth}
Let $\ssetalt \subset \R^d$ be a non-empty subset.
\begin{deflist}
\item
	The \emph{(global) mean width} of $\ssetalt$ is given by
	\begin{equation}
		\meanwidth{\ssetalt} \coloneqq \mean_{\gaussian}{}[\sup_{\hgen \in \ssetalt} \sp{\gaussian}{\hgen}],
	\end{equation}
	where $\gaussian \distributed \Normdistr{\vec{0}}{\I{d}}$ is a standard Gaussian random vector.
\item
	The \emph{local mean width} of $\ssetalt$ at scale $t > 0$ is defined as
	\begin{equation}
		\meanwidth[t]{\ssetalt} \coloneqq \meanwidth{\ssetalt \intersec t \S^{d-1}}.
	\end{equation}
	Moreover, for $\xgen \in \R^d$, we call $\meanwidth[1]{\cone{\ssetalt}{\xgen}}$ the \emph{conic mean width} of $\ssetalt$ at $\x$.\footnote{Recall that $\cone{\ssetalt}{\xgen}$ denotes the cone of $\ssetalt$ at $\xgen$; see \eqref{eq:intro:notation:cone}.}
\end{deflist}
\end{definition}
Note that these parameters originate from the field of geometric functional analysis (e.g., see \cite{giannopoulos2004asymptotic,gordon1985gaussian,gordon1988escape}). For a more extensive discussion of their basic properties and their role in signal estimation problems, the reader is referred to \cite{bartlett2003complexity,mendelson2007subgaussian,chandrasekaran2012geometry,plan2014highdim,vershynin2014estimation,plan2013robust, RomanHDP}. 
Moreover, we would like to mention that the conic mean width is very closely related to the notion of \emph{statistical dimension} \cite{amelunxen2014edge}, which is also widely used in the literature.

We are now ready to formulate an abstract recovery guarantee that gives a quite general answer to what the $\sset$-Lasso \eqref{eq:proofs:klasso} is doing to non-linear observations. In fact, this result will form the basis of proving Theorem~\ref{thm:results:direct} and Theorem~\ref{thm:results:lifting}. 
\begin{theorem}\label{thm:proofs:abstractrecovery}
	We assume that Model~\ref{model:proofs:generalobs} holds true.
	Let $\sset \subset \R^d$ be a convex subset and let $\grtrgen \in \sset$ be an arbitrary vector.
	There exist numerical constants $C, C', C'' > 0$ such that the following holds true for every (fixed) $\delta \in \intvopcl{0}{1}$ with probability at least $1 - 5 \exp(-C \cdot \subgparam^{-4} \cdot \delta^{2} \cdot m)$:
	If the number of observations obeys
	\begin{equation}\label{eq:proofs:abstractrecovery:meas}
		m \geq C' \cdot \subgparam^4 \cdot \delta^{-2} \cdot \meanwidth[1]{\cone{\sset}{\grtrgen}}^2,
	\end{equation}
	then any minimizer $\hat{\xgen}$ of \eqref{eq:proofs:klasso} satisfies
	\begin{equation}\label{eq:proofs:abstractrecovery:bound}
		\lnorm{\hat{\xgen} - \grtrgen} \leq C'' \cdot \Big( \subgparam^{-1} \cdot \modeldev(\grtrgen) \cdot \delta + \modelcovar(\grtrgen) \Big).
	\end{equation}
	\revision{In particular, if $\modeldev(\grtrgen) = \modelcovar(\grtrgen) = 0$ (which holds if and only if $y = \sp{\a}{\grtrgen}$), the error bound \eqref{eq:proofs:abstractrecovery:bound} yields exact recovery, i.e., $\hat{\xgen} = \grtrgen$.}
\end{theorem}
The proof of Theorem~\ref{thm:proofs:abstractrecovery} is given in Appendix~\ref{subsec:app:abstractrecovery}.
At first sight, it is somewhat surprising that Theorem~\ref{thm:proofs:abstractrecovery}
is valid for \emph{every} choice of $\grtrgen$.  However, in order to
turn the error bound of \eqref{eq:proofs:abstractrecovery:bound} into a meaningful statement, one needs to ensure that the two mismatch parameters $\modelcovar(\grtrgen)$ and $\modeldev(\grtrgen)$ are sufficiently small.
If the ``ansatz-vector'' $\grtrgen$ can be chosen in such a way, Theorem~\ref{thm:proofs:abstractrecovery} states that any output of the Lasso \eqref{eq:proofs:klasso} indeed forms a reliable estimate of $\grtrgen$. 

\revision{Interpreting Theorem~\ref{thm:proofs:abstractrecovery} statistically, the mismatch deviation $\modeldev(\grtrgen)$ is associated with the \emph{variance} of the estimator \eqref{eq:proofs:klasso}, controlling the size of the first term in the error bound \eqref{eq:proofs:abstractrecovery:bound}; most notably, this term can be made arbitrarily small by adjusting $\delta$ and $m$.
The mismatch covariance $\modelcovar(\grtrgen)$, in contrast, plays the role of an \emph{(asymptotic) bias} that particularly specifies whether \eqref{eq:proofs:klasso} is consistent or not.
We refer the interested reader to \cite{genzel2017mismatch} and \cite[Chap.~3~\&~4]{genzel2019thesis} for a more detailed discussion of the mismatch parameters and their role in learning semi-parametric observations models.}

The following proposition shows that we can even achieve $\modelcovar(\grtrgen) = 0$ under the hypothesis of Model~\ref{eq:results:measurements:meas}, supposed that $\grtrgen = \bar\scalfac\grtr$ for the direct method and $\grtrgen = \grtr\scalvec^\T$ for the lifting method:
\begin{proposition}\label{prop:proofs:modelcovar}
	Let Model~\ref{model:results:measurements} hold true. Recalling the notation of Definition~\ref{def:results:scalparam}, we have:
	\begin{thmlist}
	\item\label{prop:proofs:modelcovar:direct}
		\emph{Direct method:} Set $\adir \coloneqq \sum_{j = 1}^M \a^j$. Then, the measurement pair $(\tfrac{1}{\sqrt{M}}\adir, \tfrac{1}{\sqrt{M}}y)$ obeys Model~\ref{model:proofs:generalobs} (with $\subgparam = 1$ and $d = n$) and we have
		\begin{equation}
			\modelcovar(\bar\scalfac\grtr; \tfrac{1}{\sqrt{M}}\adir, \tfrac{1}{\sqrt{M}}y) = 0.
		\end{equation}
	\item\label{prop:proofs:modelcovar:lifting}
		\emph{Lifting method:} Set $\alift \coloneqq [\a^1 \cdots \a^M]$. Then, the measurement pair\footnote{Hereafter, the matrix space $\R^{n \times N}$ is canonically identified with $\R^{n \cdot M}$. In particular, we regard $\alift$ as a random vector in $\R^{n \cdot M}$.} $(\alift, y)$ obeys Model~\ref{model:proofs:generalobs} (with $\subgparam = 1$ and $d = n\cdot M$) and we have
		\begin{equation}
			\modelcovar(\grtr\scalvec^\T; \alift, y) = 0.
		\end{equation}
	\end{thmlist}
\end{proposition}
The proof of Proposition~\ref{prop:proofs:modelcovar} is provided at the end of this subsection.

\begin{remark}
	The statement of Proposition~\ref{prop:proofs:modelcovar} reflects the \emph{orthogonality principle} in linear estimation theory:
	the scaling parameters in Definition~\ref{def:results:scalparam} are precisely chosen such that the model mismatch (sometimes referred to as ``noise'') is uncorrelated with the measurement vectors $\adir$ and $\alift$, respectively.
	Note that this classical methodology was also applied in \cite{plan2015lasso,genzel2016estimation}, although not stated explicitly.
\end{remark}

The second important quantity of Theorem~\ref{thm:proofs:abstractrecovery} is the conic mean width $\meanwidth[1]{\cone{\sset}{\grtrgen}}$. 
Intuitively, this geometric parameter measures the complexity of the set $\sset$ in a local neighborhood of $\grtrgen$, and by \eqref{eq:proofs:abstractrecovery:meas}, it gets related to the sampling rate of the actual estimation problem.
It is therefore substantial to select a structural constraint for \eqref{eq:proofs:klasso} which is ``compatible'' with the target vector $\grtrgen$, implying that recovery succeeds with a very few measurements.
In the specific setup of Section~\ref{sec:results}, this is indeed the case because we can establish (sharp) upper bounds on $\meanwidth[1]{\cone{\sset}{\grtrgen}}$ that only logarithmically depend on the dimension of the ambient space:
\begin{proposition}
	\begin{thmlist}
	\item\label{prop:proofs:meanwidth:direct}
		\emph{Direct method:} Let the assumptions of Theorem~\ref{thm:results:direct} be satisfied and set
		\begin{equation}
			\sset \coloneqq \{ \x \in \R^n \suchthat \lnorm{\x}[1] \leq \lnorm{\bar\scalfac\grtr}[1] = R \}.
		\end{equation}
		Then
		\begin{equation}\label{eq:meanwidth:direct}
			\meanwidth[1]{\cone{\sset}{\bar\scalfac\grtr}} \lesssim \sqrt{ s \cdot \log(\tfrac{2n}{s})}.
		\end{equation}
	\item\label{prop:proofs:meanwidth:lifting}
		\emph{Lifting method:} Let the assumptions of Theorem~\ref{thm:results:lifting} be satisfied and set
		\begin{equation}
			\sset \coloneqq \{ \X \in \R^{n\times M} \suchthat \lnorm{\X}[1,2] \leq \lnorm[\big]{\grtr\scalvec^\T}[1,2] = R \}.
		\end{equation}
		Then
		\begin{equation}\label{eq:meanwidth:lifting}
			\meanwidth[1]{\cone{\sset}{\grtr\scalvec^\T}} \lesssim \sqrt{ s \cdot \max\{M, \log(\tfrac{2n}{s})\}}.
		\end{equation}
	\end{thmlist}
\label{prop:proofs:meanwidth}
\end{proposition}
These bounds have been already established in the literature (e.g., see \cite{chandrasekaran2012geometry,flinth2016sparsedeconv}). For the sake of self-containedness, a proof is however provided in Appendix~\ref{subsec:app:conicmeanwidth}. We are now ready to prove the two main results of Section~\ref{sec:results}:
\begin{proof}[Proof of Theorem~\ref{thm:results:direct}]
	By Proposition~\ref{prop:proofs:modelcovar}\ref{prop:proofs:modelcovar:direct}, we know that $(\tfrac{1}{\sqrt{M}}\adir, \tfrac{1}{\sqrt{M}}y)$ with $\adir \coloneqq \sum_{j = 1}^M \a^j$ satisfies Model~\ref{model:proofs:generalobs} for $\subgparam = 1$. Now, we would like to apply Theorem~\ref{thm:proofs:abstractrecovery} with $\sset = R\ball[1][n]$ and $\grtrgen = \bar\scalfac\grtr$.
	Proposition~\ref{prop:proofs:meanwidth}\ref{prop:proofs:meanwidth:direct} yields
	\begin{equation}
		\meanwidth[1]{\cone{\sset}{\grtrgen}}^2 \lesssim s \cdot \log(\tfrac{2n}{s}),
	\end{equation}
	and therefore, the assumption of \eqref{eq:results:direct:meas} indeed implies \eqref{eq:proofs:abstractrecovery:meas}.
	
	Next, we bound the mismatch deviation
	\begin{equation}
		\modeldev(\bar\scalfac\grtr) = \modeldev(\bar\scalfac\grtr; \tfrac{1}{\sqrt{M}}\adir, \tfrac{1}{\sqrt{M}}y) = \tfrac{1}{\sqrt{M}} \normsubg{\sp{\adir}{\bar\scalfac\grtr} - y} \stackrel{\eqref{eq:results:measurements:meas}}{=} \tfrac{1}{\sqrt{M}} \normsubg[\Big]{\sum_{j = 1}^M z_j - e},
	\end{equation}
	where $z_j \coloneqq \sp{\a^j}{\bar\scalfac\grtr} - \fobs_j(\sp{\a^j}{\grtr})$, $j = 1, \dots, M$. Since Model~\ref{model:results:measurements} particularly assumes $\mean[\fobs_j(\sp{\a^j}{\grtr})] = 0$, we can conclude that $z_1, \dots, z_M, e$ are independent, mean-zero sub-Gaussian variables.\footnote{Note that we have implicitly assumed that $\fobs_j(\sp{\a^j}{\grtr})$ is sub-Gaussian because otherwise we would have $\modeldev_\text{Dir} = \infty$ and the claim is trivial.}
	Hence, by \eqref{eq:intro:notation:hoeffding},
	\begin{align} 
		\modeldev(\bar\scalfac\grtr)^2 = \tfrac{1}{M} \normsubg[\Big]{\sum_{j = 1}^M z_j - e}^2 \lesssim \tfrac{1}{M} \Big( \sum_{j = 1}^M \normsubg{z_j}^2 + \normsubg{e}^2 \Big) \stackrel{\lnorm{\grtr} = 1}{\leq} \modeldev_\text{Dir}^2 + \tfrac{\nu^2}{M}.
	\end{align}
	
	Finally, observing that $\modelcovar(\bar\scalfac\grtr) = 0$ according to Proposition~\ref{prop:proofs:modelcovar}\ref{prop:proofs:modelcovar:direct}, the claim follows from Theorem~\ref{thm:proofs:abstractrecovery}. Note that \eqref{eq:proofs:klasso} and \eqref{eq:direct:lasso} are equivalent, since multiplying the objective function by a factor of $M$ does not change the set of minimizers.	
\end{proof}

\begin{proof}[Proof of Theorem~\ref{thm:results:lifting}]
	By Proposition~\ref{prop:proofs:modelcovar}\ref{prop:proofs:modelcovar:lifting}, we know that $(\alift, y)$ with $\alift \coloneqq [\a^1 \cdots \a^M]$ satisfies Model~\ref{model:proofs:generalobs} for $\subgparam = 1$. Now, we would like to apply Theorem~\ref{thm:proofs:abstractrecovery} with $\sset = \{ \X \in \R^{n\times M} \suchthat \lnorm{\X}[1,2] \leq R \}$ and $\grtrgen = \grtr\scalvec^\T$.
	Proposition~\ref{prop:proofs:meanwidth}\ref{prop:proofs:meanwidth:lifting} yields
	\begin{equation}
		\meanwidth[1]{\cone{\sset}{\grtrgen}}^2 \lesssim s \cdot \max\{M, \log(\tfrac{2n}{s})\},
	\end{equation}
	and therefore, the assumption of \eqref{eq:results:lifting:meas} indeed implies \eqref{eq:proofs:abstractrecovery:meas}.
	
	Similarly to the proof of Theorem~\ref{thm:results:direct}, we now bound the mismatch deviation
	\begin{align}
		\modeldev(\grtr\scalvec^\T) &= \modeldev(\grtr\scalvec^\T; \alift, y) = \normsubg{\sp{\alift}{\grtr\scalvec^\T} - y} = \normsubg{\sp{\alift\scalvec}{\grtr} - y} \\
		&= \normsubg[\Big]{\sp{\sum_{j = 1}^M\scalfac_j\a^j}{\grtr}  - y} \stackrel{\eqref{eq:results:measurements:meas}}{=} \normsubg[\Big]{\sum_{j = 1}^M z_j - e},
	\end{align}
	where $z_j \coloneqq \sp{\a^j}{\scalfac_j\grtr} - \fobs_j(\sp{\a^j}{\grtr})$, $j = 1, \dots, M$. Since Model~\ref{model:results:measurements} particularly assumes $\mean[\fobs_j(\sp{\a^j}{\grtr})] = 0$, we can conclude that $z_1, \dots, z_M, e$ are independent, mean-zero sub-Gaussians.
	Hence, by \eqref{eq:intro:notation:hoeffding},
	\begin{align} 
		\modeldev(\grtr\scalvec^\T)^2 = \normsubg[\Big]{\sum_{j = 1}^M z_j - e}^2 \lesssim \Big( \sum_{j = 1}^M \normsubg{z_j}^2 + \normsubg{e}^2 \Big) \stackrel{\lnorm{\grtr} = 1}{\leq} M \cdot \modeldev_\text{Lift}^2 + \nu^2.
	\end{align}
	
	From Proposition~\ref{prop:proofs:modelcovar}\ref{prop:proofs:modelcovar:lifting}, we obtain $\modelcovar(\grtr\scalvec^\T; \alift, y) = 0$, so that the error bound of Theorem~\ref{thm:proofs:abstractrecovery} reads as follows:
	\begin{align}
		\Big(\sum_{j = 1}^M\lnorm{\solu^j - \scalfac_j\grtr}^2 \Big)^{1/2} = \lnorm{[\solu^1 \cdots \solu^M] - \grtr\scalvec^\T} \lesssim \modeldev(\grtr\scalvec^\T) \cdot \delta \lesssim (M \cdot \modeldev_\text{Lift}^2 + \nu^2)^{\frac{1}{2}} \cdot \delta.
	\end{align}
	Dividing both sides by $1 / \sqrt{M}$ precisely gives the error bound \eqref{eq:results:lifting:bound} of Theorem~\ref{thm:results:lifting}.
	
	\revision{It remains to show the error bounds in \eqref{eq:results:lifting:bound:svd}. To this end, we proceed analogously to the proof of \cite[Cor.~1]{iwen2015note}. Let $\uv_1, \dots, \uv_n \in \S^{n-1}$ and $\vv_1, \dots, \vv_M \in \S^{M-1}$ be the left and right singular vectors of $\solu[\X] \in \R^{n \times M}$, respectively, and let $\tau_1 \geq \tau_2 \geq \cdots \geq \tau_r$ be the corresponding singular values with $r \coloneqq \min\{n, M\}$; in particular, note that $\solu = \uv_1$, $\solu[\scalvec] = \vv_1$, and $\tau = \tau_1$. Moreover, we set $\lambda \coloneqq \lnorm{\grtr \scalvec^\T} = \lnorm{\scalvec}$, which is the only positive singular value of the rank-one matrix $\grtr \scalvec^\T \in \R^{n \times M}$.
		
	We now apply Weyl's inequality (for non-hermitian matrices) to obtain
	\begin{align}
		\opnorm{\solu[\X] - \lambda \uv_1 \vv_1^\T} &= \opnorm[\Big]{(\tau_1 - \lambda) \uv_1 \vv_1^\T + \sum_{j = 2}^r \tau_j  \uv_j \vv_j^\T} \\
		&= \max\{ \abs{\tau_1 - \lambda}, \tau_2, \tau_3, \dots, \tau_r \} \\
		&\leq \opnorm{\solu[\X] - \grtr \scalvec^\T} \leq \lnorm{\solu[\X] - \grtr \scalvec^\T} \stackrel{\eqref{eq:results:lifting:bound}}{\lesssim} \Err \cdot \sqrt{M}. \label{eq:proofs:lifting:weyl}
	\end{align}
	The triangle inequality and \eqref{eq:results:lifting:bound} again then yield
	\begin{equation}
		\opnorm{\grtr \scalvec^\T - \lambda \uv_1 \vv_1^\T} \leq \opnorm{\grtr \scalvec^\T - \solu[\X]} + \opnorm{\solu[\X] - \lambda \uv_1 \vv_1^\T} \lesssim \Err \cdot \sqrt{M}.
	\end{equation}
	This implies
	\begin{align}
		\lambda^2 - \lambda \sp{\grtr}{\solu} \sp{\scalvec}{\solu[\scalvec]} = \tfrac{1}{2}\lnorm{\grtr \scalvec^\T - \lambda \uv_1 \vv_1^\T}^2 \stackrel{(\ast)}{\leq} \opnorm{\grtr \scalvec^\T - \lambda \uv_1 \vv_1^\T}^2 \lesssim \Err^2 \cdot M,
	\end{align}
	where $(\ast)$ follows from the fact that $\grtr \scalvec^\T - \lambda \uv_1 \vv_1^\T$ is at most rank $2$. Dividing by $\lambda^2 = \lnorm{\scalvec}^2$ then leads to
	\begin{equation}\label{eq:proofs:lifting:svdestim}
		1 - \sp{\grtr}{\solu} \sp{\tfrac{\scalvec}{\lnorm{\scalvec}}}{\solu[\scalvec]} \lesssim \Err^2 \cdot \tfrac{M}{\lnorm{\scalvec}^2} \ ,
	\end{equation}
	and due to the additional assumption that $\Err \lesssim \lnorm{\scalvec} / \sqrt{M}$, we may also assume that $1 - \sp{\grtr}{\solu} \sp{\tfrac{\scalvec}{\lnorm{\scalvec}}}{\solu[\scalvec]} \leq 1$, i.e., $\sp{\grtr}{\solu} \sp{\tfrac{\scalvec}{\lnorm{\scalvec}}}{\solu[\scalvec]} \geq 0$. 
	Since $\grtr, \solu \in \S^{n-1}$ and $\tfrac{\scalvec}{\lnorm{\scalvec}}, \solu[\scalvec] \in \S^{M-1}$, we obtain the following error estimate from \eqref{eq:proofs:lifting:svdestim} (recall that $\theta_1 = \sign(\sp{\grtr}{\solu})$ and $\theta_2 = \sign(\sp{\scalvec}{\solu[\scalvec]})$):
	\begin{align}
		& \lnorm{\solu - \theta_1 \grtr}^2 + \lnorm{\solu[\scalvec] - \theta_2 \tfrac{\scalvec}{\lnorm{\scalvec}}}^2 = 2 - 2\sp{\theta_1\grtr}{\solu} + 2 - 2\sp{\theta_2\tfrac{\scalvec}{\lnorm{\scalvec}}}{\solu[\scalvec]} \\
		= {} &  4 - 2\abs{\sp{\grtr}{\solu}} - 2\abs{\sp{\tfrac{\scalvec}{\lnorm{\scalvec}}}{\solu[\scalvec]}} \stackrel{(\ast)}{\leq} 4 - 4 \sqrt{ \abs{\sp{\grtr}{\solu}} \cdot \abs{\sp{\tfrac{\scalvec}{\lnorm{\scalvec}}}{\solu[\scalvec]}}} \\
		= {} & 4 - 4 \sqrt{ \sp{\grtr}{\solu} \sp{\tfrac{\scalvec}{\lnorm{\scalvec}}}{\solu[\scalvec]}} \leq 4 \Big( 1 - \sqrt{ \sp{\grtr}{\solu} \sp{\tfrac{\scalvec}{\lnorm{\scalvec}}}{\solu[\scalvec]}} \Big) \cdot \Big( 1 + \sqrt{ \sp{\grtr}{\solu} \sp{\tfrac{\scalvec}{\lnorm{\scalvec}}}{\solu[\scalvec]}} \Big) \\
		= {} & 4 \big( 1 - \sp{\grtr}{\solu} \sp{\tfrac{\scalvec}{\lnorm{\scalvec}}}{\solu[\scalvec]} \big) \stackrel{\eqref{eq:proofs:lifting:svdestim}}{\lesssim} \Err^2 \cdot \tfrac{M}{\lnorm{\scalvec}^2} \ , \label{eq:proofs:lifting:svdestim2}
	\end{align}
	where $(\ast)$ follows from the inequality of arithmetic and geometric means.
	This already implies $\lnorm{\solu - \theta_1 \grtr} \lesssim \Err \cdot \sqrt{M} / \lnorm{\scalvec}$, and the second bound in \eqref{eq:results:lifting:bound:svd} follows from
	\begin{align}
		\lnorm{\tau\solu[\scalvec] - \theta_2 \scalvec} &\leq \lnorm{\tau\solu[\scalvec] - \lambda \solu[\scalvec]} + \underbrace{\lnorm{\lambda \solu[\scalvec] - \theta_2 \scalvec}}_{\mathclap{\stackrel{\eqref{eq:proofs:lifting:svdestim2}}{\lesssim} \Err \sqrt{M}}} \lesssim \underbrace{\abs{\tau - \lambda}}_{\mathclap{\stackrel{\eqref{eq:proofs:lifting:weyl}}{\lesssim} \Err \sqrt{M}}} + \Err \cdot \sqrt{M} \lesssim \Err \cdot \sqrt{M}.
	\end{align}
	}
\end{proof}

To conclude our proofs, it just remains to verify Proposition~\ref{prop:proofs:modelcovar}:
\begin{proof}[Proof of Proposition~\ref{prop:proofs:modelcovar}]
\ref{prop:proofs:modelcovar:direct} First, we observe that $\tfrac{1}{\sqrt{M}}\adir \distributed \Normdistr{\vnull}{\I{n}}$, implying that $\tfrac{1}{\sqrt{M}}\adir$ is an isotropic sub-Gaussian vector with $\subgparam \coloneqq \normsubg{\tfrac{1}{\sqrt{M}}\adir} = 1$ and the hypothesis of Model~\ref{model:proofs:generalobs} is indeed satisfied.

Recalling the definition of $\modelcovar(\bar\scalfac\grtr; \tfrac{1}{\sqrt{M}}\adir, \tfrac{1}{\sqrt{M}}y)$ (see Definition~\ref{def:proofs:modelmismatch}), it suffices to show that
\begin{equation}
	\mean{}_{(\adir,y)}[(\sp{\adir}{\bar\scalfac\grtr} - y) \adir] = \vnull.
\end{equation}
For this purpose, let $\proj_{\grtr}(\cdot) = \sp{\cdot}{\grtr} \grtr$ denote the orthogonal projection onto $\spann\{\grtr\}$ and $\proj_{\orthcompl{\{\grtr\}}} = \I{} - \proj_{\grtr}$ its orthogonal complement (note that $\lnorm{\grtr} = 1$).
The rotation invariance of Gaussian random vectors particularly implies that $\sp{\a^j}{\grtr}$ and $\proj_{\orthcompl{\{\grtr\}}}(\a^j)$ are independent (and not just uncorrelated) for every $j = 1, \dots, M$. The independence of $\a^1, \dots, \a^M, e$ and the definition of $\bar\scalfac$ now yield
\begin{align}
	\mean[(\sp{\adir}{\bar\scalfac\grtr} - y) \adir] &= \mean{}\Big[\sum_{j = 1}^M \underbrace{(\sp{\a^j}{\bar\scalfac\grtr} - \fobs_j(\sp{\a^j}{\grtr}))}_{\eqqcolon z_j} \sum_{j' = 1}^M \a^{j'}\Big] - \underbrace{\mean{}\Big[e \cdot \sum_{j' = 1}^M \a^{j'} \Big]}_{= \vnull} \\
	&= \sum_{j = 1}^M \mean[z_j \a^{j}] = \sum_{j = 1}^M \mean[z_j [\proj_{\grtr} + \proj_{\orthcompl{\{\grtr\}}}](\a^{j})] \\
	&= \sum_{j = 1}^M \mean[z_j \sp{\a^{j}}{\grtr} \grtr] + \sum_{j = 1}^M \underbrace{\mean[(\bar\scalfac\sp{\a^j}{\grtr} - \fobs_j(\sp{\a^j}{\grtr})) \proj_{\orthcompl{\{\grtr\}}}(\a^{j})]}_{ \stackrel{(\ast)}{=} \vnull} \\
	&= \sum_{j = 1}^M ( \mean[(\bar\scalfac\sp{\a^{j}}{\grtr}^2 - \fobs_j(\sp{\a^j}{\grtr}) \sp{\a^{j}}{\grtr} )  \grtr] \\
	&= \Big( \bar\scalfac \sum_{j = 1}^M \underbrace{\mean[\sp{\a^{j}}{\grtr}^2]}_{= 1} - \sum_{j = 1}^M \underbrace{\mean[\fobs_j(\sp{\a^j}{\grtr}) \sp{\a^{j}}{\grtr}]}_{= \scalfac_j} \Big) \grtr \\
	&= \underbrace{\Big( M \cdot \bar\scalfac - \sum_{j = 1}^M \scalfac_j \Big)}_{= 0} \grtr = \vnull,
\end{align}
where $(\ast)$ is due to the fact that $\fobs_j(\sp{\a^j}{\grtr})$ and $\proj_{\orthcompl{\{\grtr\}}}(\a^{j})$ are independent in the Gaussian case.

\ref{prop:proofs:modelcovar:lifting}
Since $\alift \distributed \Normdistr{\vnull}{\I{n\cdot M}}$, the hypothesis of Model~\ref{model:proofs:generalobs} is satisfied with $\subgparam \coloneqq \normsubg{\alift} = 1$.
Using the notation of the first part and the independence of $\a^1, \dots, \a^M, e$, we obtain
\begin{align}
	\mean[(\sp{\alift}{\grtr\scalvec^\T} - y) \a^j] &= \mean{}\Big[\sum_{j' = 1}^M(\sp{\a^{j'}}{\scalfac_{j'}\grtr} - \fobs_{j'}(\sp{\a^{j'}}{\grtr}) \a^j\Big] - \underbrace{\mean[e \cdot \a^j]}_{= \vnull} \\
	&= \mean[(\sp{\a^{j}}{\scalfac_{j}\grtr} - \fobs_{j}(\sp{\a^{j}}{\grtr}) \a^j] \\
	&= \mean[(\sp{\a^{j}}{\scalfac_{j}\grtr} - \fobs_{j}(\sp{\a^{j}}{\grtr}) (\sp{\a^{j}}{\grtr} \grtr + \proj_{\orthcompl{\{\grtr\}}}(\a^{j}))] \\
	&= \mean[(\sp{\a^{j}}{\scalfac_{j}\grtr} - \fobs_{j}(\sp{\a^{j}}{\grtr}) \sp{\a^{j}}{\grtr} ] \grtr \\
	&= \Big(  \scalfac_j \underbrace{\mean[\sp{\a^{j}}{\grtr}^2]}_{= 1} - \underbrace{\mean[\fobs_{j}(\sp{\a^{j}}{\grtr}) \sp{\a^{j}}{\grtr}]}_{= \scalfac_j} \Big) \grtr = (\scalfac_j  - \scalfac_j) \grtr = \vnull.
\end{align}
Since this holds true for every $j = 1, \dots, M$, we conclude that $\mean[(\sp{\alift}{\grtr\scalvec^\T} - y) \alift] = \vnull$ and therefore $\modelcovar(\grtr\scalvec^\T; \alift, y) = 0$.
\end{proof}

\subsection{Proofs of Theorem~\ref{thm:extensions:hybrid} and Theorem~\ref{thm:extensions:hybridglobal}}
\label{subsec:proofs:hybridresults}

Similarly to the proofs of the previous part, we first set up a relationship between the hybrid approach of Algorithm~\ref{algo:hybrid} and Model~\ref{model:proofs:generalobs}, and compute the mismatch covariance in this specific case:
\begin{proposition}\label{prop:proofs:modelcovar:hybrid}
	Suppose that Model~\ref{model:extensions:measurements} holds true and recall the notation of Definition~\ref{def:extensions:scalparam}.
	As in the lifting case, let $\alift \coloneqq [\a^1 \cdots \a^M] \in \R^{n \times M}$ 
	and set $\Ahyb \coloneqq \alift \hybW \in \R^{n\times N}$.
	If $\hybW^\T \hybW = \tfrac{M}{N} \I{N}$, the measurement pair\footnote{Again, we regard $\Ahyb$ as a random vector $\R^{n \cdot N}$.} $(\sqrt{\tfrac{N}{M}}\Ahyb, \sqrt{\tfrac{N}{M}} y)$ obeys Model~\ref{model:proofs:generalobs} (with $\subgparam$ from Model~\ref{model:extensions:measurements} and $d = n \cdot N$) and we have
		\begin{equation}\label{eq:proofs:modelcovar:hybrid:covar}
			\modelcovar(\grtr\hybscalvec^\T; \sqrt{\tfrac{N}{M}} \Ahyb, \sqrt{\tfrac{N}{M}} y) = \tfrac{N\sqrt{N}}{M} \cdot  \modelcovar_\text{Hyb}.
		\end{equation}
\end{proposition}

\begin{proof}
	First, we note that the independence of $\a^1, \dots, \a^M$ implies that $\alift$ is an isotropic, mean-zero random vector in $\R^{n \times M}$ with $\normsubg{\alift} \leq \subgparam$.
	
	By linearity of $\hybW$, we conclude that $\mean[\Ahyb] = \vnull$. Moreover, for every $\X \in \R^{n \times N}$, we have
	\begin{align}
		\mean[\sp{\sqrt{\tfrac{N}{M}}\Ahyb}{\X}^2] &= \tfrac{N}{M} \cdot \mean[\sp{\alift\hybW}{\X}^2] = \tfrac{N}{M} \cdot \mean[\sp{\alift}{\X\hybW^\T}^2] \stackrel{\eqref{eq:intro:notation:intropic}}{=} \tfrac{N}{M} \cdot \lnorm{\X\hybW^\T}^2 \\
		&= \tfrac{N}{M} \cdot \sp{\X\hybW^\T}{\X\hybW^\T} = \tfrac{N}{M} \cdot \sp{\X}{\X\underbrace{\hybW^\T\hybW}_{= \tfrac{M}{N} \I{N}}} = \sp{\X}{\X} = \lnorm{\X}^2,
	\end{align}
	implying that $\sqrt{\tfrac{N}{M}}\Ahyb$ is indeed isotropic. Next, we show that $\normsubg{\sqrt{\tfrac{N}{M}}\Ahyb} \leq \subgparam$.
	For this purpose, let $\X \in \R^{n \times N}$ with $\lnorm{\X} = 1$. Using that $\normsubg{\alift}\leq \subgparam$, we observe (without loss of generality $\lnorm{\X\hybW^\T} \neq 0$)
	\begin{align}
		\normsubg{\sp{\sqrt{\tfrac{N}{M}}\Ahyb}{\X}} &= \sqrt{\tfrac{N}{M}} \cdot \normsubg{\sp{\alift\hybW}{\X}} = \sqrt{\tfrac{N}{M}} \cdot \normsubg{\sp{\alift}{\X\hybW^\T}} \\*
		&= \sqrt{\tfrac{N}{M}} \cdot \lnorm{\X\hybW^\T} \cdot \normsubg{\sp{\alift}{\tfrac{\X\hybW^\T}{\lnorm{\X\hybW^\T}}}} \\
		&\leq \sqrt{\tfrac{N}{M}} \cdot  \underbrace{\lnorm{\X\hybW^\T}}_{= \sqrt{\tfrac{M}{N}} \cdot \lnorm{\X}} \cdot \subgparam = \underbrace{\lnorm{\X}}_{= 1} \cdot \subgparam = \subgparam,
	\end{align}
	where the assumption $\hybW^\T \hybW = \tfrac{M}{N} \I{N}$ was used again. Taking the supremum over $\X \in \S^{n\cdot N - 1}$, it follows that $\normsubg{\sqrt{\tfrac{N}{M}}\Ahyb} \leq \subgparam$.
	
	It remains to verify \eqref{eq:proofs:modelcovar:hybrid:covar}. For this, we set $\ahyb^k \coloneqq \alift \hybw^k = \sum_{j = 1}^M w_{j,k}\a^j$ for $k = 1, \dots, N$, where $\hybw^1, \dots, \hybw^N \in \R^M$ denote the columns of $\hybW$.
	Since $\a^1, \dots, \a^M$, and $e$ are independent, we can compute
	\begin{align}
		\mean[(\sp{\Ahyb}{\grtr\hybscalvec^\T} - y) \ahyb^k] &= \mean{}\Big[\Big(\sum_{k' = 1}^N\sp{\ahyb^{k'}}{\hybscalfac_{k'}\grtr} - \sum_{j = 1}^M \fobs_{j}(\sp{\a^{j}}{\grtr}) \Big) \ahyb^k\Big] - \underbrace{\mean[e \cdot \ahyb^k]}_{= \vnull} \\*
			&= \mean{}\Big[\Big(\sum_{k' = 1}^N\sp{\sum_{j' = 1}^M w_{j',k'}\a^{j'}}{\hybscalfac_{k'}\grtr} - \sum_{j = 1}^M \fobs_{j}(\sp{\a^{j}}{\grtr}) \Big) \sum_{j'' = 1}^M w_{j'',k}\a^{j''}\Big] \\
			&= \sum_{k' = 1}^N \sum_{j' = 1}^M \mean[\sp{ w_{j',k'}\a^{j'}}{\hybscalfac_{k'}\grtr} w_{j',k}\a^{j'}] - \sum_{j = 1}^M \mean[\fobs_{j}(\sp{\a^{j}}{\grtr}) w_{j,k}\a^{j}] \\
			&= \sum_{j = 1}^M w_{j,k} \Big( \mean{}\Big[\Big(\sum_{k' = 1}^N  w_{j,k'} \hybscalfac_{k'}\Big) \sp{ \a^{j}}{\grtr} \a^{j} \Big] - \mean[\fobs_{j}(\sp{\a^{j}}{\grtr}) \a^{j}]\Big) \\
			&= \sum_{j = 1}^M w_{j,k} \Big(\Big(\sum_{k' = 1}^N  w_{j,k'} \hybscalfac_{k'}\Big) \mean[\sp{ \a^{j}}{\grtr} \a^{j}] - \mean[\fobs_{j}(\sp{\a^{j}}{\grtr}) \a^{j}]\Big).
	\end{align}
	Next, we decompose $\a^j$ again by an orthogonal projection onto $\spann\{\grtr\}$ and its orthogonal complement, in order to simplify the remaining expected values:
	\begin{align}
		\mean[\sp{ \a^{j}}{\grtr} \a^{j}] &= \mean[\sp{ \a^{j}}{\grtr} (\sp{\a^{j}}{\tfrac{\grtr}{\lnorm{\grtr}}} \tfrac{\grtr}{\lnorm{\grtr}} + \proj_{\orthcompl{\{\grtr\}}}(\a^{j}))] \\*
		&= \mean[\sp{ \a^{j}}{\grtr}^2] \tfrac{\grtr}{\lnorm{\grtr}^2} = \lnorm{\grtr}^2 \cdot \tfrac{\grtr}{\lnorm{\grtr}^2} = \grtr, \\
		\mean[\fobs_{j}(\sp{\a^{j}}{\grtr}) \a^{j}] &= \mean[\fobs_{j}(\sp{\a^{j}}{\grtr}) (\sp{\a^{j}}{\tfrac{\grtr}{\lnorm{\grtr}}} \tfrac{\grtr}{\lnorm{\grtr}} + \proj_{\orthcompl{\{\grtr\}}}(\a^{j}))] \\
		&= \scalfac_j \cdot \grtr + \underbrace{\mean[\fobs_{j}(\sp{\a^{j}}{\grtr})\proj_{\orthcompl{\{\grtr\}}}(\a^{j})]}_{= \orthmodelbias^j},
	\end{align}
	where in the first part, we have used the isotropy of $\a^j$ and \eqref{eq:intro:notation:isotropicsp}, and the second identity follows from Definition~\ref{def:extensions:scalparam}.
	Using $\hybW^\T\hybW = \tfrac{M}{N} \I{N} = \tfrac{M}{N} [\delta_{k,k'}]$ once more, we now obtain
	\begin{align}
		\mean[(\sp{\Ahyb}{\grtr\hybscalvec^\T} - y) \ahyb^k] &= \sum_{j = 1}^M w_{j,k} \Big(\Big(\sum_{k' = 1}^N  w_{j,k'} \hybscalfac_{k'}\Big) - \scalfac_j\Big) \cdot \grtr - \sum_{j = 1}^M w_{j,k} \orthmodelbias^j \\*
		&= \Big( \sum_{k' = 1}^N \hybscalfac_{k'}  \underbrace{\sum_{j = 1}^M w_{j,k} w_{j,k'}}_{= \tfrac{M}{N} \delta_{k,k'}} - \underbrace{\sum_{j = 1}^M w_{j,k} \scalfac_j}_{= \sp{\hybw^k}{\scalvec} = \tfrac{M}{N} \hybscalfac_k} \Big) \cdot \grtr - \sum_{j = 1}^M w_{j,k} \orthmodelbias^j \\
		&= \underbrace{\tfrac{M}{N} \cdot \Big( \sum_{k' = 1}^N \hybscalfac_{k'} \delta_{k,k'} - \hybscalfac_k \Big)}_{= 0} \cdot \grtr - \sum_{j = 1}^M w_{j,k} \orthmodelbias^j = - \sum_{j = 1}^M w_{j,k} \orthmodelbias^j.
	\end{align}
	Hence,
	\begin{align}
		\modelcovar(\grtr\hybscalvec^\T; \sqrt{\tfrac{N}{M}} \Ahyb, \sqrt{\tfrac{N}{M}} y) &= \lnorm[\Big]{ \mean[(\sp{\sqrt{\tfrac{N}{M}} \Ahyb}{\grtr\hybscalvec^\T} - \sqrt{\tfrac{N}{M}} y) \sqrt{\tfrac{N}{M}} \Ahyb]} \\*
		&= \tfrac{N}{M} \cdot \lnorm[\Big]{ \mean{}\big[(\sp{\Ahyb}{\grtr\hybscalvec^\T} - y) [\ahyb^1 \cdots \ahyb^N]\big]} \\
		&= \tfrac{N}{M} \cdot \lnorm[\Big]{\Big[\sum_{j = 1}^M w_{j,1} \orthmodelbias^j \dots \sum_{j = 1}^M w_{j,N} \orthmodelbias^j \Big]} \\
		&= \tfrac{N}{M} \cdot \Big( \sum_{k = 1}^N \lnorm{\orthmodelbias \hybw^k}^2\Big)^{1/2} = \tfrac{N\sqrt{N}}{M} \cdot \modelcovar_\text{Hyb}.
	\end{align}
\end{proof}

The proof of Theorem~\ref{thm:extensions:hybrid} is again a consequence of Theorem~\ref{thm:proofs:abstractrecovery}:
\begin{proof}[Proof of Theorem~\ref{thm:extensions:hybrid}]
	From Proposition~\ref{prop:proofs:modelcovar:hybrid}, we know that $(\sqrt{\tfrac{N}{M}}\Ahyb, \sqrt{\tfrac{N}{M}} y)$ with $\Ahyb \coloneqq \alift \hybW$ satisfies Model~\ref{model:proofs:generalobs}. Now, we would like to apply Theorem~\ref{thm:proofs:abstractrecovery} with $\grtrgen = \grtr\hybscalvec^\T$.

	Recalling the definition of $\modeldev_\text{Hyb}^2$ and the independence of $\a^1, \dots, \a^M$, we can bound the model deviation as follows:
	\begin{align}
		\modeldev(\grtr\hybscalvec^\T; \sqrt{\tfrac{N}{M}} \Ahyb, \sqrt{\tfrac{N}{M}} y)^2 &= \normsubg{\sp{\sqrt{\tfrac{N}{M}} \Ahyb}{\grtr\hybscalvec^\T} - \sqrt{\tfrac{N}{M}} y}^2 \\*
		&= \tfrac{N}{M} \cdot \normsubg{\sp{\alift}{\tfrac{N}{M}\grtr\scalvec^\T \hybW \hybW^\T} - y}^2 \\
		&= \tfrac{N}{M} \cdot \normsubg[\Big]{\sum_{j = 1}^M ( \sp{\a^j}{\hybgrtr^j} - \fobs_j(\sp{\a^j}{\grtr})) - e}^2 \\
		&\stackrel{\eqref{eq:intro:notation:hoeffding}}{\lesssim} \tfrac{N}{M} \cdot \sum_{j = 1}^M \normsubg[\Big]{\sp{\a^j}{\hybgrtr^j} - \fobs_j(\sp{\a^j}{\grtr})}^2 + \tfrac{N \nu^2}{M} \\
		&= N \cdot \modeldev_\text{Hyb}^2 + N \cdot \tfrac{\nu^2}{M}. \label{eq:proofs:hybrid:modeldev}
	\end{align}
	Since $\modelcovar(\grtr\hybscalvec^\T; \sqrt{\tfrac{N}{M}} \Ahyb, \sqrt{\tfrac{N}{M}} y) = \tfrac{N\sqrt{N}}{M} \cdot  \modelcovar_\text{Hyb}$ by Proposition~\ref{prop:proofs:modelcovar:hybrid}, the error bound of Theorem~\ref{thm:proofs:abstractrecovery} states
	\begin{align}
		\Big(\sum_{k = 1}^N\lnorm{\solu^j - \hybscalfac_j\grtr}^2 \Big)^{1/2} &= \lnorm{[\solu^1 \cdots \solu^N] - \grtr\hybscalvec^\T} \\*
		&\lesssim \subgparam^{-1} \cdot \modeldev(\grtr\hybscalvec^\T) \cdot \delta + \modelcovar(\grtr\hybscalvec^\T) \\
		&\lesssim \subgparam^{-1} \cdot \sqrt{N} \cdot (\modeldev_\text{Hyb}^2 + \tfrac{\nu^2}{M})^{\frac{1}{2}} \cdot \delta + \tfrac{N\sqrt{N}}{M} \cdot\modelcovar_\text{Hyb}.
	\end{align}
	Dividing both sides by $\sqrt{N}$ gives the bound of Theorem~\ref{thm:extensions:hybrid}.
	Similarly to the proof of Theorem~\ref{thm:results:direct}, note that by multiplying the objective function of \eqref{eq:hybrid:lasso} by $N / M$, we precisely end up with the $\sset$-Lasso \eqref{eq:proofs:klasso} in the setup of Model~\ref{model:proofs:generalobs}.
\end{proof}

For the proof of Theorem~\ref{thm:extensions:hybridglobal}, we cannot apply Theorem~\ref{thm:proofs:abstractrecovery} anymore, since it is limited to the conic mean width. Instead, the refined version of Theorem~\ref{thm:app:abstractrecovery:localmw} turns out to be useful:
\begin{proof}[Proof of Theorem~\ref{thm:extensions:hybridglobal}]

	As before, we first note that, by Proposition~\ref{prop:proofs:modelcovar:hybrid}, the measurement pair $(\sqrt{\tfrac{N}{M}}\Ahyb, \sqrt{\tfrac{N}{M}} y)$ indeed satisfies Model~\ref{model:proofs:generalobs}.

	In contrast to the proof of Theorem~\ref{thm:extensions:hybrid}, we now apply Theorem~\ref{thm:app:abstractrecovery:localmw} with $\grtrgen = \grtr\hybscalvec^\T$. Adapting the proof strategy of \cite[Thm.~1.3]{genzel2016estimation}, we choose the (desired) error accuracy $t$ as follows:
	\begin{equation}
		t \coloneqq D \cdot \Big[ \subgparam \cdot \Big(\frac{\meanwidth{\sset}}{\sqrt{m}}\Big)^{1/2} + \frac{u}{\sqrt{m}}\Big] + D' \cdot \modelcovar(\grtrgen) ,
	\end{equation}
	where the constants $D, D' \gtrsim 1$ are specified later on.
	First, we observe that
	\begin{align}
		\Big(\tfrac{1}{t}\meanwidth[t]{\sset-\grtrgen}\Big)^2 &\leq \tfrac{1}{t^2}\underbrace{\meanwidth{(\sset-\grtrgen)\intersec t \S^{n\cdot N - 1}}^2}_{\leq \meanwidth{\sset - \sset}^2 \leq 4 \meanwidth{\sset}^2} \leq \tfrac{4}{t^2} \cdot \meanwidth{\sset}^2 \\*
		&\leq \tfrac{4}{D^2 \cdot \subgparam^2} \cdot \tfrac{\sqrt{m}}{\meanwidth{\sset}} \cdot \meanwidth{\sset}^2 = \tfrac{4}{D^2 \cdot \subgparam^2} \cdot \sqrt{m} \cdot \meanwidth{\sset} \label{eq:proofs:hybridglobal:locglob} \\
		&\stackrel{\eqref{eq:extensions:hybridglobal:meas}}{\lesssim} \tfrac{1}{D^2 \cdot \subgparam^4} \cdot \underbrace{\delta^{2}}_{\leq 1} \cdot m \leq \tfrac{1}{D^2 \cdot \subgparam^4} \cdot m, 
	\end{align}
	and since we may just enlarge $D$ later on, the condition \eqref{eq:app:abstractrecovery:localmw:meas} of Theorem~\ref{thm:app:abstractrecovery:localmw} is indeed satisfied.
	Next, we choose $D \gtrsim \max\{1, \subgparam \cdot \modeldev(\grtrgen) \}$ and bound the right-hand side of \eqref{eq:app:abstractrecovery:localmw:implbound}:
	\begin{align}
		& C'' \cdot \Big( \underbrace{ \subgparam \cdot \modeldev(\grtrgen)}_{\lesssim D \lesssim D^2} \cdot \frac{\tfrac{1}{t}\meanwidth[t]{\sset - \grtrgen} + u}{\sqrt{m}} + \modelcovar(\grtrgen)\Big) \\*
		\lesssim{} & \frac{D^2}{2} \cdot \frac{\tfrac{1}{t}\meanwidth[t]{\sset - \grtrgen}}{\sqrt{m}} + D \cdot \frac{u}{\sqrt{m}} + D' \cdot \modelcovar(\grtrgen) \\
		\stackrel{\eqref{eq:proofs:hybridglobal:locglob}}{\leq} {} & \frac{D^2}{2} \cdot \frac{2}{D \cdot \subgparam} \cdot \frac{\sqrt{\meanwidth{\sset}} \cdot m^{1/4}}{\sqrt{m}} + D \cdot \frac{u}{\sqrt{m}} + D' \cdot \modelcovar(\grtrgen) \\
		={} & D \cdot \underbrace{\subgparam^{-1}}_{\stackrel{\eqref{eq:proofs:abstractrecovery:subgparamlower}}{\leq} 2 \subgparam} \cdot \Big(\frac{\meanwidth{\sset}}{\sqrt{m}}\Big)^{1/2} + D \cdot \frac{u}{\sqrt{m}} + D' \cdot \modelcovar(\grtrgen) \\
		\lesssim {} & D \cdot \Big[ \subgparam \cdot \Big(\frac{\meanwidth{\sset}}{\sqrt{m}}\Big)^{1/2} + \frac{u}{\sqrt{m}}\Big] + D' \cdot \modelcovar(\grtrgen)  = t.
	\end{align}
	Thus, if $D' > 0$ is large enough and $D = \tilde{C} \cdot \max\{1, \subgparam \cdot \modeldev(\grtrgen) \}$ for a sufficiently large numerical constant $\tilde{C} > 0$, we can conclude that condition \eqref{eq:app:abstractrecovery:localmw:implbound} is also satisfied. Therefore, under the assumptions of Theorem~\ref{thm:extensions:hybridglobal}, the statement of Theorem~\ref{thm:app:abstractrecovery:localmw} yields the following error bound:
	\begin{align}
		\Big(\sum_{k = 1}^N\lnorm{\solu^j - \hybscalfac_j\grtr}^2 \Big)^{1/2} &= \lnorm{[\solu^1 \cdots \solu^N] - \grtrgen} \\*
		&\leq t \lesssim \max\{1, \subgparam \cdot \modeldev(\grtrgen) \} \cdot \Big( \underbrace{\subgparam \cdot \Big(\frac{\meanwidth{\sset}}{\sqrt{m}}\Big)^{1/2}}_{\stackrel{\eqref{eq:extensions:hybridglobal:meas}}{\lesssim} \delta} + \frac{u}{\sqrt{m}} \Big) + \modelcovar(\grtrgen)  \\
		&\stackrel{\mathllap{u \coloneqq \delta \sqrt{m}}}{\lesssim} \max\{1, \subgparam \cdot \modeldev(\grtrgen) \} \cdot \delta + \modelcovar(\grtrgen)
	\end{align}
	with probability at least $1 - 2 \exp(-C \cdot \delta^2 \cdot m) - 2 \exp(-C \cdot m) - \exp(-C \cdot \subgparam^{-4} \cdot m)$.
	
	Finally, using that $\modelcovar(\grtr\hybscalvec^\T) = \tfrac{N\sqrt{N}}{M} \cdot  \modelcovar_\text{Hyb}$ (see Proposition~\ref{prop:proofs:modelcovar:hybrid}) and that $\modeldev(\grtr\hybscalvec^\T)^2 \lesssim N \cdot \modeldev_\text{Hyb}^2 + N \cdot \tfrac{\nu^2}{M}$ (by \eqref{eq:proofs:hybrid:modeldev}), we obtain
	\begin{align}
		\Big(\tfrac{1}{N}\sum_{k = 1}^N\lnorm{\solu^j - \hybscalfac_j\grtr}^2 \Big)^{1/2} &\lesssim \tfrac{1}{\sqrt{N}} \cdot \max\Big\{1, \subgparam \cdot \modeldev(\grtr\hybscalvec^\T) \Big\} \cdot \delta + \tfrac{1}{\sqrt{N}} \cdot \modelcovar(\grtr\hybscalvec^\T)) \\
		&\lesssim \max\{\tfrac{1}{\sqrt{N}}, \subgparam \cdot (\modeldev_\text{Hyb}^2 + \tfrac{\nu^2}{M})^{1/2} \} \cdot \delta + \tfrac{N}{M} \cdot  \modelcovar_\text{Hyb}.
	\end{align}
	Since $\delta \leq 1$ and $\subgparam \geq  1 / \sqrt{2}$ (by \eqref{eq:proofs:abstractrecovery:subgparamlower}), the above probability of success can be bounded from below by $1 - 5 \exp(- C \cdot \subgparam^{-4} \cdot \delta^2 \cdot m )$ for an appropriately chosen $C > 0$.
\end{proof}




\appendix
\section{Appendix}

\subsection{Proof of Theorem~\ref{thm:proofs:abstractrecovery} (Recovery via the $\sset$-Lasso)}
\label{subsec:app:abstractrecovery}

We shall first prove the following slightly more general version of
Theorem~\ref{thm:proofs:abstractrecovery}.  It involves the concept of \emph{local mean width} introduced in
Definition~\ref{def:proofs:meanwidth}, which can be very helpful if the conic mean width $\meanwidth{\cone{\sset}{\grtrgen}}$ in \eqref{eq:proofs:abstractrecovery:meas} behaves ``inappropriately'' (see also Subsection~\ref{subsec:extensions:hybridglobal}).
\begin{theorem}\label{thm:app:abstractrecovery:localmw}
  We assume that Model~\ref{model:proofs:generalobs} holds true.  Let
  $\sset \subset \R^d$ be a convex subset and let $\grtrgen \in \sset$
  be an arbitrary vector.  There exist numerical constants
  $C, C', C'' > 0$ such that for every $u > 0$ and $t > 0$ the
  following holds true with probability at least
  $1 - 2 \exp(-C \cdot u^2) - 2 \exp(-C \cdot m) - \exp(-C \cdot
  \subgparam^{-4} \cdot m)$: If
  \begin{equation}\label{eq:app:abstractrecovery:localmw:meas}
    m \geq C' \cdot \subgparam^4 \cdot \Big(\tfrac{1}{t}\meanwidth[t]{\sset-\grtrgen}\Big)^2,
  \end{equation}
  and
  \begin{equation}\label{eq:app:abstractrecovery:localmw:implbound}
    t > C'' \cdot \Big( \subgparam \cdot \modeldev(\grtrgen) \cdot \frac{\tfrac{1}{t}\meanwidth[t]{\sset - \grtrgen} + u}{\sqrt{m}} + \modelcovar(\grtrgen)\Big),
  \end{equation}
  then any minimizer $\solu$ of \eqref{eq:proofs:klasso} satisfies
  $\lnorm{\solu - \grtrgen} \leq t$.
\end{theorem}

The proof of Theorem~\ref{thm:app:abstractrecovery:localmw} is based
on two fundamental concentration results on empirical sub-Gaussian
processes, originating from \cite{liaw2016randommat} and
\cite{mendelson2016multiplier}, respectively:

\begin{theorem}\label{thm:app:abstractrecovery:empproc}
  Assume that Model~\ref{model:proofs:generalobs} is satisfied and let $\ssetalt \subset t \S^{d-1}$ for $t > 0$.
  \begin{thmlist}
  \item (See \cite[Thm.~1.3]{liaw2016randommat}) There exists a
    numerical constant $C_1 > 0$ such that for every $u \geq 0$ the
    following holds true with probability at least $1 - \exp(-u^2)$:
    \begin{equation}\label{eq:app:abstractrecovery:empproc:quadratic}
      \sup_{\h \in \ssetalt} \abs[\Big]{\Big(\tfrac{1}{m}\sum_{i = 1}^m \abs{\sp{\a_i}{\h}}^2\Big)^{1/2} - t } \leq C_1 \cdot \subgparam^2 \cdot \frac{\meanwidth{\ssetalt} + u \cdot t}{\sqrt{m}}.
    \end{equation}
  \item\label{thm:app:abstractrecovery:empproc:multiplier} (See \cite[Thm. 4.4]{mendelson2016multiplier}) Fix
    $\grtrgen \in \R^d$ and set
    $z_i(\grtrgen) \coloneqq \sp{\a_i}{\grtrgen} - y_i$ for
    $i = 1, \dots, m$.  There exist numerical constants $C, C_2 > 0$
    such that for every $u > 0$ the following holds true with probability
    at least $1 - 2 \exp(-C \cdot u^2 ) - 2 \exp(-C \cdot m)$:
    \begin{equation}\label{eq:app:abstractrecovery:empproc:multiplier}
      \sup_{\h \in \ssetalt} \abs[\Big]{\tfrac{1}{m}\sum_{i = 1}^m \Big(z_i(\grtrgen) \sp{\a_i}{\h} - \mean{}_{(\a_i,y_i)}[z_i(\grtrgen) \sp{\a_i}{\h}]\Big)} \leq C_2 \cdot \subgparam \cdot \modeldev(\grtrgen) \cdot  \frac{\meanwidth{\ssetalt} + u \cdot t}{\sqrt{m}}.
    \end{equation}
  \end{thmlist}
\end{theorem}

\begin{proof}[Proof of Theorem~\ref{thm:app:abstractrecovery:localmw}]
  Let us denote the objective function of \eqref{eq:proofs:klasso} --- typically referred to as the \emph{empirical loss} --- by
  \begin{equation}
    \lossemp{}(\x) \coloneqq \tfrac{1}{2m} \sum_{i = 1}^m (y_i - \sp{\agen_i}{\x})^2, \quad \x \in \R^d.
  \end{equation}
  A straightforward computation shows that the excess loss takes the following form:
  \begin{align}
    \lossemp{}(\x) - \lossemp{}(\grtrgen) = \underbrace{\tfrac{1}{m} \sum_{i = 1}^m z_i(\grtrgen) \sp{\a_i}{\x - \grtrgen}}_{\eqqcolon \multiplterm{\lossemp}{\x,\grtrgen}} + \underbrace{\tfrac{1}{2m} \sum_{i = 1}^m \abs{\sp{\a_i}{\x - \grtrgen}}^2}_{\eqqcolon \quadrterm{\lossemp}{\x,\grtrgen}}, \label{eq:app:abstractrecovery:localmw:decomp}
  \end{align}
  where $\x \in \R^d$ and $z_i(\grtrgen) \coloneqq \sp{\a_i}{\grtrgen} - y_i$ for $i = 1, \dots, m$.
	
  Now, we would like to use
  Theorem~\ref{thm:app:abstractrecovery:empproc} to bound both the
  \emph{multiplier process} $\multiplterm{\lossemp}{\x,\grtrgen}$
  and the \emph{quadratic process}
  $\quadrterm{\lossemp}{\x,\grtrgen}$ uniformly on a ``small''
  neighborhood of $\grtrgen$.  For this purpose, let us assume that
  the events of Theorem~\ref{thm:app:abstractrecovery:empproc} with
  $\ssetalt \coloneqq (\sset - \grtrgen) \intersec t \S^{d-1}$ have indeed
  occurred. For every
  $\x \in \ssetalt + \grtrgen = \sset \intersec (t \S^{d-1} + \grtrgen)$, the bound of
  \eqref{eq:app:abstractrecovery:empproc:quadratic} yields
  \begin{align}
    \sqrt{2} \cdot \quadrterm{\lossemp}{\x,\grtrgen}^{1/2} &= \Big( \tfrac{1}{m} \sum_{i = 1}^m \abs{\sp{\a_i}{\x - \grtrgen}}^2 \Big)^{1/2} \geq t - C_1 \cdot \subgparam^2 \cdot \frac{\meanwidth{\ssetalt} + u \cdot t}{\sqrt{m}} \\
                                            &= t \cdot \Big( 1 - C_1 \cdot \underbrace{\frac{\subgparam^2 \cdot \tfrac{1}{t}\meanwidth[t]{\sset- \grtrgen}}{\sqrt{m}}}_{\stackrel{\eqref{eq:app:abstractrecovery:localmw:meas}}{\leq} 1 / \sqrt{C'}} {}-{} C_1 \cdot \frac{\subgparam^2 \cdot u}{\sqrt{m}} \Big) \\
                                            &\stackrel{\mathllap{u\coloneqq \sqrt{C \subgparam^{-4} m}}}{\geq} t \cdot \underbrace{\Big( 1 - C_1 / \sqrt{C'} - C_1 \sqrt{C} \Big)}_{\eqqcolon \sqrt{2C_0}}.\label{eq:app:abstractrecovery:localmw:quadraticbound}
  \end{align}
  Adjusting the (numerical) constants $C, C'$, we can always achieve
  that $C_0 > 0$, and therefore, we have
  \begin{equation}
    \quadrterm{\lossemp}{\x,\grtrgen} \geq C_0 \cdot t^2
  \end{equation}
  for all $\x \in \sset \intersec (t \S^{d-1} + \grtrgen)$ with
  probability at least $1 - \exp(-C \cdot \subgparam^{-4} \cdot m)$.
  
  The multiplier process can be handled in a very similar way: For
  every $\x \in \sset \intersec (t \S^{d-1} + \grtrgen)$, the bound of
  \eqref{eq:app:abstractrecovery:empproc:multiplier} yields
  \begin{align}
    \multiplterm{\lossemp}{\x,\grtrgen} &= \tfrac{1}{m} \sum_{i = 1}^m z_i(\grtrgen) \sp{\a_i}{\x - \grtrgen} \\
                                        &\geq \tfrac{1}{m} \sum_{i = 1}^m \mean{}_{(\a_i,y_i)}[z_i(\grtrgen) \sp{\a_i}{\x - \grtrgen}] - C_2 \cdot \subgparam \cdot \modeldev(\grtrgen) \cdot  \frac{\meanwidth{\ssetalt} + u \cdot t}{\sqrt{m}} \\
                                        &= \mean{}_{(\a,y)}[(\sp{\agen}{\grtrgen} - y) \sp{\a}{\x - \grtrgen}] - C_2 \cdot \subgparam \cdot \modeldev(\grtrgen) \cdot  \frac{\meanwidth[t]{\sset- \grtrgen} + u \cdot t}{\sqrt{m}} \\
                                        &= t \cdot \Big( \mean{}_{(\a,y)}[(\sp{\agen}{\grtrgen} - y) \sp{\a}{\tfrac{\x - \grtrgen}{t}}] - C_2 \cdot \subgparam \cdot \modeldev(\grtrgen) \cdot  \frac{\tfrac{1}{t}\meanwidth[t]{\sset- \grtrgen} + u}{\sqrt{m}} \Big) \\
                                        &\stackrel{(\ast)}{\geq} - t \cdot \underbrace{\Big( \modelcovar(\grtrgen; \agen, y) + C_2 \cdot \subgparam \cdot \modeldev(\grtrgen) \cdot  \frac{\tfrac{1}{t}\meanwidth[t]{\sset- \grtrgen} + u}{\sqrt{m}} \Big)}_{\eqqcolon t_0},
  \end{align}
  where $(\ast)$ follows from the Cauchy-Schwarz inequality:
  \begin{align}
    \mean{}_{(\a,y)}[(\sp{\agen}{\grtrgen} - y) \sp{\a}{\tfrac{\x - \grtrgen}{t}}] &= \sp{\mean{}_{(\a,y)}[(\sp{\agen}{\grtrgen} - y) \a]}{\tfrac{\x - \grtrgen}{t}} \\
                                                                                   &\geq - \lnorm{\mean{}_{(\a,y)}[(\sp{\agen}{\grtrgen} - y) \a]} \cdot \underbrace{\lnorm{\tfrac{\x - \grtrgen}{t}}}_{= 1} = - \modelcovar(\grtrgen; \agen, y).
  \end{align}
  
  Now, using the condition
  \eqref{eq:app:abstractrecovery:localmw:implbound} for an appropriate
  $C'' > 0$, we end up with
  \begin{equation}
    \score(\x) \coloneqq \lossemp{}(\x) - \lossemp{}(\grtrgen) = \quadrterm{\lossemp}{\x,\grtrgen} + \multiplterm{\lossemp}{\x,\grtrgen}  \geq C_0 t^2 - t_0 t = t \underbrace{(C_0 t - t_0)}_{\stackrel{\eqref{eq:app:abstractrecovery:localmw:implbound}}{>} 0} > 0
  \end{equation}
  for all $\x \in \sset \intersec (t \S^{d-1} + \grtrgen)$ with
  probability at least
  $1 - 2 \exp(-C \cdot u^2) - 2 \exp(-C \cdot m) - \exp(-C \cdot
  \subgparam^{-4} \cdot
  m)$. 
  On the other hand, any minimizer $\solu \in \sset$ of
  \eqref{eq:proofs:klasso} clearly satisfies $\score(\solu) \leq 0$. Hence, if
  we would have $\lnorm{\solu - \grtrgen} > t$, there would exist (by
  convexity of $\sset$) an
  $\x \in \sset \intersec (t \S^{d-1} + \grtrgen)$ such that
  $\solu \in \{ \grtrgen + \lambda (\x - \grtrgen) \suchthat \lambda > 1 \}$. But
  this already contradicts the fact that
  $\lambda \mapsto \score(\grtrgen + \lambda(\x - \grtrgen))$ is a convex
  function, since it holds that $\score(\solu) \leq 0$,
  $\score(\x) > 0$, and $\score(\grtrgen) = 0$. Consequently, we have
  $\lnorm{\solu - \grtrgen} \leq t$, which proves the claim.
\end{proof}
\begin{remark}\label{rmk:app:abstractrecovery:techniques}
	The above proof strategy is closely related to the statistical learning framework of Mendelson \cite{mendelson2014learning,mendelson2014learninggeneral}.
	It is based on the simple idea of decomposing the excess loss into a linear (multiplier) and a quadratic part, such as we did in \eqref{eq:app:abstractrecovery:localmw:decomp}. Using sophisticated tools from statistical learning and concentration of measure, one may now show that the quadratic term dominates the linear term except for a ``small'' neighborhood of $\grtrgen$. Hence, the excess loss becomes positive outside of this region, excluding all those vectors as potential minimizers.
	Such a localization argument is in fact widely used in estimation theory and dates back to classical works in geometric functional analysis and statistical learning (e.g., see \cite{milman1986banach,pajor1986gelfand,mendelson2002improving,bartlett2005local,mendelson2007subgaussian}).
	
	Apart from that, let us emphasize that the above analysis substantially improves the related approaches of \cite{plan2015lasso,genzel2016estimation}. These works do only focus on \emph{Gaussian} single-index models and handle the multiplier term by a more naive argument based on Markov's inequality (e.g., see \cite[Proof~of~Thm.~1.4]{plan2015lasso}), which eventually leads to a very pessimistic probability of success. Applying the sophisticated chaining-based result of Theorem~\ref{thm:app:abstractrecovery:empproc}\ref{thm:app:abstractrecovery:empproc:multiplier}, we are able to tackle these issues, while allowing for arbitrary observation rules and \emph{sub-Gaussian} measurements.
\end{remark}

The statement of Theorem~\ref{thm:proofs:abstractrecovery} is now a direct consequence of Theorem~\ref{thm:app:abstractrecovery:localmw}:
\begin{proof}[Proof of Theorem~\ref{thm:proofs:abstractrecovery}]
  Let us apply Theorem~\ref{thm:app:abstractrecovery:localmw} with
  \begin{equation}\label{eq:proofs:abstractrecovery:desiredaccuracy}
    t \coloneqq 2 C'' \cdot \Big( \subgparam \cdot \modeldev(\grtrgen) \cdot \frac{\meanwidth[1]{\cone{\sset}{\grtrgen}} + u}{\sqrt{m}} + \modelcovar(\grtrgen)\Big).
  \end{equation}
  Thus, we need to show that the conditions \eqref{eq:app:abstractrecovery:localmw:meas} and \eqref{eq:app:abstractrecovery:localmw:implbound} in Theorem~\ref{thm:app:abstractrecovery:localmw} are indeed satisfied for this specific choice of $t$.
  To see this, we first assume that $t > 0$ and observe that
  \begin{align}
    \tfrac{1}{t}\meanwidth[t]{\sset- \grtrgen} &= \tfrac{1}{t}\meanwidth[t]{(\sset- \grtrgen) \intersec t \S^{d-1}} \\
                                               &= \meanwidth{\tfrac{1}{t}(\sset- \grtrgen) \intersec \S^{d-1}} \\
                                               &\leq \meanwidth{\cone{\sset}{\grtrgen} \intersec \S^{d-1}} = \meanwidth[1]{\cone{\sset}{\grtrgen}}.
  \end{align}	
  Hence
  \begin{align}
    t &\geq 2 C'' \cdot \Big( \subgparam \cdot \modeldev(\grtrgen) \cdot \frac{\tfrac{1}{t}\meanwidth[t]{\sset - \grtrgen} + u}{\sqrt{m}} + \modelcovar(\grtrgen)\Big) \\
    &> C'' \cdot \Big( \subgparam \cdot \modeldev(\grtrgen) \cdot \frac{\tfrac{1}{t}\meanwidth[t]{\sset - \grtrgen} + u}{\sqrt{m}} + \modelcovar(\grtrgen)\Big),
  \end{align}
  and together with the assumption of \eqref{eq:proofs:abstractrecovery:meas} and
  $\delta \leq 1$, we conclude that \eqref{eq:app:abstractrecovery:localmw:meas} and \eqref{eq:app:abstractrecovery:localmw:implbound} are both fulfilled.
  Consequently, we obtain the error bound
  \begin{equation}
    \lnorm{\solu - \grtrgen} \leq t = 2 C'' \cdot \Big( \subgparam \cdot \modeldev(\grtrgen) \cdot \frac{\meanwidth[1]{\cone{\sset}{\grtrgen}} + u}{\sqrt{m}} + \modelcovar(\grtrgen)\Big)
  \end{equation}
  with probability at least $1 - 2 \exp(-C \cdot u^2) - 2 \exp(-C \cdot m) - \exp(-C \cdot \subgparam^{-4} \cdot m)$.
  
  In order to derive the actual statement of Theorem~\ref{thm:proofs:abstractrecovery}, we just set $u \coloneqq \subgparam^{-2} \cdot \delta \cdot \sqrt{m}$ and observe that
  \begin{equation}
    \subgparam \cdot \frac{\meanwidth[1]{\cone{\sset}{\grtrgen}}}{\sqrt{m}} \stackrel{\eqref{eq:proofs:abstractrecovery:meas}}{\leq} \tfrac{1}{\sqrt{C'}} \cdot \subgparam^{-1} \cdot \delta.
  \end{equation}
  The probability of success results from adjusting the constant $C$ and the basic fact that 
  \begin{align}
    \subgparam \geq \normsubg{\a} = \sup_{\x \in \S^{n-1}} \normsubg{\sp{\a}{\x}} &\stackrel{\eqref{eq:intro:notation:normsubg}}{\geq} \sup_{\x \in \S^{n-1}} 2^{-1/2} \mean[\abs{\sp{\a}{\x}}^2]^{1/2} \\ &\stackrel{\eqref{eq:intro:notation:intropic}}{=} 2^{-1/2} \sup_{\x \in \S^{n-1}} \lnorm{\x} = 2^{-1/2}.\label{eq:proofs:abstractrecovery:subgparamlower}
  \end{align}

	It remains to analyze the case of $\modeldev(\grtrgen) = \modelcovar(\grtrgen) = 0$, which is equivalent to $t = 0$ in \eqref{eq:proofs:abstractrecovery:desiredaccuracy}.
	Theorem~\ref{thm:app:abstractrecovery:localmw} is not applicable in this situation and we have to argue in a slightly different way.
	Indeed, if $\modeldev(\grtrgen) = \normsubg{\sp{\agen}{\grtrgen} - y} = 0$, one has $\sp{\agen}{\grtrgen} = y$ and $\multiplterm{\lossemp}{\x,\grtrgen} = 0$ for every $\x \in \R^d$. Repeating the argument of
  \eqref{eq:app:abstractrecovery:localmw:quadraticbound} for $t = 1$
  and $\ssetalt = \cone{\sset}{\grtrgen} \intersec \S^{d-1}$, we
  obtain
  \begin{equation}\label{eq:app:abstractrecovery:quadraticexact}
    \lossemp{}(\x) - \lossemp{}(\grtrgen) = \quadrterm{\lossemp}{\x,\grtrgen} = \tfrac{1}{2m} \sum_{i = 1}^m \abs{\sp{\a_i}{\x - \grtrgen}}^2 > 0
  \end{equation}
  for all
  $\x \in (\cone{\sset}{\grtrgen} \intersec \S^{d-1}) + \grtrgen$ with
  probability at least $1 - \exp(-C \cdot \subgparam^{-4} \cdot m)$.
  At the same time (event), if $\x \in \sset \setminus \{ \grtrgen \}$, we have
  $\tfrac{\x - \grtrgen}{\lnorm{\x - \grtrgen}} \in \cone{\sset}{\grtrgen}
  \intersec \S^{d-1}$ by the convexity of $\sset$.  Hence,
  \begin{equation}
    \lossemp{}(\x) - \lossemp{}(\grtrgen) = \quadrterm{\lossemp}{\x,\grtrgen} = \lnorm{\x - \grtrgen}^2 \cdot \tfrac{1}{2m} \sum_{i = 1}^m \abs{\sp{\a_i}{\tfrac{\x - \grtrgen}{\lnorm{\x - \grtrgen}}}}^2 \stackrel{\eqref{eq:app:abstractrecovery:quadraticexact}}{>} 0.
  \end{equation}
  This particularly implies that $\x$ cannot be minimizer of \eqref{eq:proofs:klasso}, and we can conclude that $\solu = \grtrgen$.
\end{proof}

\subsection{Proof of Proposition~\ref{prop:proofs:meanwidth} (Bounds on the Conic Mean Width)}
\label{subsec:app:conicmeanwidth}

We need the following upper bound on the conic mean width, which is a consequence of standard duality arguments (see also \cite[Sec.~4]{amelunxen2014edge} for more details):
\begin{lemma}[\protect{\cite[Prop.~3.6]{chandrasekaran2012geometry}}]\label{lem:app:conicmeanwidth:duality}
	Let $\norm{\cdot}$ be a norm on $\R^d$. For $\grtrgen \in \R^d \setminus \{ \vnull \}$, define the convex (descent) set $\sset \coloneqq \{\x \in \R^d \suchthat \norm{\x} \leq \norm{\grtrgen}\}$. Then, we have\footnote{Here, $\dist{\gaussian}{L} \coloneqq \inf_{\w \in L} \lnorm{\gaussian - \w}$ for a subset $L \subset \R^d$.}
	\begin{equation}
		\meanwidth[1]{\cone{\sset}{\grtrgen}}^2 \leq \mean[\inf_{\tau>0} \dist{\gaussian}{\tau\cdot\partial \norm{\grtrgen}}^2], \quad \gaussian \distributed \Normdistr{\vnull}{\I{d}},
	\end{equation}
	where $\partial \norm{\grtrgen} \subset \R^d$ denotes the \emph{subdifferential} of $\norm{\cdot}$ at $\grtrgen$:
	\begin{equation}\label{eq:app:conicmeanwidth:duality:subd}
		\partial\norm{\grtrgen} \coloneqq \{\w \in \R^d \suchthat \norm{\x} \geq \norm{\grtrgen}+ \sp{\w}{\x - \grtrgen} \text{ for all } \x \in \R^d \}.
	\end{equation}
\end{lemma}

\begin{proof}[Proof of Proposition~\ref{prop:proofs:meanwidth}]
	It suffices to prove part \ref{prop:proofs:meanwidth:lifting}, since the first one follows from the special case of $M = 1$.
	We would like to apply Lemma~\ref{lem:app:conicmeanwidth:duality} for $\grtrgen = \grtr \scalvec^\T \in \R^{n \times M}$ and $\norm{\cdot} = \lnorm{\cdot}[1,2]$, where the matrix space $\R^{n \times M}$ is canonically identified with $\R^{n \cdot M}$ ($d = n \cdot M$).
	Note that we can also assume without loss of generality that $\scalvec \neq \vnull$ because the claim would be trivial otherwise.
	In order to determine $\partial\norm{\grtrgen}$, we make use of the following identity that holds true for every norm (cf. \cite[Ex.~VI.3.1]{hiriart1993convex}):
	\begin{equation}
		\partial\norm{\grtrgen} = \{ \w \in \R^d \suchthat \norm{\w}^{\circ} \leq 1 \text{ and } \sp{\w}{\grtrgen} = \norm{\grtrgen} \},
	\end{equation}
	where $\norm{\cdot}^{\circ}$ is the dual norm of $\norm{\cdot}$. For our specific situation, we have $\lnorm{\cdot}[1,2]^{\circ} = \lnorm{\cdot}[\infty,2]$ with
	\begin{equation}
		\lnorm{\W}[\infty,2] \coloneqq \max_{1 \leq k \leq n} \lnorm{\w_k},
	\end{equation}
	as $\w_1, \dots, \w_n \in \R^M$ denote the rows of $\W \in \R^{n \times M}$.
	Now, let $\suppx \subset [n]$ be the support of $\grtr = (x_{0,1},\dots, x_{0,n}) \in \R^n$.
	Then, $\W \in \partial\lnorm{\grtr\scalvec^\T}[1,2]$ if and only if
	\begin{equation}\label{eq:app:meanwidth:dualcondition}
		\lnorm{\grtr\scalvec^\T}[1,2] = \sp{\W}{\grtr\scalvec^\T} = \sp{\W\scalvec}{\grtr} = \sum_{k \in \suppx} \sp{\w_k}{\scalvec} \cdot x_{0,k},
	\end{equation}
	and $\lnorm{\w_k} \leq 1$ for all $k \in [n]$.
	Since $\lnorm{\grtr\scalvec^\T}[1,2] = \lnorm{\grtr}[1] \cdot \lnorm{\scalvec} = \lnorm{\scalvec} \cdot \sum_{k \in \suppx} \abs{x_{0,k}}$ and $\lnorm{\w_k} \leq 1$, the condition \eqref{eq:app:meanwidth:dualcondition} can be only satisfied if
	\begin{equation}
		\w_k = \w_k' \coloneqq \sign(x_{0,k}) \cdot \frac{\scalvec}{\lnorm{\scalvec}} \in \S^{M-1} \quad \text{for all $k \in \suppx$.}
	\end{equation}
	Hence, we conclude that
	\begin{equation}
		\partial\lnorm{\grtr\scalvec^\T}[1,2] = \{ \W \in \R^{n \times M} \suchthat \w_k = \w_k' \text{ for $k \in \suppx$ and } \lnorm{\w_k} \leq 1 \text{ for $k \not\in \suppx$}  \}.
	\end{equation}
	
	With this, we achieve the following bound from Lemma~\ref{lem:app:conicmeanwidth:duality} and Jensen's inequality (here, $\gaussian_1, \dots, \gaussian_n \distributed \Normdistr{\vnull}{\I{M}}$ denote the rows of $\vec{G} \distributed \Normdistr{\vnull}{\I{n \cdot M}}$):
	\begin{align}
		\meanwidth[1]{\cone{\sset}{\grtr\scalvec^\T}}^2 &\leq \mean[\inf_{\tau>0} \dist{\vec{G}}{\tau\cdot\partial\lnorm{\grtr\scalvec^\T}[1,2]}^2] \\
		&\leq \inf_{\tau>0} \mean[ \dist{\vec{G}}{\tau\cdot\partial\lnorm{\grtr\scalvec^\T}[1,2]}^2] \\
		&= \inf_{\tau>0} \mean{} \Big[  \sum_{k \in \suppx} \lnorm{\gaussian_k - \tau \w_k'}^2 + \inf_{\substack{\lnorm{\w_k} \leq 1 \\ k \not\in \suppx}} \sum_{k \not\in \suppx} \lnorm{\gaussian_k - \tau \w_k}^2 \Big] \\
		&= \inf_{\tau>0} \Big(  \sum_{k \in \suppx} ( \underbrace{\mean[\lnorm{\gaussian_k}^2]}_{= M} + \tau^2 \underbrace{\lnorm{\w_k'}^2}_{= 1}) + \mean{} \Big[\sum_{k \not\in \suppx} \pospart{\lnorm{\gaussian_k} - \tau}^2 \Big] \Big) \\
		&= \inf_{\tau>0} \Big( s \cdot (M + \tau^2) + (n - s) \cdot \mean{}\big[ \pospart{\lnorm{\gaussian_1} - \tau}^2 \big] \Big), \label{eq:app:meanwidth:dualitybound}
	\end{align}
	where $\pospart{v} \coloneqq \max\{0, v\}$ for $v \in \R$ and $s = \cardinality{\suppx}$.
	Next, we would like to find an upper bound on the remaining expected value:
	\begin{equation}
		\mean{}\big[ \pospart{\lnorm{\gaussian_1} - \tau}^2 \big] = \int_{\tau}^\infty (t - \tau)^2 \cdot \prob[\lnorm{\gaussian_1} > t] dt.
	\end{equation}
	For this purpose, let us apply the following tail bound for $\chi^2$-distributions from \cite[Lem.~1]{laurent2000chisquare}:
	\begin{equation}
		\prob[\lnorm{\gaussian_1}^2 > M + 2\sqrt{M a} + 2 a] \leq \exp(- a), \quad a \geq 0.
	\end{equation}
	Setting $a = (t - \sqrt{2M})^2 / 2$ for $t \geq \sqrt{2M}$, we obtain
	\begin{align}
		\prob[\lnorm{\gaussian_1} > t] &= \prob[\lnorm{\gaussian_1}^2 > (\sqrt{2a} + \sqrt{2M})^2] \\
		&\leq \prob[\lnorm{\gaussian_1}^2 > M + 2\sqrt{M a} + 2 a] \\
		& \leq \exp(- a) = \exp(- \tfrac{(t - \sqrt{2M})^2}{2}).
	\end{align}
	Hence, as long as $\tau \geq \sqrt{2M}$, one has
	\begin{align}
		\mean{}\big[ \pospart{\lnorm{\gaussian_1} - \tau}^2 \big] &\leq \int_{\tau}^\infty (t - \tau)^2 \cdot \exp(- \tfrac{(t - \sqrt{2M})^2}{2}) dt \\
		&= \int_{0}^\infty t^2 \cdot \exp(- \tfrac{(t + \tau - \sqrt{2M})^2}{2}) dt \\
		&= \exp(- \tfrac{(\tau - \sqrt{2M})^2}{2}) \cdot \int_{0}^\infty t^2 \cdot \underbrace{\exp(- \tfrac{t^2 + 2 t (\tau - \sqrt{2M})}{2})}_{\leq \exp(-t^2 / 2)} dt \\
		&\leq \exp(- \tfrac{(\tau - \sqrt{2M})^2}{2}) \cdot \int_{0}^\infty t^2 \cdot \exp(- \tfrac{t^2}{2}) dt \lesssim \exp(- \tfrac{(\tau - \sqrt{2M})^2}{2}).
	\end{align}
	Finally, we continue in \eqref{eq:app:meanwidth:dualitybound} by fixing $\tau = \tilde{\tau} \coloneqq \sqrt{2M} + \sqrt{2 \log(2n / s)}$:
	\begin{align}
		\meanwidth[1]{\cone{\sset}{\grtr\scalvec^\T}}^2 &\leq s \cdot (M + \tilde{\tau}^2) + (n - s) \cdot \mean{}\big[ \pospart{\lnorm{\gaussian_1} - \tilde{\tau}}^2 \big] \\
		&\lesssim s \cdot (M + \tilde{\tau}^2) + (n - s) \cdot \exp(- \tfrac{(\tilde\tau - \sqrt{2M})^2}{2}) \\
		&= s \cdot M + s \cdot \underbrace{\Big(\sqrt{2M} + \sqrt{2 \log(2n / s)}\Big)^2}_{\lesssim \max\{M, \log(\tfrac{2n}{s})\}} + \underbrace{(n - s) \cdot \frac{s}{2n}}_{\leq s / 2} \\
		&\lesssim s \cdot \max\{M, \log(\tfrac{2n}{s})\},
	\end{align}
	which proves the claim.
\end{proof}

\section*{Acknowledgements}
The authors thank Axel Flinth and Gitta Kutyniok for fruitful discussions.
M.G. is supported by the DEDALE project, contract no. 665044, within the H2020 Framework Program of the European Commission.
P.J. is partially supported by DFG grant JU 2795/3.
\revision{Finally, the authors would like to thank the anonymous referees for their useful comments and suggestions which have helped to improve the original manuscript.}

\renewcommand*{\bibfont}{\small}
\begin{refcontext}[sorting=nyt]
	\printbibliography
\end{refcontext}

\newpage
\listoftodos

\end{document}